\documentclass[10pt]{amsart}
\usepackage{epsf,latexsym,amsmath,amssymb,amscd,datetime}
\usepackage{amsmath,amsthm,amssymb,enumerate,eucal,url,calligra,mathrsfs}

\usepackage{blkarray} 

\usepackage{tikz,scalerel}
\usetikzlibrary{arrows}

\usepackage{imakeidx}     

\usepackage[outdir=./]{epstopdf}

\usepackage{graphicx}
\usepackage{color}
\newenvironment{jfnote}{ \bgroup \color{blue} }{\egroup}

\newcommand{\red}{\color[rgb]{1.0,0.2,0.2}} 
\newcommand{\blue}{\color[rgb]{0.2,0.2,1.0}} 
\newcommand{\green}{\color[rgb]{0.0,0.5,0.0}} 

\newcommand{\oldStuff}[1]{}

\newcommand{\Coord}{{\rm Coord}}
 
\newcommand{\og}{{\scriptscriptstyle \le}}
\newcommand{\Bg}{{\scalebox{1.0}{$\!\scriptscriptstyle /\!B$}}}

\newcommand{\cert}{\xi}
 
\newcommand{\myindd}[1]{\index[defs]{#1}}	

\DeclareMathOperator{\SHom}{\mathscr{H}\text{\kern -3pt {\calligra\large om}}\,}

\IfFileExists{my_xrefs}{\input my_xrefs}{}

\DeclareMathOperator{\ViSu}{VisSub}

\newcommand{\naturals}{{\mathbb N}}

\newcommand{\Eor}{E^{\mathrm{or}}}
\newcommand{\mec}[1]{{\bf #1}}	
\newcommand{\bec}[1]{{\boldsymbol #1}}	

\DeclareMathOperator{\trace}{Trace}
\DeclareMathOperator{\Trace}{Trace}

\newcommand{\rhonew}{\rho^{\mathrm{new}}}
\newcommand{\specnew}{\Spec^{\mathrm{new}}}

\usepackage{mathrsfs}
\usepackage{amssymb}
\usepackage{dsfont}
\usepackage{verbatim}
\usepackage{url}

\newcommand{\Edir}{E^{\mathrm{dir}}}

\theoremstyle{plain}
\newtheorem{theorem}{Theorem}[section]
\newtheorem{lemma}[theorem]{Lemma}
\newtheorem{proposition}[theorem]{Proposition}
\newtheorem{corollary}[theorem]{Corollary}

\theoremstyle{definition}
\newtheorem{definition}[theorem]{Definition}

\newtheorem{xca}{Exercise}[section]

\newtheorem{example}[theorem]{Example}

\newcommand{\isom}{\simeq} 

\newcommand{\ignore}[1]{}

\newcommand{\reals}{{\mathbb R}}

\newcommand{\integers}{{\mathbb Z}}

\newcommand{\complex}{{\mathbb C}}

\newcommand\EE{\mathbb{E}}

\newcommand\II{\mathbb{I}}

\DeclareMathAlphabet{\mathcal}{OMS}{cmsy}{m}{n}

\newcommand\cA{\mathcal{A}}

\newcommand\cC{\mathcal{C}}

\newcommand\cE{\mathcal{E}}

\newcommand\cH{\mathcal{H}}
\newcommand\cI{\mathcal{I}}
\newcommand\cJ{\mathcal{J}}

\newcommand\cM{\mathcal{M}}

\newcommand\cP{\mathcal{P}}

\newcommand\cR{\mathcal{R}}
\newcommand\cS{\mathcal{S}}
\newcommand\cT{\mathcal{T}}

\DeclareMathOperator{\Prob}{Prob}

\DeclareMathOperator{\VLG}{VLG}

\DeclareMathOperator{\Line}{Line}

\DeclareMathOperator{\SNBC}{SNBC}

\DeclareMathOperator{\snbc}{snbc}

\def\from{\colon}

\def\isom{\simeq}

\def\eqdef{\overset{\text{def}}{=}}

\DeclareMathOperator{\id}{id}

\DeclareMathOperator{\ord}{ord}

\DeclareMathOperator{\Spec}{Spec}

\def\implies{\Rightarrow}

\DeclareRobustCommand
  \rddots{\mathinner{\mkern1mu\raise\p@
    \vbox{\kern7\p@\hbox{.}}\mkern2mu
    \raise4\p@\hbox{.}\mkern2mu\raise7\p@\hbox{.}\mkern1mu}}

\newcommand\xhookrightarrow[2][]{\ext@arrow 0062{\hookrightarrowfill@}{#1}{#2}}
\def\hookrightarrowfill@{\arrowfill@\lhook\relbar\rightarrow}

\usepackage{relsize}
\usepackage{tikz}
\usetikzlibrary{matrix,arrows,decorations.pathmorphing}
\usepackage{tikz-cd}
\usetikzlibrary{cd}

\usepackage[pdftex,colorlinks,linkcolor=blue]{hyperref}

\newcommand{\myDeleteNote}[1]{{}}

\begin{document}

\title[Relativized Alon Conjecture I] 
{On the Relativized Alon Second Eigenvalue
Conjecture I: Main Theorems, Examples, and Outline of Proof}

\author{Joel Friedman}
\address{Department of Computer Science, 
        University of British Columbia, Vancouver, BC\ \ V6T 1Z4, CANADA}
\curraddr{}
\email{{\tt jf@cs.ubc.ca}}
\thanks{Research supported in part by an NSERC grant.}

\author{David Kohler}
\address{Department of Mathematics, 
        University of British Columbia, Vancouver, BC\ \ V6T 1Z2, CANADA}
\curraddr{422 Richards St, Suite 170, Vancouver BC\ \  V6B 2Z4, CANADA}
\email{{David.kohler@a3.epfl.ch}}
\thanks{Research supported in part by an NSERC grant.}

%
\date{\today}

\subjclass[2010]{Primary 68R10}

\keywords{}

\begin{abstract}

This is the first in a series of six articles devoted to showing that a typical
covering map of large degree to a fixed, regular graph has its new adjacency
eigenvalues within the bound conjectured by Alon for random regular graphs.
Many of the techniques we develop hold whether or not the base graph is
regular.

Our first main theorem in this series of articles is that if the base graph
is $d$-regular, then for any $\epsilon>0$, as the degree, $n$, of the
covering map tends to infinity, some new adjacency eigenvalue of the map is
larger in absolute value that $2(d-1)^{1/2}+\epsilon$ with probability at
most order $1/n$.  Our second main theorem is that if, in addition, the
base graph is Ramanujan, then this probability is bounded above and below
by $1/n$ to the power of a positive integer that we call the {\em tangle
power} of the model, i.e., of the probability spaces of random covering
maps of degree $n$.

The tangle power is fairly easy to bound from below, and at times to
compute exactly; it measures the probability that certain {\em tangles}
appear in the random covering graph, where a {\em tangle} is a local event
that forces the covering graph to have a new eigenvalue strictly larger
than $2(d-1)^{1/2}$.

If the base graph has no {\em half-loops}, then our simplest model of a
random covering map is the model where one uniformly and independently
chooses a permutation for each edge of the base graph; a half-loop is,
roughly speaking, an unorientable self-loop.  More generally, our theorems
are valid for any model of random covering maps that is {\em algebraic},
which is a set of conditions that our trace methods require.  Our main
theorems are relativizations of Alon's conjecture on the second eigenvalue
of random regular graphs of large degree.

In this first article of the series, we introduce all the terminology
needed in this series, motivate this terminology, precisely state all the
results in the remaining articles, and make some remarks about their
proofs.  As such, this article provides an overview of the entire series of
articles; furthermore, the rest of the articles in this series may be read
independently of one another.

\end{abstract}

\maketitle
\setcounter{tocdepth}{3}
\tableofcontents

\newcommand{\sePrelimProofs}{17}

\section{Introduction}\label{se_intro}

The main goal of this series of six articles is to prove
a relativization of
Alon's Second Eigenvalue Conjecture,
formulated in \cite{friedman_relative},
for any {\em base graph}, $B$, that is regular;
a proof of this theorem appears in our 
preprint \cite{friedman_kohler}.
This series of six articles represents a
``factorization'' of the proof in \cite{friedman_kohler}
into many independent parts.
This includes some original work beyond \cite{friedman_kohler}, 
and serves 
to clarify underlying principles of the proof.
It also makes it 
easier to generalize the results here---for possible future
use to related questions.

This series of articles also represents some improvements over
\cite{friedman_alon}, which resolved Alon's original conjecture, in that
(1) one technical tool of \cite{friedman_alon}---the 
{\em selective traces}---is replaced with a much simpler tool
of {\em certified traces},
and (2) for certain values of $d$,
our results get improved
bounds on the probability estimates in the Alon's original conjecture
for $d$-regular graphs.
We also correct a minor error in \cite{friedman_alon} regarding
the model $\cH_{n,d}$ there, which we generalize here as the
{\em cyclic model}.

This particular article has two main goals:
first, to give an overall view of this series of articles,
and, second, to motivate and precisely state the 
terminology used and the main results
in the subsequent articles.
In this way, each of the subsequent can be read independently
of one another, assuming the terminology we define in this article;
each subsequent article summarizes the terminology it needs, and the reader
of subsequent articles may prefer to begin with its summary
and consult this article for motivation as needed.

In additional to our main goal, this article has a number of 
additional resources, such as:
(1) the optional Section~\ref{se_bsf}
that summarizes some of the definitions and methods of
\cite{broder,friedman_random_graphs}, to help the reader
better understand the definitions of 
Sections~\ref{se_ordered_B_strong_alg}--\ref{se_new_algebraic};
(2) aside from precisely stating the main theorems of subsequent articles,
we make additional remarks on them and/or their proofs;
(3) in one appendix we explain some of the ideas of
\cite{friedman_alon} in terms of {\em certified traces}, which is
an idea new to \cite{friedman_kohler} and this series of articles,
which replaces and significantly simplifies
the {\em selective traces} of \cite{friedman_alon};
(4) in another appendix we list all the terminology we use and make
comments on it (the reader can search this appendix for terms to see
where they are formally defined).

\subsection{Historical Context}

Recall that
Alon's Second Eigenvalue Conjecture
says that for
fixed integer $d\ge 3$, and a real $\epsilon>0$,
a random $d$-regular graph on $n$ vertices
has second adjacency eigenvalue at most $2(d-1)^{1/2}+\epsilon$ with
{\em high probability},
i.e., probability than tends to one as $n$ tends to infinity.
The interest in this conjecture is that the conclusion implies that
most graphs have, in a sense, almost optimal spectral properties,
which in turn implies a number of ``expansion'' or ``well connectedness''
properties of the graph.
The conjecture was established with weaker bounds---i.e., with
$2(d-1)^{1/2}$ replaced by a larger function of $d$---in
\cite{broder,friedman_kahn_szemeredi,friedman_random_graphs}, and finally
settled affirmatively in \cite{friedman_alon}.
All these papers bound not only the second eigenvalue with high
probability, but also give the same bound on the absolute value of the
all eigenvalues except the first, i.e., on the most negative eigenvalue.
The paper \cite{friedman_alon} obtained bounds on the probability
of having an eigenvalue, excepting the first, larger in absolute value than
$2(d-1)^{1/2}+\epsilon$,
where the upper and lower bounds match to within a factor of $n$, 
and for many values of
$d$ they match to within a constant
factor; in this series of articles
we determine bounds matching to within a constant
factor for all $d$, and, more generally, the analog for covering
maps to a fixed base graph that is regular and Ramanujan.

One generalization of the above spectral bounds for random $d$-regular
graphs involves the
notion of a relative expander, discussed 
in \cite{friedman_relative_boolean}\footnote{
This article
was circulated in a limited fashion, but was rejected for
publication on the basis of having no
interesting applications (at the time).
};
roughly speaking, for any covering map $G\to B$, we consider its
{\em new adjacency eigenvalues}, i.e., the eigenvalues of the
adjacency matrix of $G$ not arising from eigenfunctions
pulled back from the {\em base graph} $B$.
In the special case 
where $B$ has only one vertex, say of degree $d$, then
$G$ is a random $d$-regular graph,
and the {\em new eigenvalues} are 
all eigenvalues except the first,
namely $d$.
The article \cite{friedman_relative_boolean} was searching for distinguished
covering maps to (i.e., whose target is)
the Boolean cube; this article identified 
a unique degree two
covering map to the cube that was an ``optimal relative expander.''
The motivation for this search was the connection between
covering maps in graph theory and extension fields discussed
in \cite{friedman_geometric_aspects}, with the idea
that some covering maps of the Boolean cube may shed some light
on the complexity theory of Boolean functions.
Since then the study of covering maps has yielded new ways to
build expanders, including the remarkable works
\cite{bilu,mssI} (\cite{bilu} building on \cite{frieze_molloy})
regarding degree two covering maps, which
proves the existence of families of Ramanujan graphs of any given degree.

The relativized Alon conjecture, regarding random covering maps,
was formulated in \cite{friedman_relative}, inspired both by
\cite{nagnibeda} and by the success of Grothendieck's notion
of {\em relativization}.
Weaker forms of this conjecture were proven in
\cite{friedman_relative,linial_puder,lubetzky,a-b,puder}.
The conjecture for regular base graphs was established in
\cite{friedman_kohler}.
As in \cite{friedman_alon}, the high probability bound in
\cite{friedman_kohler} is of form
$1-O(n^{-\tau})$ with $\tau\ge 1$;
furthermore the largest possible value of $\tau$ can be determined for 
our {\em basic models} of random covering maps of a base graph that
is $d$-regular and Ramanujan.  In such cases
$$
\tau \ge 
\Bigl\lfloor \bigl( (d-1)^{1/2} - 1 \bigr)/2  \Bigr\rfloor +1 
$$
(where $\lfloor x\rfloor$ is the ``floor'' function,
denoting the largest integer whose value is
at most $x$)
and this bound is achieved for certain $d$-regular $B$.
We remark that the results in this series of articles therefore 
improve upon \cite{friedman_alon}
where the optimal value of $\tau$ was not determined for certain
models of random $d$-regular graphs for certain values of $d$.

There are a number of notable results related to ours.
Puder's \cite{puder} results prove a relative Alon conjecture
with new eigenvalue bounds
within a multiplicative factor of $3$ for any base graph $B$,
regular or not; his results also give a bound for $d$-regular
graphs that is close to $2(d-1)^{1/2}$, and his proof
is conceptually simple, although
requires the results of
\cite{puder_p}.
Recently 
Bordenav\'e \cite{bordenave} has given a proof
of the original Alon conjecture \cite{friedman_alon} as well
as the relativized Alon conjecture for regular base graphs,
which uses trace methods that avoids
tangles of order greater than zero, but that require more involved
probabilistic estimates.
More recently, Bordenav\'e and Collins \cite{bordenave_collins2019}
have proved a very general result about random permutation
matrices, that in particular
proves the full relativized Alon conjecture, i.e.,
for arbitrary base graphs.

Bilu and Linial \cite{bilu} point out that if the
base graph is ``approximately'' a disjoint unions of many
small graphs, then most degree two covers will be very poor
relative expanders.  Our main theorem, by contrast, shows that for any
fixed regular base graph, covers of large degree are, with high probability,
nearly relatively Ramanujan.

Our approach to the relativized Alon conjecture
follows the Broder-Shamir trace method of \cite{broder},
with its refinements of \cite{friedman_random_graphs,friedman_alon},
which we adapt to the more general situation
of random, degree $n$ covering maps of a
fixed graph, $B$.
However, the proofs here (and in \cite{friedman_kohler})
 significantly simplify some of the arguments
of \cite{friedman_alon}; perhaps the greatest simplifications are
(1) we replace the {\em selective trace} of \cite{friedman_alon}
by the much simpler
{\em certified trace} of this series of articles, and
(2) we give an improved ``Sidestepping Lemma'' that
is much easier to apply.
As mentioned before, this series of articles factors the
proof in \cite{friedman_kohler},
and some of the parts are written in greater generality,
for possible future use.

This series of articles has two main results: the first is that
the relativized Alon conjecture holds for regular base graphs.
The second is that one can determine, to within a constant factor,
the probability that a graph
does not satisfy the Alon bound, provided that the base graph 
is {\em Ramanujan}.
Curiously, this is analogous to the work
\cite{lubetzky}, where the new eigenvalue bounds
depend on the spectrum of the $d$-regular base graph, $B$, and degrade
when $B$ has eigenvalues close to---but less than---$d$.

\subsection{Organization of This Series of Articles}

This is the first in a series of six articles devoted to 
factoring the main theorems in \cite{friedman_kohler} and their proofs
into independent
parts;
\cite{friedman_kohler} contain some additional results not covered
in these six articles, especially
regarding ``mod-$S$ functions'' (Section~3.5 there).

For brevity, we refer to the articles in this series as
Article~I through Article~VI.

The individual articles have the following content.
\begin{description}
\item[Article I (this article)]
statement of the main theorems in this series of articles, proven in 
Articles~V and~VI,
and of the results from Articles~II-IV needed in Article~V;
definitions all terminology needed for these statements;
some supplemental material, including additional remarks on
Articles~II--VI; an optional Section~\ref{se_bsf} that explain
aspects of \cite{broder,friedman_random_graphs} and its connection
to some of our terminology;
an optional appendix that reviews some of the 
methods in Article~III based on \cite{friedman_alon} but introduces
{\em certified traces};
a second appendix that gathers the terminology in all the definitions
of this article.
\end{description}
Each subsequent article is independent of the others, although they
all require some of the definitions and notation of Article~I
(each subsequent article reviews those definitions and notation
it needs).
\begin{description}
\item[Article II] 
main theorems regarding asymptotic expansions in random coverings;
this is an adaptation of \cite{friedman_random_graphs} to our more
general situation, but we factor this material into a few independent
parts and we simplify some of the proofs in
\cite{friedman_random_graphs}.
\item[Article III]
main theorems on {\em certified traces}, using the results of Article~II.  
These ideas rely heavily
on \cite{friedman_alon}, although
the certified traces are new
to \cite{friedman_kohler}, which replace the significantly
more cumbersome idea of {\em selective traces} of 
\cite{friedman_alon}.
\item[Article IV] the {\em Sidestepping Theorem} we use in this series
of articles; this is a lemma in probability theory needed to
infer eigenvalue location based in the results in Article~III;
it is an strengthening of the Sidestepping Lemma in
\cite{friedman_alon}.
\item[Article V] 
our {\em basic models} are {\em algebraic};
conclusion of the proof of the relativized Alon conjecture for
regular base graphs (our first main result);
definition of {\em algebraic power} and {\em tangle power} of an
algebraic model;
for a fixed {\em tangle}, any covering graph of sufficiently high
degree containing the tangle has a non-Alon new eigenvalue;
more precise form of the first main result in terms of
{\em algebraic power} and {\em tangle power};
\item[Article VI]
our standard models are {\em pseudo-magnifying};
proof that for regular, Ramanujan graphs, the 
{\em algebraic power}\, equals $+\infty$; 
corollary that our bounds on the Alon conjecture probabilities
are tight to within a constant factor
when the base graph is Ramanujan (our second main result); 
estimates on the {\em tangle power} of a model.
\end{description}

\subsection{Organization of This Article}

The rest of this article is organized as follows.
The results of this series of articles
that will likely be of most interest in applications are stated
in Section~\ref{se_some_results} after some preliminary definitions
in Section~\ref{se_term1_basic}.
The main theorems in this series of articles
are stated in Section~\ref{se_main_theorems} after some
preliminary terminology in Section~\ref{se_term2_walks_traces}.

Section~\ref{se_bsf} is an optional section where we review
some ideas of \cite{broder,friedman_random_graphs} and give some examples
to motivate and more easily understand some
technical aspects of definitions that we give in 
Sections~\ref{se_ordered_B_strong_alg}--\ref{se_new_algebraic}.

In Section~\ref{se_ordered_B_strong_alg} we introduce the
notions of an {\em ordered graph}, a {\em $B$-graph}, and a
{\em strongly algebraic} model.
The permutation model is an example of a model that is strongly
algebraic.
In Section~\ref{se_new_homot} we discuss homotopy types and VLG's
(variable-length graphs).

At this point the reader has a choice.
Sections~\ref{se_ordered_B_strong_alg} and~\ref{se_new_homot}
are all that is needed to read
our summary of the results in Articles~II-VI, which are respectively
covered in Sections~\ref{se_art_expansion}--\ref{se_art_sharp},
and the reader may skip to there.
Section~\ref{se_new_algebraic} is
more technical, but (1) defines algebraic models (some of our basic models are
algebraic but not strongly algebraic),
and (2) gives the reader
more insight into Article~II (beyond Section~\ref{se_bsf}).

In Appendix~\ref{se_append_cert_ind} we describe a bit more regarding
the techniques used in Article~III (which are based
on those of \cite{friedman_alon}, but has some simplifications
such as {\em certified traces}).

Appendix~\ref{se_def_summary} collects all the definitions in
this article and makes additional remarks on the definitions that are
not standard.

\section{Basic Terminology}
\label{se_term1_basic}

In this section we introduce some preliminary terminology
needed to state the main results in this article.

\subsection{Basic Notation and Conventions}
\label{su_very_basic}

We use $\reals,\complex,\integers,\naturals$
to denote, respectively, the
the real numbers, the complex numbers, the integers, and positive
integers or
natural numbers;
we use $\integers_{\ge 0}$ to denote the set of non-negative
integers.
We denote $\{1,\ldots,n\}$ by $[n]$.

If $A$ is a set, we use $\naturals^A$ to denote the set of
maps $A \to \naturals$; we will refers to its elements as 
{\em vectors}, denoted in bold face letters, e.g., $\mec k\in A^\naturals$
or $\mec k\from A\to\naturals$; we denote its {\em component} 
in the regular face equivalents, i.e., $k(a)\in\naturals$ for
the $a$-component of $\mec k$.
As usual, $\naturals^n$ denotes $\naturals^{[n]}=\naturals^{\{1,\ldots,n\}}$.
We use similar conventions for $\naturals$ replaced by $\reals$,
$\complex$, etc.

If $A$ is a set, then $\# A$ denotes the cardinality of $A$.
We often denote a set with all capital letters, and its cardinality
in lower case letters; for example,
when we define
$\SNBC(G,k)$, we will write
$\snbc(G,k)=\#\SNBC(G,k)$.

If $A'\subset A$ are sets, then $\II_{A'}\from A\to\{0,1\}$ (with $A$
understood) denotes
the characteristic function of $A'$, i.e., $\II_{A'}(a)$ is $1$ if
$a\in A'$ and otherwise is $0$;
at times we write $\II_{A'}$ (with $A$ understood) when
$A'$ is not a subset of $A$, by which we mean
$\II_{A'\subset A}$.

All probability spaces $\cP=(\Omega,P)$ are finite, i.e., 
$\#\Omega< \infty$,
and all elements of $\Omega$ (which we also call {\em atoms})
have nonzero probability, i.e.,
$\omega\in\Omega$ implies that $P(\omega)>0$; an event is any subset of
$\Omega$.  We use $\cP$ and $\Omega$ interchangeably when $P$ is
understood and confusion is unlikely.
A {\em complex-valued random variable} on $\cP$ or $\Omega$ 
is a function $f\from\Omega\to\complex$, and similarly for real-, 
integer-, and natural-valued random variable; we use denote its 
$\cP$-expected value by
$$
\EE_{\omega\in\Omega}[f(\omega)]=\sum_{\omega\in\Omega}f(\omega)P(\omega).
$$
If $\Omega'\subset\Omega$ we denote the probability of $\Omega'$ by
$$
\Prob_{\omega\in\Omega}[\Omega']=\sum_{\omega\in\Omega'}P(\omega')
=
\EE_{\omega\in\Omega}[\II_{\Omega'}(\omega)].
$$
At times we write $\Prob_{\omega\in\Omega}[\Omega']$ where $\Omega'$ is
not a subset of $\Omega$, by which we mean
$\Prob_{\omega\in\Omega}[\Omega'\cap\Omega]$.

\subsection{Conventions Regarding Digraphs and Graphs}
\label{sb:digraphs_graphs}

\begin{definition}\label{de_digraph}
A {\em directed graph} is a tuple 
$B=(V_B,\Edir_B,t_B,h_B)$---or
more simply $B = (V,\Edir,t,h)$ when $B$ is clear---where
$V$ and $E$ are sets---the {\em vertex} and {\em (directed) edge} sets---and
$t$ and $h$ are maps $E\to V$---the 
{\em tail}
and {\em head} map;
we say that $e\in E$ is a {\em self-loop}
if $h(e)=t(e)$.
We also say that $e\in E$ {\em runs ({\rm or} is) from} $t(e)$ to $h(e)$.
\end{definition}

In particular, our directed graphs can have {\em multiple edges}---more
than one edge with the same tail and head---and self-loops.

\begin{definition}\label{de_graph}
A {\em undirected graph}, or simply a {\em graph},
is a tuple $G=(V,\Edir,t,h,\iota)$ where $(V,\Edir,t,h)$ is a
directed graph---the {\em underlying directed graph}---and
$\iota\from \Edir\to \Edir$
is an involution
\myindd{involution} (i.e., $\iota^2={\rm id}_{\Edir}$)---the
graph's {\em edge involution}\myindd{edge involution}---that
is orientation reversing (i.e., $h\iota=t$, and hence
$t\iota = h$);
we also refer to $\iota e$
{\em opposite edge\myindd{opposite edge} of $e$}.
We say that $e\in \Edir$ is a {\em half-loop}
if $h(e)=t(e)$ and
$\iota e=e$, and a {\em whole-loop}
if $h(e)=t(e)$ and $\iota e\ne e$.
We denote by $E_G$---the set of {\em edges} of $G$---the set of
orbits of $\iota$, i.e.,
all singleton sets $\{e\}$ where $e\in\Edir_G$ is a half-loop, and otherwise
all two element sets of the form $\{e,\iota e\}$ with $e\in\Edir_G$
not a half-loop.
By an {\em orientation} of an edge $\{e,\iota e\}$ in the graph, $G$,
we mean (the choice of) either $e$ or $\iota e$;
by an {\em orientation} of $G$ we mean a subset of $\iota_G$
orbit representatives
$\Eor_G\subset \Edir_G$, i.e., $\Eor_G$ contains all the half-loops and
one orientation for each two-element edge, $\{e,\iota e\}$.
\end{definition}

Whole-loops are the more standard type of self-loop one sees in graph
theory.
Half-loops are useful in a number of ways,
such as to give models of regular random graphs of odd degree
\cite{friedman_alon}.

\subsection{Coordinatized Coverings and Our Basic Models}

The following generalizes the models of random regular graphs used in
\cite{friedman_relative,friedman_alon}, based on 
\cite{broder,friedman_random_graphs}.

\begin{definition}\label{de_coordinatized_digraph}
Let $B$ be a digraph, and $n\ge 1$ an integer.  A digraph, $G$, is called
a {\em coordinatized cover (over $B$ and of degree $n$)} if $G$ satisfies
\begin{equation}\label{eq_coordinatized_digraph}
V_G = V_B \times [n],\quad
\Edir_G = \Edir_B \times [n],\quad
t_G(e,i)=(t_B e,i),\quad 
h_G(e,i) = \bigl(h_B e, \sigma(e) i\bigr)
\end{equation}
for some $\sigma\from \Edir_B\to\cS_n$.
Given $G$, the map $\sigma$ is uniquely determined and conversely;
we refer to $\sigma$ as {\em the permutation
map $\Edir_B\to\cS_n$ associated to $G$},
and to $G$ as the (coordinatized) cover associated to $B$.
\end{definition}

It is extremely useful---although rather pedantic---to remark that if $G$
is a coordinatized cover, then $B$ and $n$ are uniquely determined by $G$,
for the following set theoretic reasons: $V_G$ consists of pairs $(v,i)$,
and the set of all $v$ that appear as such form $V_B$; similarly the
set of all $i$ determine $n$; similarly $\Edir_G,t_G,h_G$ determine
$\Edir_B,t_B,h_B$.
Hence if $G$ is a ``coordinatized cover,'' then $B,n$ are uniquely determined
from $G$.

\begin{definition}\label{de_coordinatized_graph}
If $B$ is a graph and $n$ is an integer, then $G$ is a 
{\em coordinatized cover (over $B$ of degree $n$)}
if the underlying digraph of $G$ is a coordinatized cover over the
underlying digraph of $B$, and that
$$
\iota_G(e,i) = \bigl( \iota_B(e), \sigma(e) i \bigr).
$$
We use ${\rm Coord}_n(B)$ to denote the set of all coordinatized
covers of $B$ of degree $n$.
\end{definition}

If $\sigma\from\Edir_B\to\cS_n$ is the permutation map associated to a
coordinatized graph $G$ over $B$, then the above definitions imply that
\begin{equation}\label{eq_coordinatized_graph}
\sigma(\iota e_B)=\bigl(\sigma(e)\bigr)^{-1}, \quad
\iota_G(e,i) = \bigl( \iota_B(e), \sigma(e) i \bigr) \ ;
\end{equation}
conversely, any $\sigma\from\Edir_B\to\cS_n$ satisfying
\eqref{eq_coordinatized_graph} is a permutation map associated to a
(unique) coordinatized graph $G$ over $B$ of degree $n$.

We now turn to defining the models of random covering maps of interest to us.
It is useful to introduce some common properties of our models.

\begin{definition}\label{de_model_conventions}
Let $B$ be a graph.
A {\em model over $B$}
is a family of probability spaces $\{\cC_n(B)\}_{n\in N}$ indexed by a 
parameter
$n$ that ranges over some infinite subset, $N$, of $\naturals$, such that
the atoms of each $\cC_n(B)$ lie in ${\rm Coord}_n(B)$;
we say that the model is {\em edge-independent} if for any orientation
$\Eor_B$ the random variables $\{\sigma(e)\}_{e\in\Eor_B}$
of the associated
maps $\Edir_B\to\cS_n$ are independent.
Also we often write simply $\cC_n(B)$ or $\{\cC_n(B)\}$ for
$\{\cC_n(B)\}_{n\in N}$ if confusion is unlikely to occur.
\end{definition}
Each edge-independent models is therefore
described by specifying the distribution
of $\sigma(e)\in\cS_n$ for every edge $e\in\Edir_B$ and $n\in N$.

We now describe what we call {\em our basic models}; these models are the
ones that are most convenient for our methods.

\begin{definition}\label{de_models}
Let $B$ be a graph.
By {\em our basic models} we mean one of the models 
edge-independent models $\{\cC_n(B)\}_{n\in N}$ over $B$ of degrees in $N$:
\begin{enumerate}
\item The {\em permutation model} assumes $B$ is any graph without half-loops
and $N=\naturals$: for each $n$ and $e\in\Edir_B$, 
$\sigma(e)\in\cS_n$ is a uniformly chosen permutation.
\item The {\em permutation-involution of even degrees} is defined for any $B$ 
and for $N$ being the even naturals: this is the same as the permutation,
except that if $e$ is a half-loop, then $\sigma(e)$ is a 
uniformly chosen {\em perfect matching} on $[n]$, i.e., a map $\sigma\in\cS_n$
that has no fixed points and satisfies $\sigma^2=\id$.
\item The {\em permutation-involution of odd degrees} is defined the same,
except that if $e$ is a half-loop, then $\sigma(e)$ is a
uniformly chosen {\em near perfect matching} on $[n]$, by which we mean
a map $\sigma\in\cS_n$ 
with exactly one fixed point and with $\sigma^2=\id$.
\item
The {\em full cycle} model (or simply {\em cyclic} model)
is defined liked the permutation model
(so $B$ is assumed to have no half-loops),
except
that when $e$ is a whole-loop then $\sigma(e)$ is a uniform
permutation whose cyclic structure consists of a single cycle of length $n$.
\item 
The {\em full cycle-involution of even degree} and {\em odd degree} models
(or simply {\em cyclic-involution} of either degree) are
defined for arbitrary $B$ and either $n$ even or $n$ odd, is the
full cycle model with the distributions of $\sigma(e)$ for half-loops, $e$,
as in the permutation-involution.
\end{enumerate}
\end{definition}

\begin{definition}\label{de_in_out_degree}
Let $B=(V_B,\Edir_B,t_B,h_B)$ be a digraph.
The {\em adjacency matrix of $B$}, denoted $A_B$, is the square
matrix indexed on $V_B$ such that for $v,v'\in V_B$,
$(A_B)_{v,v'}$ is the number of edges from $v$ to $v'$.
The {\em indegree}
(respectively, {\em outdegree})
of a vertex, $v\in V_B$, is the number of edges
whose head (tail) is $v$.
We say that $B$ is {\em strongly $d$-regular}
if $V_B\ne \emptyset$ and the indegree and outdegree of each vertex
equals $d$.
If $B$ is a graph, then the {\em degree}
of a vertex, $v\in V_B$, denoted $\deg_B(v)$,
is its indegree in the underlying directed graph
(which equals its outdegree there);
we say that $B$ is $d$-regular if its underlying directed graph
is strongly $d$-regular, 
i.e., if $V_B\ne\emptyset$ and the degree of each vertex is $d$.
\end{definition}

\begin{definition}\label{de_bouquet}
If $a,b\ge 0$ are integers, we say that $B$ is a {\em bouquet of $a$
whole-loops and $b$ half-loops} if $\#V_B=1$ and $E_B$ consists of
$a$ whole-loops and $b$ half-loops.
\end{definition}

\begin{example}
If $B$ is a bouquet of whole-loops or half-loops, then the models
$\cC_n(B)$ in Definition~\ref{de_models} are models of a random $d$-regular
graph.
The papers \cite{broder,friedman_random_graphs} deal exclusively with the
permutation model in the
case where $B$ is a bouquet of $d/2$ whole-loops for some even
integer $d\ge 4$.
The paper \cite{friedman_alon} works with a number of other models,
including the permutation-involution model for the bouquet of $d$ half-loops.
\end{example}

The permutation-involution model is a natural generalization of the models
$\cI_{n,d}$ ($n$ even) and $\cJ_{n,d}$ ($n$ odd) in
\cite{friedman_alon}, and the full cycle-involution model of
$\cH_{n,d}$ there.  
In \cite{friedman_alon} we see:
(1) $\cH_{n,d}$ is interesting since it has a higher
probability of satisfying the Alon bound 
(compare the value of
$\tau_{\rm fund}$ in Theorems~1.1 and 1.2 of \cite{friedman_alon}), and
(2) $\cI_{n,d}$ ($n$ even) and
$\cJ_{n,d}$ ($n$ odd) need to be treated separately since they give
rise to different asymptotic expansions in the trace method.

\subsection{Morphisms, Covering and Etale Morphisms}

It is important to note that
coordinatized covers are really covering morphisms of graphs.
This latter (``coordinate free'') view is crucial to methods.

\begin{definition}\label{de_digraph_morphisms}
A {\em morphism
$\pi\from G\to B$ of directed graphs} is
a pair $\pi=(\pi_V,\pi_E)$ of maps, $\pi_V\from V_G\to V_B$ and
$\pi_E\from E_G\to E_B$ which respect the heads and tail maps
(i.e., $h_B \pi_E = \pi_V h_G$ and $t_B \pi_E = \pi_V t_G$);
we refer to the values of $\pi_V^{-1}$ as the {\em vertex fibres}
of $\pi$, and similarly with $\pi_E^{-1}$ as
{\em edge fibres};
furthermore we say that $\pi$ is a {\em covering map}
(respectively,
{\em \'etale map})
if for each $v\in V_G$, $\pi_E$
gives an isomorphism (respectively, injection)
of those edges in $G$ with head $v$ to those in
$B$ with head $\pi_V(v)$, and the same with ``head'' replaced with
``tail.''
A covering map $\pi\from G\to B$ is {\em of degree $n$}
if each
vertex fibre and each edge fibre is of size $n$.
\end{definition}
If $B$ is connected and $\pi\from G\to B$ is a covering map, then one easily
shows that $\pi$ is of degree $n$ for some $n$.

\begin{definition}\label{de_graph_morphisms}
A {\em morphism}, $\pi\from G\to B$, of graphs is a
morphism $\pi=(\pi_V,\pi_E)$ of the underlying directed graphs
which respects the edge involutions (i.e., $\iota_B\pi_E=\pi_E\iota_G$);
furthermore, $\pi$ is a 
{\em covering morphism} (respectively,
{\em \'etale morphism})
if this is true of
$\pi$ as a map of underlying directed graphs.
\end{definition}

Etale morphisms are important in studying
walks in graphs for the following reason: we will organize walks by
the subgraph they traverse; it will useful to notice that if a graph
$G'$ is a subgraph of a $G$ that admits a covering map $\pi\from G\to B$,
then $G'\to B$ is \'etale.
This statement has certain converses: for example, if $B$ has
no half-loops and $G'\to B$ is 
\'etale, then for sufficiently large $n$, $G'$ is a subgraph
of at least one element of $\Coord_n(B)$\footnote{
  This theorem is sometimes called {\em Hall's theorem}; see 
  \cite{stallings83}.
}.  More important to us is---for similar reasons---the 
prominence of \'etale graphs in
the definition of {\em strongly algebraic}
(see Definition~\ref{de_strongly_algebraic}).

\subsection{Adjacency Matrices and New/Old Spectrum}

The definitions in this subsection are first given for digraphs; the
corresponding notion for graphs reduces to the
digraph notion of the underlying digraph(s).


\begin{definition}\label{de_graph_adjacency}
The {\em adjacency matrix}, $A_B$, of a graph $B$ is that of $B$'s underlying
digraph (Definition~\ref{de_in_out_degree}).  
As such $A_B$ is real symmetric, and if $n=\#B$, we use
$$
\lambda_1(B)\ge \cdots \ge \lambda_n(B)
$$
to denote the $n$ eigenvalues $A_B$ (listed with multiplicities).
\end{definition}

\begin{definition}\label{de_new_spec}
Let $\pi\from G\to B$ be a covering map of digraphs.  
A {\em vertex-fibre} is any subset of $V_G$ of the form
$\pi_V^{-1}(v)$ with $v\in V_B$.
We say that a function $V_G\to\reals$ is an {\em old function (of $V_G$)}
if it is constant on each vertex-fibre, and a {\em new function (of $V_G$)}
if its sum on any vertex-fibre is zero.
[We easily see that the space of functions $V_G\to\reals$ decomposes
as a direct sum of old and new functions, and that $A_G$ leaves
each of these two spaces invariant.]
We define the {\em new spectrum of $A_G$}, denoted
$\specnew_B(A_G)$, 
to be spectrum (i.e., the multiset of eigenvalues counted with multiplicites)
of $A_G$ restricted to
the new functions of $V_G$; we define the {\em new spectral radius
of $A_G$}, denoted $\rhonew_B(A_G)$ 
to be the maximum absolute value of the elements
of $\specnew_B(A_G)$.
We define the {\em old spectrum} to be the multiset of eigenvalues of $A_G$ 
restricted
to the old functions, or, equivalently, the multiset of eigenvalues of
$A_B$.
\end{definition}
Since
$$
\sum_{\lambda\in\specnew_B(A_G)}\lambda^k
=
\Trace(A_G^k) - \Trace(A_B^k),
$$
we see that $\specnew_B(A_G)$ is
determined from $G$ and $B$ alone, and not the particular covering
map $\pi\from G\to B$.

\begin{definition}\label{de_new_spec_graphs}
If $G\to B$ is a covering map of graphs, then
all terms in
Definition~\ref{de_new_spec} are those of the covering map
of the underlying directed graphs.
\end{definition}

\section{Some Results in this Series of Articles}
\label{se_some_results}

The main results in this series of articles, stated in
Section~\ref{se_main_theorems},
require some 
definitions in Section~\ref{se_term2_walks_traces}.
However, at this point we can state some consequences of these
results that are simpler to state and likely
of most interest in applications.

\begin{definition}\label{de_nonAlon}
Let $\pi\from G\to B$ be a covering map of $d$-regular graphs.
For an $\epsilon>0$ we define the
{\em $\epsilon$-non-Alon multiplicity of $G$
relative to $B$} to be the 
integer
$$
{\rm NonAlon}_B(G;\epsilon) \eqdef
\# \bigl\{\lambda\in\specnew_B(A_G)\ \bigm|\
|\lambda|>
2\sqrt{d-1} +\epsilon\bigr\} ,
$$
where
the above $\lambda$ are counted with their multiplicity in
$\specnew_B(A_G)$.
\end{definition}

If $B$ is not $d$-regular, then one can similarly define this notion
by replacing $2\sqrt{d-1}$ above with
$\|A_{\widehat B}\|_2$, i.e.,
the $L^2$ norm of the adjacency operator
on the universal covering, $\widehat B$, of $B$;
this is due to the well-known fact that
$$
\|A_{\widehat B}\|_2 = 2 \sqrt{d-1}
$$
if $B$ is $d$-regular, since then $\widehat B$ is the (there is only
one, up to isomorphism) infinite
$d$-regular tree.

Alon was interested in $\lambda_2(G)$ for a random $d$-regular
graph, $G$, on $n$ vertices, with $n$ large, and not in $\lambda_n(G)$,
i.e., Alon was interested in the case where $B$ has one vertex, and
was interested 
in only the {\em positive} non-Alon eigenvalues defined
above.  However our trace methods, like
most trace methods, 
simultaneously bound the negative non-Alon eigenvalues as well.

\begin{theorem}\label{th_rel_Alon_regular_simple}
Let $\cC_n(B)$ be any of our basic models over a $d$-regular graph $B$.
Then
%
%
for any $\epsilon>0$ there is a constant $C=C(\epsilon)$ for which
\begin{equation}\label{eq_implies_relative_alon}
\Prob_{G\in\cC_n(B)}[ {\rm NonAlon}_B(G;\epsilon)>0 ]
\le  C(\epsilon)/n \ .
\end{equation} 
\end{theorem}

Of course, \eqref{eq_implies_relative_alon} implies
\begin{equation}\label{eq_rel_alon_prob_mild}
\Prob_{G\in\cC_n(B)}[{\rm NonAlon}_d(G;\epsilon)>0]  \to 0
\quad\mbox{as $n\to\infty$}
\end{equation}
conjectured (in parenthesis) in Section~1 of \cite{friedman_relative};
we call this
the {\em relativized Alon conjecture} in \cite{friedman_kohler}.

Our second main result gives a much sharper result for $B$ that are
$d$-regular and {\em Ramanujan}
\begin{definition}\label{de_Ramanujan}
We say that a $d$-regular graph $B$ is {\em Ramanujan} if all eigenvalues
of $A_B$ lie in
$$
\{d,-d\} \cup \Bigl[ -2\sqrt{d-1}, 2\sqrt{d-1} \Bigr] .
$$
\end{definition}

\begin{theorem}\label{th_consequences_of_second}
Let $\cC_n(B)$ be any of our basic models over a graph $B$ that is
$d$-regular Ramanujan.
Then there is an integer $\tau_{\rm tang}$ and a constant $C'>0$ such that
for sufficiently small $\epsilon>0$ we have
$$
C' /n^{\tau_{\rm tang}}
\le \Prob_{G\in\cC_n(B)}[ {\rm NonAlon}_B(G;\epsilon)>0 ],
$$
and for any $\epsilon>0$ there is a constant $C=C(\epsilon)$ for which
$$
\Prob_{G\in\cC_n(B)}[ {\rm NonAlon}_B(G;\epsilon)>0 ]
\le  C(\epsilon)/n^{\tau_{\rm tang}} \ .
$$
Furthermore
$$
\tau_{\rm tang} \ge 
\bigl( (d-1)^{1/2} - 1 \bigr)/2  +1 .
$$
\end{theorem}

We remark that the above theorem improves the results of
\cite{friedman_alon} for certain values of $d$: for example,
if $d$ is even and $B$ consists of $d/2$ whole-loops, then
upper and lower bounds in \cite{friedman_alon} for 
$$
\Prob_{G\in\cC_n(B)}[ {\rm NonAlon}_B(G;\epsilon)>0 ]
$$
differ by a factor of $n$ if $d=m^2+1$ for an odd integer $m$
(but for all other $d$ are tight within a constant factor).
Hence our results improve those in \cite{friedman_alon} in these
cases.
We also claim that this series of articles is easier to read---or at
least more ``factored''---than \cite{friedman_alon}.

We conjecture that
Theorem~\ref{th_consequences_of_second} holds for any $d$-regular graph, $B$,
but this conjecture requires us to estimate another integer,
$\tau_{\rm alg}$, of the model, which seems difficult to compute
directly.

We will define $\tau_{\rm tang}$ 
in terms of what we will call {\em tangles}; in the terminology
of Definition~4.3 of
\cite{friedman_alon}, $\tau_{\rm tang}$ is the smallest
{\em order} of a hypercritical tangle.  This integer was computed
in \cite{friedman_alon} in a number of cases where $B$ has one
vertex (in which case $\cC_n(B)$ is a model of a random $d$-regular
graph on $n$ vertices), and these results are:
\begin{enumerate}
\item
if $d$ is even and $B$ is a bouquet of $d/2$ self-loops, then
for the permutation model we have
\begin{equation}\label{eq_tau_tang_permutation}
\tau_{\rm tang} = 
\Bigl\lfloor \bigl( (d-1)^{1/2} - 1 \bigr)/2  \Bigr\rfloor +1
\end{equation} 
(where $\lfloor x\rfloor$ denotes the floor function, i.e., the
largest integer no greater than $x$),
and for the full-cycle model we have
\begin{equation}\label{eq_tau_tang_full_cycle}
\tau_{\rm tang} = 
\Bigl\lfloor (d-1)^{1/2} - 1  \Bigr\rfloor + 1
\end{equation} 
\item
if $d$ is a bouquet of $d$ half-loops, then 
for the even degree permutation involution or 
full-cycle involution model we have
$$
\tau_{\rm tang} =
\Bigl\lfloor (d-1)^{1/2} - 1  \Bigr\rfloor + 1
$$
and for the odd degree permutation involution or
full-cycle involution model we have
$$
\tau_{\rm tang} = 
\Bigl\lfloor (d-1)^{1/2} - 1  \Bigr\rfloor + 1
$$
\end{enumerate}
In Article~VI we shall prove the same bounds for slightly more
general $B$.
As noted in \cite{friedman_alon}, 
\eqref{eq_tau_tang_permutation} and 
\eqref{eq_tau_tang_full_cycle} show that $\tau_{\rm tang}$ for
the full cycle model is roughly twice as large as for the
permutation model; hence the model, for the same base graph
$B$, can make a significant difference in the probability
of having non-Alon eigenvalues.

\section{Definitions Regarding Walks and Expected Traces}
\label{se_term2_walks_traces}


In this section we give some background needed to for our trace methods,
and introduce some terminology necessary for our discussion of
asymptotic expansions of the $\cC_n(B)$-expected values of the 
the traces we use.

\subsection{Walks and Adjacency Matrices}

\begin{definition}\label{de_walks} 
Let $G$ be a digraph and $k\in\integers_{\ge 0}$.
A \emph{walk of
length \(k\) in \(G\) from $v_0$ to $v_k$}
(alternatively {\em originating in $v_0$} and {\em terminating in $v_k$})
is an alternating sequence of vertices and directed edges, 
\begin{equation}\label{eq_walk_def}
w=(v_0, e_1,
v_1, e_2, v_2, \ldots, e_k, v_k )
\end{equation}
for which \( t_G(e_i) = v_{i-1} \) and
\(h_G(e_i) = v_i \) for all $i\in [k]$.
We say that a walk is
{\em closed} if \( v_0 = v_k \).
A {\em walk} (respectively, {\em of length $k$}, and/or {\em closed})
in a graph 
is a walk (of length $k$, and/or closed)
its underlying directed graph.
\end{definition}
In the above definition, a walk of length $k=0$ is simply a vertex, and
a walk of length
$k\ge 1$ can be inferred from its sequence $e_1,\ldots,e_k$ of directed edges
(and the knowledge of $h_G$ and $t_G$).

It is a standard fact that the number of walks of length $k$ from $e$ to $e'$
equals the $e,e'$ entry of $A_G^k$, and hence
\begin{equation}\label{eq_adjacency_traces}
\Trace(A_G^k) = \# \{ \mbox{closed walks of length $k$ in $G$} \} \ .
\end{equation}

\subsection{SNBC Walks and the Hashimoto or Non-backtracking Matrix}

Graphs, unlike directed graphs, have a notion of {\em non-backtracking}
walk and numerous related notions crucial to our trace methods.

\begin{definition}\label{de_nonback_oriented}
Let $G$ be a graph, and $w$ a walk in $G$ given by
\eqref{eq_walk_def}. 
We say that $w$ is {\em non-backtracking}
if $\iota_G e_{i+1}\ne e_i$ for $i=1,\ldots,k-1$, and moreover
{\em strictly non-backtracking closed},
abbreviated {\em SNBC},
if in addition $\iota_G e_k\ne e_1$.  
The {\em oriented line graph of $G$},
denoted $\Line(G)$,
is the directed graph given by $V_{\Line(G)}=\Edir_G$, and
$\Edir_G\subset V_{\Line(G)}\times V_{\Line(G)}$ is the
subset of pairs $(e,e')$
such that
$(e,e')$ forms a non-backtracking
walk of length $2$ in $G$, and the tail and head of $(e,e')$ are,
respectively, $e$ and $e'$.
The {\em Hashimoto matrix}
of $G$ is the adjacency
matrix, $H_G$,
of its oriented line graph.
We use $\mu_1(G)$ to denote the the Perron-Frobenius
eigenvalue of $H_G$, and $\mu_2(G),\ldots,\mu_m(G)$ with
$m=\#\Edir_G$ (the cardinality of $\Edir_G$) to denote the other
eigenvalues of $H_G$, ordered arbitrarily.
\end{definition}

The term {\em Hashimoto matrix} is used in 
\cite{st1,st2,st3,terras_zeta} based on Hashimoto's work
\cite{hashimoto1}, which
connects to Ihara's \cite{ihara},
although Serre (\cite{serre_arbres}, page~5) first connected
Ihara's work to graph theory and 
Sunada \cite{sunada1} describes the Hashimoto matrix and the
connected to Ihara's Zeta function for regular graphs.
The Hashimoto matrix is also called the {\em non-backtracking matrix} 
elsewhere in the literature.

\begin{definition}\label{de_SNBC}
For any graph, $G$, and $k,r\in\naturals$, we use $\SNBC(G,k)$ to denote
the set of strictly non-backtracking walks of length $k$ in $G$; 
we use $\snbc(G,k)$ 
to denote the cardinality of
$\SNBC(G,k)$.
\end{definition}

We easily see that for $k\ge 1$, the map
$$
(v_0,e_1,\ldots,e_k,v_k) \mapsto (e_1,e_2,\ldots,e_k,e_1)
$$
gives a bijection from $\SNBC(G,k)$ to closed walks in ${\rm Line}(G)$
of length $k$ (where we omit the directed edges in ${\rm Line}(G)$
since there is at most one edge from one vertex in ${\rm Line}(G)$
to another).
Hence,
in view of \eqref{eq_adjacency_traces}, we have
\begin{equation}\label{eq_hashimoto_traces}
\Trace(H_G^k) = \snbc(G,k) .
\end{equation} 

The trace methods we use most directly work with the $G\in\cC_n(B)$
expected value of \eqref{eq_hashimoto_traces} rather than
\eqref{eq_adjacency_traces}.
For $d$-regular graphs, $G$, it is known that $\mu_1(G)=d-1$,
and that for every $\epsilon>0$ there
is a $\delta> 0$ such that
\begin{equation}\label{eq_adj_iff_hash}
\max_{i>1} \bigl| \lambda_i \bigr| \le 2\sqrt{d-1} + \epsilon
\quad\iff\quad
\max_{j>1} \bigl| \mu_j \bigr| \le \sqrt{d-1} + \delta .
\end{equation} 
Both of these facts 
follow from the following version of the {\em Ihara determinantal formula}
that we will prove in Article~V, which states that 
\begin{equation}\label{eq_Ihara_det}
\det(u I -H_G) = \det\bigl(u^2 I - u A_G + (D_G-I) \bigr)
(u+1)^{o_1(G)}
(u^2-1)^{o_2(G)-n} , 
\end{equation}
where $o_1(G)$ is the number of half-loops in $G$ and $o_2(G)$ is the
number of edges that are not half-loops.
Our proof is a straightforward generalization of proofs of this 
well-known result for graphs without half-loops (see \cite{terras_zeta}).
It follows that \eqref{eq_adj_iff_hash} holds if we replace
the $\lambda_i$ with $i>1$ with the new adjacency eigenvalues of $G$,
and the $\mu_j$ with $j>1$ with the new adjacency eigenvalues of the
directed graph $\Line(G)$.

In view of the above remarks and the fact that
\begin{equation}\label{eq_new_Hashimoto}
\sum_{\mu\in \specnew_B(H_G)} \mu^k = \Trace(H_G^k) - \Trace(H_B^k)
={\rm snbc}(G,k)-{\rm snbc}(B,k)  \ ,
\end{equation}
we focus on estimating
the $G\in\cC_n(B)$ expected value of $\snbc(G,k)$, whose
dominant term we expect to be
$\Trace(H_B^k)$ for $k$ small relative to $n$
(namely for $k=o(n^{1/2})$).

\subsection{The Visited Subgraph of a Walk}
\label{su_visited_subgraph}

Our trace methods will ultimately organize all walks by their
{\em homotopy type}, and more finely by their
{\em visited subgraph} which we now define.

\begin{definition}\label{de_visited_subgraph}
Let $w=(v_0,e_1,v_1,\ldots,e_k,v_k)$ be a walk in a digraph
(respectively, graph) $G$.
The {\em visited subgraph} of $w$, denoted $\ViSu_G(w)$,
is the smallest subdigraph (respectively, subgraph)
of $G$ containing all elements of $w$.
\end{definition}
We warn the reader that when $G$ is a graph,
then $S=\ViSu_G(w)$ depends on $G$, i.e.,
on the map $\iota_G$, rather than on $v_0,e_1,\ldots,e_k,v_k$ alone
(since $\Edir_S$ includes $\iota_G e_i$
for $i\in[k]$, so we need to know $\iota_G$ to determine $\Edir_S$);
by contrast $\ViSu_G(w)$ does not depend on $G$, i.e., can be 
inferred from the sequence $(v_0,\ldots,e_k,v_k)$ alone,
if $G$ is a digraph (this is very easy to see) or if $G$ is a 
coordinatized graph (this is not hard to see).

\subsection{The Order of a Graph and Its Fundamental Properties}

\begin{definition}\label{de_order}
Let $S$ be a graph.  We define its
{\em order} to be
$$
\ord(S) \eqdef (\#E_S)-(\#V_S)\ ;
$$
we define the order of a walk, $w$, to be $\ord(S)$ where
$S=\ViSu(w)$.
We let $\snbc_{r}(G,k)$ 
(respectively $\snbc_{<r}(G,k)$ and $\snbc_{\ge r}(G,k)$)
denote the number of SNBC walks in $G$ of
length $k$ and order $r$ (respectively $<r$ and $\ge r$).
\end{definition}
Note that each half-loop contributes $1$ to the order.
The order of a graph is fundamental to all of our trace
methods since
for our basic models
$\cC_n(B)$ we will prove that 
$$
\EE_{G\in\cC_n(B)}
\bigl[\#\{\mbox{subgraphs of $G$ isomorphic to $S$} \} \bigr]
= c n^{-\ord(S)}(1+o(1/n))
$$
for some $c>0$, unless this expected value is $0$ for all $n$
sufficiently large; see also
Definitions~\ref{de_strongly_algebraic} and
\ref{de_algebraic}.

\subsection{$(B,\nu)$-Bounded Functions and Expansions}

For $d\in\naturals$, \cite{friedman_alon} formally defined
{\em $d$-Ramanujan functions}, which are fundamental to our
{\em asymptotic expansions} of expected traces; these are also
implicit in \cite{broder,friedman_random_graphs}.  Here we
define a generalization of this notion, needed when $\#V_B>1$.

\begin{definition}\label{de_polyexponential_growth}
By a {\em (univariate) polyexponential function} we mean a function
$\naturals\to\complex$ (or sometimes $\integers_{\ge 0}\to\complex$
of the form
$$
f(k) = \sum_{i=1}^m \nu_i^k p_i(k) \ ,
$$
where $\nu_i\in\complex$ and the $p_i$ are polynomials with coefficients
in $\complex$, where we understand that if $\nu_i=0$ then the
expression $\nu_i^k p_i(k)$ refers to some
function that vanishes for sufficiently
large $k$; 
we refer to the $\nu_i$ as the {\em bases} of $f$, and the
$p_i$ as the {\em polynomials} of $f$.
We say that a function $f\from\naturals\to\complex$ (or
$\integers_{\ge 0}\to\complex$) is 
{\em of growth $\rho$} for a real $\rho\ge 0$ if for every 
$\epsilon>0$ 
$$
|f(k)|\le (\rho+\epsilon)^k.
$$
for sufficiently large $k$.
\end{definition}

Jordan canonical form shows that
if $M$ is an $n\times n$ matrix, then for any $i,j\in[n]$,
$f(k)\eqdef (M^k)_{i,j}$ is linear combination of functions of the form
$$
k(k-1)\ldots(k-\ell+1)\nu^{k-\ell}
$$
where $\nu$ is an eigenvalue of $M$ and $\ell$ is any integer less than
the maximum Jordan block of $M$ associated to the eigenvalue of $\nu$.
Hence $f(k)$ is 
a polyexponential function whose bases
are the eigenvalues of $M$; this also explains our convention regarding
$\nu_i=0$ in Definition~\ref{de_polyexponential_growth}.

\begin{definition}\label{de_Ramanujan_function}
Let $B$ be a graph and $\nu\ge 1$ a real number.
By a {\em $(B,\nu)$-bounded function} we mean a function
$f\from\naturals\to\complex$
that is the sum of a function of growth $\nu$ plus a
polyexponential function whose bases are bounded in absolute
value by $\mu_1(B)$;
furthermore we say that $f$ is a {\em $(B,\nu)$-Ramanujan function} 
if all these bases are eigenvalues of $H_B$.
\end{definition}

\begin{definition}\label{de_Ramanujan_expansion}
Let $N\subset \naturals$ be an infinite subset.
Let $f=f(k,n)$ be a function $\naturals\times N\to\complex$,
let $B$ be a graph,
$\nu>0$ a real number, and 
$r\in\naturals$.
We say that $f$ has a
{\em $(B,\nu)$-Ramanujan (respectively, $(B,\nu)$-bounded)
asymptotic expansion to order $r$} if
there is a function $c_r(k)$ of growth bounded by $\mu_1(B)$,
and $(B,\nu)$-Ramanujan (respectively, $(B,\nu)$-bounded)
functions $c_0(k),\ldots,c_{r-1}(k)$, and a constant
$C$ such that
$$
f(k,n)=c_0(k)+c_1(k)/n+\cdots+c_{r-1}(k)/n^{r-1}+
O(1)c_r(k)/n^r,
$$
for all $1\le k\le n^{1/2}/C$ and $n\in N$
(where the constant in $O(1)$ is
universal, i.e., independent of $k,n$).
\end{definition}
In our trace methods,
$c_r(k)$ is actually bounded by $O(k^{2r}\mu_1^k(B))$, so the above
definition is looser but simpler.
Also, for our trace methods
it suffices to have the expansion to hold in the range $1\le k\le g(n)$ 
for any $g(n)\gg \log n$; however, typically the expansions are
valid in the above larger range of $k$.

\begin{example}
Theorem~2.18 of
\cite{friedman_random_graphs} implies that if $d$ is even and $B$ has
$d/2$ whole loops, then
\begin{equation}\label{eq_expected_snbc}
f(k,n) \eqdef 
\EE_{G\in\cC_n(B)}[\snbc(G,k)]
=\EE_{G\in\cC_n(B)}[\trace(H_G^k)]
\end{equation} 
has a $(B,\nu)$-Ramanujan expansion of order $r$, where $r$ is proportional
to $d^{1/2}$; Theorem~2.12 of
\cite{friedman_alon} implies that $f$ above cannot have a
$(B,\nu)$-Ramanujan expansion to order roughly $d^{1/2}\log d$, due 
to {\em tangles} that we describe below.
\end{example}

[In the above example the word ``implies'' is used since the theorems quoted
above work with the expected values of $A_G^k$, not $H_G^k$;
\eqref{eq_Ihara_det} allows us to translate between asymptotic expansions
regarding such expansions when $B$ is regular.]

\subsection{Tangles}

\begin{definition}\label{de_tangle}
Let $\nu>1$ be a real number, and $r\in\naturals$.
We say that a connected graph, $\psi$, is a
{\em $\ge\nu$-tangle} (respectively, {\em $(\ge\nu,<r)$-tangle})
if
$\mu_1(\psi)\ge\nu$ (respectively, and $\ord(\psi)<r$) and
all vertices of $\psi$ have degree at least two.
We use
$$
{\rm TangleFree}(\ge\nu,<r) \quad\mbox{and}\quad
{\rm HasTangles}(\ge\nu,<r)
$$
to denote, respectively,
the class of graphs, $G$, such that no (respectively, some)
subgraph of $G$ is
a $(\ge\nu,<r)$-tangle;
we refer to the elements of this
class as {\em $(\ge\nu,<r)$-tangle-free}
(respectively, {\em having $(\ge\nu,<r)$-tangles}).
\end{definition}
In Article~III we will prove that for any $\nu,r$ there is a finite
set of graphs, $\psi_1,\ldots,\psi_s$ such that any graph in
${\rm HasTangles}(\ge\nu,<r)$ has a subgraph isomorphic to some
$\psi_i$ (see also Lemma~9.2 of
\cite{friedman_alon}); this is not true for
${\rm HasTangles}(>\nu,<r)$ defined in the evident sense, i.e.,
replacing the condition 
$\mu_1(\psi)\ge\nu$ with $\mu_1(\psi)>\nu$.
For this reason it will turn out to be crucial---at least in our
methods---that we work with the condition
$\mu_1(\psi)\ge\nu$ and not $\mu_1(\psi)>\nu$ when we speak of tangles;
by contrast, the condition $<r$ is equivalent to $\le r-1$, and it is
not crucial (but mildly more notationally convenient) to work with
the strict inequality $<r$.
We write $(\ge\nu,<r)$
instead of $(\nu,r)$ to emphasize the weak inequality $\ge\nu$
and the strict inequality $<r$.

\section{Main Theorems in This Series of Articles}
\label{se_main_theorems}

In this section we will state the main theorems in this series of 
articles and comment on their proofs.
More comments on these proofs will be made in
Sections~\ref{se_art_expansion}--\ref{se_art_sharp} where we discuss
the individual contents of each of Articles~II--VI.
Complete proofs of the theorems appear in Articles~II--VI.

Our theorems are valid for all of our basic models
(Definition~\ref{de_models}); however, in this section we often
use the term {\em algebraic models}
(Definition~\ref{de_algebraic}), since this is a wider class of
models and a simpler setting for our proofs.
Hence the reader is free to substitute ``our basic models''
for any theorem in this section
in which the term ``algebraic models'' appears.

\subsection{The Tangle Power}

\begin{definition}\label{de_occurs}
Let $\{\cC_n(B)\}_{n\in N}$ be a model over a graph, $B$.
We say that a graph, $S$, {\em occurs in $\{\cC_n(B)\}_{n\in N}$} if
for all sufficiently large $n\in N$ there is an element $G\in\cC_n(B)$ such
that some subgraph of $G$ is isomorphic to $S$.
\end{definition}
In our basic models, and the {\em algebraic models} that we define
later, it turns out that if $S$ occurs in $\{\cC_n(B)\}_{n\in N}$
then there are $C_1,C_2>0$ for which
\begin{equation}\label{eq_contains}
C_1 n^{-\ord(S)} \le 
\Prob_{G\in\cC_n(B)}[\mbox{$G$ contains a subgraph isomorphic to $S$}]
\le C_2 n^{-\ord(S)} 
\end{equation}
for all sufficiently large $n$
(i.e., aside for some small values of $n$ where
the above probability can be zero);
the lower bound is proven in Article~III 
(see Theorem~\ref{th_extra_needed} below), and the upper bound
follows easily from the definition of {\em algebraic} model,
Definition~\ref{de_algebraic} below.

To motivate the above definition, we note that if $n(\#V_B)<\#V_S$,
then no $G\in\cC_n(B)$ has a subgraph isomorphic to
$S$, since $\#V_G<\#V_S$.  Hence the term ``sufficiently large''
is necessary.  Furthermore, our models of interest will have a
dichotomy: either $S$ occurs in the model $\{\cC_n(B)\}_{n\in N}$,
or for sufficiently large $n$ no $G\in\cC_n(B)$ will contain
a subgraph isomorphic to $S$. 

\begin{definition}\label{de_tau_tang}
Let $\{\cC_n(B)\}_{n\in N}$ be a model over a graph, $B$.
By the {\em tangle power of $\{\cC_n(B)\}$}, denoted $\tau_{\rm tang}$,
we mean the smallest order, $\ord(S)$, of any graph, $S$, that
occurs in $\{\cC_n(B)\}$ and satisfies
$\mu_1(S)>\mu_1^{1/2}(B)$.
\end{definition}

The tangle power is relatively easy to bound from below.  In fact,
in Article~VI we use
the results of Section~6.3 of \cite{friedman_alon}
to prove that for any algebraic model over a $d$-regular graph, $B$,
$$
\tau_{\rm tang} \ge  m=m(d)
$$
where
$$
m(d) =
\Bigl\lfloor \bigl( (d-1)^{1/2} - 1 \bigr)/2  \Bigr\rfloor +1 ,
$$
and for each even $d\ge 4$ there are $d$-regular $B$ where equality holds
(namely, if $B$ has $m(d)+1$ whole-loops about some vertex,
which is the case if $B$ is a bouquet of $d/2$ whole loops).

\subsection{The HasTangles Probability Upper Bound}

Our interest
in the tangle power, $\tau_{\rm tang}$, of an
algebraic model is explained by the theorems we state in this subsection
and the next.
These theorems deal with non-Alon eigenvalues
for the event ${\rm HasTangles}(\ge\nu,<r)$, and they are relatively
easy to prove.

\begin{theorem}\label{th_hastangles_upper_bound}
Let $\{\cC_n(B)\}_{n\in N}$ be an algebraic model of tangle
power $\tau_{\rm tang}$ over a graph, $B$.
For every $\nu>(d-1)^{1/2}$ and 
any $r\in\naturals$
there is a constant $C=C(\nu,r)$ such that
\begin{equation}\label{eq_hastangles_upper_bound}
\Prob_{G\in\cC_n(B)}[ 
{\rm HasTangles}(\ge\nu,<r)
 ] \le
C(\nu,r) n^{-\tau_{\rm tang}} .
\end{equation} 
\end{theorem}
Hence the condition that a $G\in\cC_n(B)$ has a $(\ge\nu,<r)$-tangle
for any $r$ and any $\nu$ of interest to us---namely those
$\nu>(d-1)^{1/2}$---is an event of probability 
of order at most $n^{-\tau_{\rm tang}}$.

This theorem is not difficult to prove and is based on
the proof of Lemma~9.2
in \cite{friedman_alon}; let us sketch the proof: 
a ``compactness'' property
of variable-length graphs shows that there are only finitely
many graphs that are minimal with respect to inclusion
among all graphs, $S$, satisfying
$\mu_1(S)\ge\nu$ and $\ord(S)<r$.
Since any such $S$ occurs as a subgraph of some $G\in\cC_n(B)$
with probability $O(n^{-\ord(S)})$, which is
$O(n^{-\tau_{\rm tang}})$ for $\mu_1(S)>(d-1)^{1/2}$,
the union bound completes the proof.

We remark that the constant $C=C(\nu,r)$ in
Theorem~\ref{th_hastangles_upper_bound} seems very difficult to
bound: it depends---at least in
principle---on the number of minimal $S$ with
$\mu_1(S)\ge\nu$ and $\ord(S)<r$;
the ``compactness'' approach above
gives no effective bound on this number of $S$.

\subsection{The HasTangles Probability Lower Bound}

To complement Theorem~\ref{th_hastangles_upper_bound} we give the 
following lower bound.

\begin{theorem}\label{th_hastangles_lower_bound}
Let $\{\cC_n(B)\}_{n\in N}$ be an algebraic model of tangle
power $\tau_{\rm tang}$ over a graph, $B$.
Let $S$ be a connected graph that occurs in $\cC_n(B)$ with
$\ord(S)=\tau_{\rm tang}$ and $\mu_1(S)>(d-1)^{1/2}$.
Then for any $r>\ord(S)$ and $\nu\le\mu_1(S)$,
there is a constant $C'>0$ and $\epsilon_0>0$ such that
for all $n$ sufficiently large we have
$$
\Prob_{G\in\cC_n(B)}\Bigl[ 
\bigl( G \in {\rm HasTangles}(\ge\nu,<r) \bigr)
\ \mbox{\rm and}\  %
\bigl({\rm NonAlon}_d(G;\epsilon_0)>0  \bigr)
\Bigr] \ge
C' n^{-\tau_{\rm tang}}  .
$$
\end{theorem}

This theorem is based on two facts:
first, 
$S$, as above,
is a subgraph of $G\in\cC_n(B)$ with probability at
least $C' n^{-\tau_{\rm tang}}$ (this follows from
Theorem~\ref{th_extra_needed} below, proven in Article~III).
Second,
in Article~V,
using the
methods of Friedman-Tillich \cite{friedman_tillich_generalized} and the
``Curious Theorem'' of \cite{friedman_alon}, we prove the following
theorem.

\begin{theorem}\label{th_relative_friedman_tillich}
Let $B$ be a $d$-regular graph, and
$S$ be any graph with $\mu_1(S)\ge (d-1)^{1/2}$.
If $G\to B$ is a covering map of degree $n$, and $G$ has
a subgraph isomorphic to $S$, then $G\to B$
has a new adjacency eigenvalue greater than
$$
\mu_1(S)+\frac{d-1}{\mu_1(S)} - \epsilon(n)
$$
where $\epsilon(n)\to 0$ as $n\to\infty$.
\end{theorem}

\subsection{The Expected Number of Non-Alon Eigenvalues in
Tangle-Free Graphs}

\begin{theorem}\label{th_rel_Alon_regular2}
Let $\cC_n(B)$ be an
algebraic model
over a $d$-regular graph $B$.
For any $\nu$ with
$(d-1)^{1/2}<\nu<d-1$, 
let $\epsilon'$ be given by
\begin{equation}\label{eq_nu_epsilon_prime}
2(d-1)^{1/2}+\epsilon' = \nu + \frac{d-1}{\nu}
\end{equation} 
(since $\nu+(d-1)/\nu$ is a monotone increasing function in
$\mu>(d-1)^{1/2}$, we have that $\epsilon'>0$).
Then either
\begin{enumerate}
\item
there is an
integer $\tau=\tau_{\rm alg}(\nu,r)\ge 1$ such that
for any sufficiently small $\epsilon>0$ there are constants
$C=C(\epsilon),C'>0$ such that for sufficiently large $n$ we have
\begin{equation}\label{eq_rel_alon_expect_lower_and_upper}
n^{-\tau} C'
\le
\EE_{G\in\cC_n(B)}[ \II_{{\rm TangleFree}(\ge\nu,<r)}(G)
{\rm NonAlon}_d(G;\epsilon'+\epsilon) ] 
\le 
n^{-\tau}   C(\epsilon) ,
\end{equation} 
or
\item
for all $j\in\naturals$ and $\epsilon>0$ we have
$$
\EE_{G\in\cC_n(B)}[ \II_{{\rm TangleFree}(\ge\nu,<r)}(G)
{\rm NonAlon}_d(G;\epsilon'+\epsilon) ] \le O(n^{-j}) 
$$
in which case we use the notation $\tau_{\rm alg}(\nu,r)=+\infty$.
\end{enumerate}
\end{theorem}

[We believe that $\tau_{\rm alg}(\nu,r)$ is never $+\infty$ in our
basic models, although this doesn't seem important to us.]

The proof of Theorem~\ref{th_rel_Alon_regular2} represents the majority
of Articles~II--IV.
Article~II is based on the asymptotic expansions proven in
\cite{friedman_random_graphs};
Articles~III and~IV are based on the way \cite{friedman_alon} utilizes
such expansions to prove the original Alon conjecture.

The reason we use the term {\em algebraic power} is because
$\tau_{\rm alg}(\nu,r)$ is determined from the polyexponential
parts of the coefficients of the
asymptotic expansion of
$$
f(k,n) =
\EE_{G\in\cC_n(B)}[ \II_{{\rm TangleFree}(\ge\nu,<r)}(G)
\Trace(H_G^k) ] ,
$$
and the existence of these expansions 
is due to the fact that the model is assumed to be algebraic
(in the sense of
Definition~\ref{de_algebraic}).

If a graph $G\in\cC_n(B)$ has a nonzero number of non-Alon new adjacency
eigenvalues, then this number is between $1$ and $(n-1)(\#V_B)$.
Hence we get the following corollary.

\begin{theorem}\label{th_rel_Alon_regular_prob_tanglefree}
Let $\cC_n(B)$ be an
algebraic model
over a $d$-regular graph $B$, and
let $r\in\naturals$ and $(d-1)^{1/2}<\nu<d-1$; let
$\epsilon'$ and
$\tau=\tau_{\rm alg}(\nu,r)$ be as in Theorem~\ref{th_rel_Alon_regular2}.
Then for any $\epsilon$
there are constants $C=C(\epsilon),C'$ such that for sufficiently large
$n$ he have
\begin{equation}\label{eq_rel_alon_prob_tanglefree}
C' n^{-\tau-1}
\le
\Prob_{G\in\cC_n(B)}\Bigl[
\bigl( G \in {\rm TangleFree}(\ge\nu,<r) \bigr)
\ \mbox{\rm and}\  %
\bigl({\rm NonAlon}_d(G;\epsilon'+\epsilon)>0  \bigr)
\Bigr] 
\le
C n^{-\tau}
\end{equation} 
(unless $\tau=+\infty$, in which case the upper bound above holds for
all $\tau\in\naturals$).
\end{theorem}

\subsection{The Algebraic Power of a Model}

For fixed $\nu>(d-1)^{1/2}$ and
$r$, all graphs are divided into two families, namely
$$
{\rm HasTangles}(\ge\nu,<r) , \ \  %
{\rm TangleFree}(\ge\nu,<r).
$$
Theorem~\ref{th_rel_Alon_regular_prob_tanglefree} implies that
$$
\Prob_{G\in\cC_n(B)}
\bigl[{\rm NonAlon}_d(G;\epsilon)>0  \bigr]
$$
is bounded above by
\begin{equation}\label{eq_new_alon_upper}
\Prob_{G\in\cC_n(B)}
\bigl[ {\rm HasTangles}(\ge\nu,<r) \bigr] +
O\bigl( n^{-\tau(\nu,r)} \bigr).
\end{equation} 
In this section we define the {\em algebraic power} of a model, which
arises when we choose $\nu,r$ to minimize 
\eqref{eq_new_alon_upper} to within a multiplicative constant.
This involves two observations.

The first observation is that, 
for trivial reasons, if 
$\nu_1\le\nu_2$ and $r_1\ge r_2$, then the event of being
$(\ge\nu_1,<r_1)$ tangle-free is more restrictive than being
$(\ge\nu_2,<r_2)$ tangle-free.  Hence the characteristic function
of ${\rm TangleFree}(\ge\nu_1,<r_1)$ is smaller than that of
${\rm TangleFree}(\ge\nu_2,<r_2)$, and hence
\begin{equation}\label{eq_tau_monotonicity}
\tau(\nu_1,r_1) \ge
\tau(\nu_2,r_2).
\end{equation} 
Said otherwise, by decreasing $\nu$ and increasing $r$,
the set ${\rm HasTangles}(\ge\nu,<r)$ either says the same or gets larger
and ${\rm TangleFree}(\ge\nu,<r)$ the same or smaller.

The second observation is that
Theorem~\ref{th_hastangles_lower_bound} implies that for
$\epsilon>0$ sufficiently small, 
$$
\Prob_{G\in\cC_n(B)}
\bigl[{\rm NonAlon}_d(G;\epsilon)>0  \bigr]
\ge C' n^{-\tau_{\rm tang}}
$$
for $n$ sufficiently large.  Hence 
\eqref{eq_new_alon_upper} is at best $O(n^{-\tau_{\rm tang}})$.
But this theorem also implies that 
$$
\Prob_{G\in\cC_n(B)}
\bigl[ {\rm HasTangles}(\ge\nu,<r) \bigr]
$$
is never larger any larger than a constant times $n^{-\tau_{\rm tang}}$.
Hence we can take $\nu$ as close to $(d-1)^{1/2}$ and $r$ as large
as we like, which in \eqref{eq_new_alon_upper} makes the
$\tau(\nu,r)$ as favourable as possible
in view of \eqref{eq_tau_monotonicity},
without giving up more than a constant in the
HasTangles term.

This discussion motivates the following definition.

\begin{definition}\label{de_algebraic_power}
Let $\{\cC_n(B)\}_{n\in N}$ be an algebraic model over a $d$-regular
graph $B$.
For each $r\in\naturals$ and $\nu$ with $(d-1)^{1/2}<\nu<d-1$, let
$\tau(\nu,r)$ be as in Theorem~\ref{th_rel_Alon_regular2}.
We define the {\em algebraic power} of the model $\cC_n(B)$ to be
$$
\tau_{\rm alg} =
\max_{\nu>(d-1)^{1/2},r} \tau(\nu,r) =
\limsup_{r\to\infty,\ \nu\to(d-1)^{1/2}}
\tau(\nu,r)
$$
where $\nu$ tends to $(d-1)^{1/2}$ from above
(and we allow $\tau_{\rm alg}=+\infty$ when this maximum is
unbounded or if $\tau(\nu,r)=\infty$ for some $r$ and $\nu>(d-1)^{1/2}$).
\end{definition}
Of course, according to Theorem~\ref{th_rel_Alon_regular2},
$\tau(\nu,r)\ge 1$ for all $r$ and all relevant $\nu$, and hence
$\tau_{\rm alg}\ge 1$.

\subsection{The First Main Theorem}

Combining the above results, and diving graphs by whether
or not they have $(\ge\nu,<r)$-tangles, with $r$ sufficiently
large and $\nu$ sufficiently close to $(d-1)^{1/2}$ to make
$\tau(\nu,r)=\tau_{\rm alg}$, we get the following theorem.

\begin{theorem}\label{th_rel_Alon_regular}
Let $B$ be a $d$-regular graph, and let
$\cC_n(B)$ be an
algebraic model of tangle power $\tau_{\rm tang}$
and algebraic power $\tau_{\rm alg}$.
Let 
$$
\tau_1 = \min(\tau_{\rm tang},\tau_{\rm alg}), \quad
\tau_2 = \min(\tau_{\rm tang},\tau_{\rm alg}+1).
$$
Then $\tau_2\ge \tau_1\ge 1$, and
for $\epsilon>0$ sufficiently small there are
$C,C'$ such that for sufficiently large $n$ we have
\begin{equation}\label{eq_first_main}
C' n^{-\tau_2}
\le
\Prob_{G\in\cC_n(B)}\bigl[
{\rm NonAlon}_d(G;\epsilon)>0  
\bigr]
\le
C n^{-\tau_1}.
\end{equation}
\end{theorem}

\subsection{The Second Main Theorem}

Generally speaking, $\tau_{\rm tang}$ tends to be relatively
easy to compute or approximate; it was computed exactly in
\cite{friedman_alon} when $B$ is a bouquet of either whole-loops
or of half-loops, for some of our basic models.
On the other hand, $\tau_{\rm alg}$ is very difficult to compute
directly, at least in the asymptotic expansions that determine it.

In \cite{friedman_random_graphs}, where $B$ was a bouquet of
$d/2$ of whole-loops (so $d$ is even), 
the analog of $\tau_{\rm alg}$ was determined in two steps:
first one proves that for certain $r$ one has
$$
\EE_{G\in\cC_n(B)}[\Trace(A_G^k)]
= c_{-1}(k)n + c_0(k) + \ldots + c_{r-1}(k)/n^{r-1} + O(1) c_r(k)/n^r
$$
where for $i<r$ the $c_i(k)$ are all of the form
$d^k p(k) + g(k)$ where $g$ is of growth $2(d-1)^{1/2}$
and $p(k)$ is a polynomial (which is not explicitly computed).
Second, 
standard counting arguments about
the {\em magnification} of expanders and taking $k=\log^2(n)$ and
$n\to\infty$ 
to deduce that these polynomials
must all vanish.
A similar two-part 
strategy was used in \cite{friedman_alon}, which involved
traces of powers of $H_G$, where the coefficients are first provably
of the form $(d-1)^k p(k)$ plus a function of growth $(d-1)^{1/2}$,
and then a ``Sidestepping Lemma'' is used.
Critical to both these computations is that the polyexponential parts
of these coefficients can be linked to a lack of magnification;
if $B$ has more than one vertex, this same strategy works provided
that $\pm (d-1)$ are the only possible larger bases of the
coefficients $c_i(k)$ in expansions of
$$
\EE_{G\in\cC_n(B)}[ \II_{{\rm TangleFree}(\ge\nu,<r)}(G)
\Trace(H_G^k) ] 
$$ 
as $\nu\to(d-1)^{1/2}$ and $r\to\infty$.
This is necessarily true if 
$B$ is a $d$-regular
Ramanujan graph (Definition~\ref{de_Ramanujan}); if $B$ is not
Ramanujan, then
our expansion theorems do not rule out polyexponential parts
whose bases may be between $d-1$ and $(d-1)^{1/2}$ in absolute value;
at present we have no proof that such bases cannot occur.

\begin{theorem}\label{th_second_main_theorem}
Let $\{\cC_n(B)\}_{n\in N}$ be one of our basic models
over $d$-regular Ramanujan
graph, $B$.
Then $\tau_{\rm alg}=+\infty$.
\end{theorem}
The above theorem holds for any algebraic model that satisfies
a certain weak {\em magnification} condition;
in Article~VI we describe this condition and prove that it holds
for all of our basic models.
The proof uses standard counting arguments; for large values of
$d$ the argument is very easy; for
small values of $d$ our argument is a more delicate
calculation similar to those in Chapter~12 of
\cite{friedman_alon}.

Whenever $\tau_{\rm alg}=+\infty$, or merely $\tau_{\rm alg}\ge\tau+1$,
then in Theorem~\ref{th_rel_Alon_regular},
$\tau_1=\tau_2=\tau_{\rm tang}$.
In \cite{friedman_alon}, some values of $\tau_{\rm tang}$ were
computed for our basic models in the case where $B$ has one
vertex.  In the next subsection we give some implications of
these computations for general $B$.

\subsection{Bounds on $\tau_{\rm tang}$}
\label{su_bounds_tau_tang}

Sections~6.3 and~6.4
of \cite{friedman_alon} develops a number of techniques to
determine what is called there $\tau_{\rm fund}$,
which is the smallest
order of a graph, $S$, that occurs in the model and has
$\mu_1(S)\ge (d-1)^{1/2}$ (note the weak inequality here).
So $\tau_{\rm tang}$ is the same except that the weak inequality
is replaced with the strong inequality 
$\mu_1(S)> (d-1)^{1/2}$.
The techniques of \cite{friedman_alon} easily prove the following results.

\begin{theorem}
Let $S$ be a connected graph of order $m$ for some $m\in\integers_{\ge 0}$.
Then
\begin{enumerate}
\item $\mu_1(S)\le 2m+1$, with equality if $S$ is a bouquet of
$m+1$ whole-loops;
\item
if $S$ has no whole-loops, then
$\mu_1(S)\le m+1$, with equality if $S$ consists of
two vertices joined by $m+2$ edges.
\end{enumerate}
\end{theorem}

Here are a number of conclusions that follow.

\begin{corollary}
\label{co_tau_tang_bounds}
For every $d\in\naturals$ with $d\ge 3$ let
\begin{align*}
m(d) & \eqdef 
\Bigl\lfloor \bigl( (d-1)^{1/2} - 1 \bigr)/2  \Bigr\rfloor + 1,
\\
m'(d) & \eqdef
\Bigl\lfloor (d-1)^{1/2} - 1  \Bigr\rfloor + 1,
\end{align*}
where $\lfloor x \rfloor$ denotes the ``floor'' of $x$, i.e.,
the largest integer no greater than $x$.
Let $B$ be a $d$-regular graph with $d\ge 3$.  
Then
\begin{enumerate}
\item for any algebraic model,
\begin{equation}\label{eq_tang_lower_bound}
\tau_{\rm tang}\ge m(d)
\end{equation} 
with equality whenever
$B$ is a graph for which there is some vertex that has $m(d)+1$
whole-loops;
\item
Moreover we have
$$
\tau_{\rm tang}\ge m'(d)
$$
in any of the following cases:
\begin{enumerate}
\item the cyclic model;
\item the cyclic-involution model of either even or odd degree;
\item $B$ has no whole-loops;
\end{enumerate}
equality holds if some vertex of $B$ has at least $m'(d)+2$ half-loops
or some two vertices of $B$ are joined by at least $m'(d)+2$ edges.
\end{enumerate}
\end{corollary}

\section{The Trace Methods of Broder-Shamir and Friedman for Random 
Graphs}
\label{se_bsf}

\myDeleteNote{
{\red Longer version in in {\tt Obsolete} subdirectory, dated 20190705}
}
In this section we review some ideas of
\cite{broder,friedman_random_graphs} and give a few examples
which help to motivate
some terminology in
Sections~\ref{se_ordered_B_strong_alg}--\ref{se_new_algebraic},
including
ordered graphs, $B$-graphs, homotopy type, algebraic models, 
(our interest in) regular languages, and $B$-types.

This section is not needed for the rest of this article, but we
find it easier to
read
Sections~\ref{se_ordered_B_strong_alg}--\ref{se_new_algebraic}
while keeping in mind the examples we give here.
All these examples are implicit in the articles
\cite{broder,friedman_random_graphs} on random graphs.

We will be brief on details here, and for a more complete
discussion we refer the reader to
Sections~\ref{se_ordered_B_strong_alg}--\ref{se_new_algebraic}
in this article,
and
to Article~II and \cite{broder,friedman_random_graphs,friedman_alon}.

\subsection{The Framework for Random Graphs}

Throughout this section $d\in\naturals$ will be even, and
$B$ will be a bouquet of $d/2$ whole-loops, with the notation
\begin{equation}\label{eq_bouquet}
V_B = \{u\}, \quad
\Edir_B=\{ f_1,\iota_B f_1,\ldots,f_{d/2},\iota_B f_{d/2} \} ;
\end{equation} 
we let $\cC_n(B)$ be the permutation model unless otherwise
stated.  Hence $\cC_n(B)$ is a
model of a random $d$-regular graph on $n$ vertices.

We will explain the methods
of Broder-Shamir \cite{broder} and Friedman \cite{friedman_random_graphs}
as they apply to estimating
$$
\EE_{G\in\cC_n(B)}[\Trace(H_G^k)] .
$$
In fact \cite{broder,friedman_random_graphs} estimate the above with
$A_G$ replacing $H_G$,
so the asymptotic expansions are slightly different; also
\cite{broder,friedman_random_graphs} work with a $2d$-regular graph,
so the $d$ in those articles is not the regularity of $B$.

\subsection{Subdividing Walks by Their Order}

The broad framework of \cite{broder,friedman_random_graphs}, stated
with $H_G$ replacing $A_G$, is as follows:
we consider
\begin{equation}\label{eq_trace_H_G_k}
f(k,n) = \EE_{G\in\cC_n(B)}[\Trace(H_G^k)]
=
\EE_{G\in\cC_n(B)}[\snbc(G,k)] ,
\end{equation} 
and for a fixed $r\in\naturals$ we want to find an asymptotic expansion
to order $r$ for $f(k,n)$, i.e.,
\begin{equation}\label{eq_asymp_bsf}
f(k,n)=
c_0(k) + c_1(k)/n + \cdots + c_{r-1}(k)/n^{r-1} + O(1) c_r(k)/n^r
\end{equation} 
for appropriate $c_i(k)$.  To do so, we write
\begin{equation}\label{eq_subdivide_by_order}
\snbc(G,k) = \snbc_0(G,k) + \snbc_1(G,k) + \cdots + \snbc_{r-1}(G,k)
+ \snbc_{\ge r}(G,k)
\end{equation}
recalling (Definition~\ref{de_order})
that $\snbc_m(G,k)$ denotes the number of SNBC walks whose visited subgraph
has order $m$, and $\snbc_{\ge r}(G,k)$ those of order at least $r$.
We now take expected values in \eqref{eq_subdivide_by_order}.

The first step is to show that for $k\le n/2$ we have
\begin{equation}\label{eq_walks_large_order}
\EE_{G\in\cC_n(B)}[\snbc_{\ge r}(G,k)] \le O(d-1)^k k^{2r}/n^r   
\end{equation} 
(\cite{friedman_random_graphs}, middle of page~352, based on
Lemma~3 of \cite{broder}).
Hence if we can show that for each $m=0,1,\ldots,r-1$ there
is an order $r$ expansion \eqref{eq_asymp_bsf} for 
\begin{equation}\label{eq_snbc_m_expected}
f_m(k,n)=\EE_{G\in\cC_n(B)}[\snbc_m(G,k)] ,
\end{equation} 
then it follows that we such an expansion for \eqref{eq_trace_H_G_k}.

\subsection{Homotopy Type}
\label{su_homotopy_type_examples}

The next broad step is to fix $m<r$ in \eqref{eq_snbc_m_expected},
and 
to subdivide all SNBC walks of order $m$ (in an arbitrary graph) as 
belonging to a finite number of
possible {\em homotopy types}.  Here is the rough idea.

If $w$ is an SNBC walk in an arbitrary graph, then each vertex of
$S=\ViSu(w)$ has degree at least two.  We easily see that
\begin{equation}\label{eq_order_local}
\ord(S)=(1/2)\sum_{v} \bigl( \deg_S'(v) - 2 \bigr)
\end{equation} 
(where $\deg_S'(v)$ is the degree of $v$ in $S$, except
that a half-loop contributes $2$ to the degree of a vertex,
since a half-loop's contribution to $\ord(S)$ is, by definition, $-1$).
This formula implies that there are a finite number of possible
{\em homotopy types} or ``shapes'' of such graphs
(see Lemma~2.4 of \cite{friedman_random_graphs}), for example,
if $S$ does not have half-loops (which is the case here since
$B$ is a bouquet of whole-loops), then
\begin{enumerate}
\item
if $m=0$, then $S$ is necessarily a {\em cycle} (i.e., $S$ is a connected
graph all of whose vertices have degree $2$)
\item
if $m=1$, (see Figure~\ref{fi_m_equals_one}, which reproduces
Figure~6 of \cite{linial_puder})
then except for the vertices of degree $2$,
either 
\begin{enumerate}
\item $S$ has one vertex of degree $4$, which is called a
``figure 8 graph,'' or 
\item $S$ has two vertices of degree $3$, which is called 
a ``barbell graph'' if each vertex has one whole-loop, or a
``theta graph'' if the two degree $3$ vertices are joined by
three edges
\end{enumerate}
\end{enumerate}
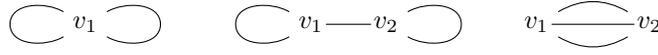
\begin{figure}
\begin{tikzpicture}
\node at (0,0)(1){$v_1$};
\draw (1) to [out=150,in=90] (-1,0) to [out=270, in=210] (1) ;
\draw (1) to [out=30,in=90] (1,0) to [out=270, in=330] (1) ;
\node at (3,0)(11){$v_1$};
\node at (4,0)(12){$v_2$};
\draw (3.2,0) to (3.8,0) ;
\draw (11) to [out=150,in=90] (2,0) to [out=270, in=210] (11) ;
\draw (12) to [out=30,in=90] (5,0) to [out=270, in=330] (12) ;
\node at (6,0)(21){$v_1$};
\node at (7.5,0)(22){$v_2$};
\draw (6.2,0) to (7.3,0) ;
\draw (21) to [out=30,in=150] (22) ;
\draw (21) to [out=-30,in=-150] (22) ;
\end{tikzpicture}
\caption{``Figure 8,'' ``Barbell,'' and ``Theta''
homotopy types for SNBC walks of order $1$}
\label{fi_m_equals_one}
\end{figure}
(If $B$ has half-loops, then there are more homotopy types of
order $1$, such as a cycle where one vertex has, in addition,
a half-loop, or a bouquet of two half-loops, or 
a path where the two endpoints each have an extra half-loop.)

Our notion of the homotopy type of a walk will remember
some extra information beyond the above homotopy type
defined for connected graphs each of whose vertices have degree
at least two.
We remark that the above notion of homotopy type is not the topological
one, since topologically the above figure 8, barbell, and theta
graphs are topologically homotopy equivalent.
Instead, our notion of homotopy type---common to trace methods
(e.g., Figure~6 in \cite{linial_puder})---is
the graph obtained by {\em suppressing} the {\em beads}, i.e., by
suppressing the
vertices of degree two in the graph, and retaining only the vertices
of degree three or higher.

\subsection{Expansions for Order Zero Walks}

The methods of Broder-Shamir yield an estimate
\begin{equation}\label{eq_broder_shamir_expansion}
\EE_{G\in\cC_n(B)}[\snbc_0(G,k)] =
c_0(k) + O(1) k^2 (d-1)^k /n,
\end{equation} 
where $c_0=c_0(k)$ is a function of $k$ that for fixed $d$ and large 
$k$ satisfies the estimate
\begin{equation}\label{eq_broder_shamir_coefficient}
c_0(k)=(d-1)^k + O(k) (d-1)^{k/2}.
\end{equation}
Refinements of \eqref{eq_broder_shamir_expansion} are given
in \cite{friedman_random_graphs}, the simplest of which is
$$
\EE_{G\in\cC_n(B)}[\snbc_0(G,k)] =
c_0(k) + c_1(k)/n + O(k^4) (d-1)^k/n^2
$$
where $c_0(k)$ satisfies \eqref{eq_broder_shamir_coefficient} and
\begin{equation}\label{eq_next_total_coef_order_zero}
c_1(k) = p(k) (d-1)^k + O(k^2)(d-1)^{k/2}
\end{equation} 
where $p(k)$ is some polynomial that is not explicitly computed.

To show \eqref{eq_broder_shamir_coefficient} 
and \eqref{eq_next_total_coef_order_zero}, we first
make a general remark.  For any $G\in\Coord_n(B)$, we easily
see that any walk in $G$ of length $k$ can be written as 
\begin{equation}\label{eq_walk_in_G}
w_G = \bigl( (u,i_0), (e_1,i_0), (u,i_1), \ldots, (e_k,i_{k-1}), (u,i_k) \bigr)
\end{equation}
where $i_0,\ldots,i_k\in[n]$ and
\begin{equation}\label{eq_walk_in_B}
w_B = (u,e_1,u,e_2,\ldots,e_k,u)
\end{equation}
is a walk in $B$; we also easily check that
$w_G$ is SNBC iff $i_k=i_0$ and $w_B$ is SNBC in $B$.
Our discussion of homotopy type implies that $S=\ViSu(w_G)$ has
order $0$ implies that $S$ is a cycle of some length $k'$
(i.e., $S$ is connected with $k'$ vertices, each of degree two).
Hence it suffices to check, for each $k'$ between $1$ and $k$,
whether or not $w_G$ is a cycle of length $k'$; since $i_k=i_0$,
we must $k'|k$ ($k'$ divides $k$).

So fix an $i_0\in[n]$ and an SNBC
walk $w_B$ as in \eqref{eq_walk_in_B}; if $\sigma$ is the 
random permutation assignment associated to a $G\in\cC_n(B)$,
then $i_1,\ldots,i_k$ are given by
\begin{equation}\label{eq_i_given_recursively}
i_1 = \sigma(e_1)i_0, \ \ldots, i_k = \sigma(e_k)i_{k-1} ;
\end{equation} 
for $k'|k$,
let $\cE_{{\rm Cycle},k'}(w_B,i_0)$ be the event that the visited subgraph 
of $w_G$ is a cycle of length $k'$, where $w_G$ is
determined by $w_B,i_0$ and \eqref{eq_i_given_recursively}.
The Broder-Shamir method shows that 
\begin{equation}\label{eq_cycle_k_bs}
\Prob_{G\in\cC_n(B)}[\cE_{{\rm Cycle},k}(w_B,i_0)] = 1/n + O(k^2)/n^2;
\end{equation} 
their proof considers the probability that $i_1$ is distinct from
$i_0$, and then that $i_2$ is distinct from $i_0,i_1$, etc.;
this argument is quite robust and is used in
\cite{friedman_relative} for arbitrary base graph, $B$.
Checking that the eigenvalues of $H_B$ are $(d-1)$ (with multiplicity one)
and $\pm 1$, we easily
check that the number of SNBC walks, $w_B$, of length $k$ is
the trace of $H_B^k$, which is therefore $(d-1)^k + O(1)$.
Multiplying the right-hand-side of \eqref{eq_cycle_k_bs} by
$(d-1)^k + O(1)$ and then by $n$ for the $n$ possible values of $i_0$
yields the $(d-1)^k$ terms of \eqref{eq_broder_shamir_expansion}
and \eqref{eq_broder_shamir_coefficient}.

The only other contributions to $\snbc_0(G,k)$ come from $w_G$ whose
visited subgraph is a cycle of length $k'<k$ and $k'|k$,
and therefore $k'\le k/2$;  since such $w_B$ must be powers
of an SNBC walk of length $k'$, we get a terms bounded by
$(d-1)^{k/2}+O(1)$ for each $k'|k$, which implies
\eqref{eq_broder_shamir_expansion} and
\eqref{eq_broder_shamir_coefficient}.

To refine \eqref{eq_broder_shamir_expansion} we refine
\eqref{eq_cycle_k_bs} by showing that
\begin{equation}\label{eq_cycle_k_refined}
\Prob_{G\in\cC_n(B)}[\cE_{{\rm Cycle},k}(w_B,i_0)] = 
1/n + c_1(w_B)/n^2 + O(k^4)/n^2
\end{equation} 
where $c_1(w_B)$ genuinely depends on $w_B$, and is given by 
\begin{equation}\label{eq_next_coef_order_zero}
c_1(w_B) 
= \sum_{1\le j_1<j_2\le d/2} a_{w_B}(f_{j_1},f_{j_2}) .
\end{equation} 
where $a_{w_B}(f_j)$ is the number of times the directed edges
$f_j,\iota_B f_j$ appear in the list of edges $e_1,\ldots,e_k$.
This is based on the exact formula 
\begin{align}
\label{eq_prob_cycle_k}
& \Prob_{G\in\cC_n(B)}[\cE_{{\rm Cycle},k}(w_B,i_0)]  
\\
\label{eq_exact_formula_order_zero}
= & (n-1)(n-2)\ldots(n-k+1) 
\prod_{j=1}^{d/2} \frac{1}{n(n-1)\ldots \bigl( n-a_{w_B}(f_j)+1 \bigr)} ,
\end{align}
where the $(n-1)(n-2)\ldots(n-k+1)$ represents the number of ways
to choose $i_1,\ldots,i_{k-1}\in[n]$ such that
$i_0,\ldots,i_{k-1}$ are distinct, and the other terms are the
probability that \eqref{eq_i_given_recursively} holds
for any fixed such $i_0,\ldots,i_{k-1}$.
We then expand
\begin{equation}\label{eq_n_k_asymp}
(n-1)(n-2)\ldots(n-k+1) = n^{k-1} - n^{k-2} \binom{k}{2} +
O\bigl( n^{-k-3} \bigr) k^3
\end{equation} 
and similarly expand each factor in the product in
\eqref{eq_exact_formula_order_zero};
this, and the fact that $k=a_{w_B}(f_1)+\cdots+a_{w_B}(f_{d/2})$ yields 
\eqref{eq_cycle_k_refined} with
$$
c_1(w_B) = -\binom{k}{2} + \binom{a_{w_B}(f_1)}{2} + \cdots 
+ \binom{a_{w_B}(f_{d/2})}{2}
$$
which simplifies to \eqref{eq_next_coef_order_zero}.

Given \eqref{eq_cycle_k_refined}, we prove that
\begin{equation}\label{eq_snbc_sum}
\sum_{w_B\in \SNBC(B,k)} c_1(w_B) 
\end{equation} 
is of the form $p(k)(d-1)^k+(-1)^k q(k)+r(k)$ for
polynomials $p,q,r$ (see top of page~346, \cite{friedman_random_graphs}),
which gives the main term in \eqref{eq_next_total_coef_order_zero}.

We remark that \eqref{eq_snbc_sum} is an abstract sum over
$\SNBC(B,k)$ without reference to trace methods or models $\cC_n(B)$.
Article~II in this series and \cite{friedman_random_graphs} deal
with generalizations of such sums.

\subsection{Algebraic Models}

Sums like
\eqref{eq_snbc_sum} arise for similar reasons in our trace methods,
and the key fact is that $c_1(w_B)$ is a polynomial in
the parameters $a(f_j)=a_{w_B}(f_j)$, such as
\eqref{eq_next_coef_order_zero}.
This approach works for any model $\cC_n(B)$ where generalizations
of probabilities like \eqref{eq_prob_cycle_k}
have
similar asymptotic expansions in powers of $1/n$ to any fixed
order $r$, whose coefficients are polynomials in
the $a_j(w_B)$.
This is the main condition for a model to be {\em algebraic}.
We will elaborate on this later in this section.

\subsection{Expansions for Order One Walks: Theta Graph Example}
\label{su_theta_graph_example}

Next we discuss asymptotic expansions for
$$
\EE_{G\in\cC_n(B)}[\snbc_1(G,k)]  .
$$
In \cite{friedman_random_graphs} we write this as
$$
\EE_{G\in\cC_n(B)}[{\rm fig8}(G,k)]
+ \EE_{G\in\cC_n(B)}[{\rm barb}(G,k)]
+ \EE_{G\in\cC_n(B)}[{\rm theta}(G,k)]
$$
where ${\rm fig8}(G,k)$ denotes those walks, $w_G$, with
$S=\ViSu(w_G)$ being a figure 8 graph
(Figure~\ref{fi_m_equals_one}), and similarly for barbell graphs
and theta graphs.
Let us focus on the theta graph term and study some example of walks
whose visited subgraphs are theta graphs.

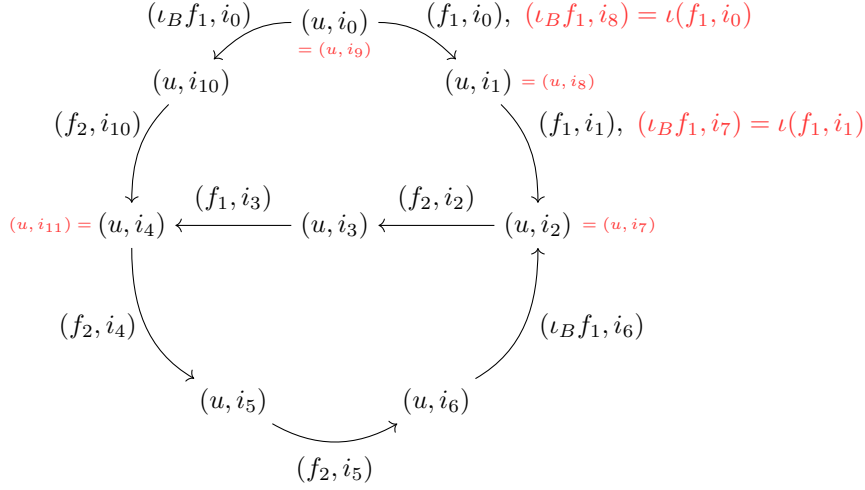
\begin{figure}
\begin{tikzpicture}[scale=0.9]
\node at (0,3)(0){$(u,i_0)$};
\node at (0,2.6)(99){\tiny\red $=(u,i_9)$};
\node at (2.1213,2.1213)(1){$(u,i_1)$};
\node at (3.3,2.1213)(88){\tiny\red $=(u,i_8)$};
\node at (3,0)(2){$(u,i_2)$};
\node at (4.2,0)(77){\tiny\red $=(u,i_7)$};
\node at (0,0)(3){$(u,i_3)$};
\node at (-3,0)(4){$(u,i_4)$};
\node at (-4.2,0)(1111){\tiny\red $(u,i_{11})=$};
\node at (-1.5,-2.5981)(5){$(u,i_5)$};
\node at (1.5,-2.5981)(6){$(u,i_6)$};
\node at (-2.1213,2.1213)(7){$(u,i_{10})$};
\draw[->] (0) to [out=0,in=135] (1) ;
\draw[->] (1) to [out=-45,in=90] (2) ;
\draw[->] (2) -- (3) ;
\draw[->] (3) -- (4) ;
\draw[->] (4) to [out=-90,in=150] (5) ;
\draw[->] (5) to [out=-30,in=210] (6) ;
\draw[->] (6) to [out=30,in=-90] (2) ;
\draw[->] (0) to [out=180,in=45] (7) ;
\draw[->] (7) to [out=225,in=90] (4) ;
\node at (2.0,3.1)(10){$\ \ \ \ \ \ \ \ \ \ \ \ \ \ \ \ \ \ \ \ \ \ \ \ \ \ \ (f_1,i_0),\ {\tiny\red (\iota_B f_1,i_8)=\iota(f_1,i_0)}$} ;
\node at (3.4,1.5)(11){$\ \ \ \ \ \ \ \ \ \ \ \ \ \ \ \ \ \ \ \ \ \ \ \ \ \ \ \ \ \ \ (f_1,i_1),\ {\tiny\red (\iota_B f_1,i_7)=\iota(f_1,i_1)}$} ;
\node at (1.5,0.4)(12){$(f_2,i_2)$} ;
\node at (-1.5,0.4)(13){$(f_1,i_3)$} ;
\node at (-3.5,-1.5)(14){$(f_2,i_4)$} ;
\node at (0,-3.6)(15){$(f_2,i_5)$} ;
\node at (3.8,-1.5)(16){$(\iota_B f_1,i_6)$} ;
\node at (-3.5,1.5)(17){$(f_2,i_{10})$} ;
\node at (-2.0,3.1)(18){$(\iota_B f_1,i_0)$} ;
\end{tikzpicture}
\caption{The first 11 steps of walk that visits a theta graph}
\label{fi_sample_w_G_theta}
\end{figure}
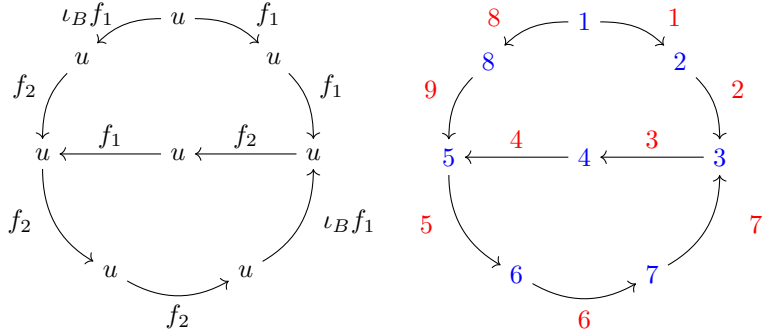
\begin{figure}
\begin{tikzpicture}[scale=0.6]
\node at (0,3)(0){$u$};
\node at (2.1213,2.1213)(1){$u$};
\node at (3,0)(2){$u$};
\node at (0,0)(3){$u$};
\node at (-3,0)(4){$u$};
\node at (-1.5,-2.5981)(5){$u$};
\node at (1.5,-2.5981)(6){$u$};
\node at (-2.1213,2.1213)(7){$u$};
\draw[->] (0) to [out=0,in=135] (1) ;
\draw[->] (1) to [out=-45,in=90] (2) ;
\draw[->] (2) -- (3) ;
\draw[->] (3) -- (4) ;
\draw[->] (4) to [out=-90,in=150] (5) ;
\draw[->] (5) to [out=-30,in=210] (6) ;
\draw[->] (6) to [out=30,in=-90] (2) ;
\draw[->] (0) to [out=180,in=45] (7) ;
\draw[->] (7) to [out=225,in=90] (4) ;
\node at (2.0,3.1)(10){$f_1$} ;
\node at (3.4,1.5)(11){$f_1$} ;
\node at (1.5,0.4)(12){$f_2$} ;
\node at (-1.5,0.4)(13){$f_1$} ;
\node at (-3.5,-1.5)(14){$f_2$} ;
\node at (0,-3.6)(15){$f_2$} ;
\node at (3.8,-1.5)(16){$\iota_B f_1$} ;
\node at (-3.4,1.5)(17){$f_2$} ;
\node at (-2.0,3.1)(18){$\iota_B f_1$} ;
\end{tikzpicture}\quad%
\begin{tikzpicture}[scale=0.6]
\node [blue] at (0,3)(0){$1$};
\node [blue] at (2.1213,2.1213)(1){$2$};
\node [blue] at (3,0)(2){$3$};
\node [blue] at (0,0)(3){$4$};
\node [blue] at (-3,0)(4){$5$};
\node [blue] at (-1.5,-2.5981)(5){$6$};
\node [blue] at (1.5,-2.5981)(6){$7$};
\node [blue] at (-2.1213,2.1213)(7){$8$};
\draw[->] (0) to [out=0,in=135] (1) ;
\draw[->] (1) to [out=-45,in=90] (2) ;
\draw[->] (2) -- (3) ;
\draw[->] (3) -- (4) ;
\draw[->] (4) to [out=-90,in=150] (5) ;
\draw[->] (5) to [out=-30,in=210] (6) ;
\draw[->] (6) to [out=30,in=-90] (2) ;
\draw[->] (0) to [out=180,in=45] (7) ;
\draw[->] (7) to [out=225,in=90] (4) ;
\node [red] at (2.0,3.1)(10){$1$} ;
\node [red] at (3.4,1.5)(11){$2$} ;
\node [red] at (1.5,0.4)(12){$3$} ;
\node [red] at (-1.5,0.4)(13){$4$} ;
\node [red] at (-3.5,-1.5)(14){$5$} ;
\node [red] at (0,-3.6)(15){$6$} ;
\node [red] at (3.8,-1.5)(16){$7$} ;
\node [red] at (-3.4,1.5)(17){$9$} ;
\node [red] at (-2.0,3.1)(18){$8$} ;
\end{tikzpicture}
\caption{$B$-Graph Structure and Ordered Graph Structure}
\label{fi_B_ordered_graph_theta}
\end{figure}

The number of SNBC walks of length $k$ in a fixed theta
graph $S$ is a far more complicated function of $k$ than in the case of
cycles and walks of order zero.  For example,
consider a walk $w_G$ whose visited subgraph, $S=\ViSu(w_G)$,
is determined by its first 11 ``steps''
$$
(u,i_0), (e_1,i_0), (u,i_1), \ldots, (e_{11},i_{10}), (u,i_{11})
$$
and
is depicted in Figure~\ref{fi_sample_w_G_theta}:
our conventions are that 
repeated vertices are indicated in red, and we
assume that these first 11 steps encounter every vertex and every 
edge---in at least one of its orientations---of $S=\ViSu(w_G)$.
Let us make some observations about such a walk, $w_G$, and about
$S=\ViSu(w_G)$ used in \cite{friedman_random_graphs}.

First, fix $i_0,\ldots,i_{11}\in [n]$, subject to the
constraints in Figure~\ref{fi_sample_w_G_theta}
($i_2=i_7$, $i_1=i_8$, etc.), and consider the event that this
graph occurs (i.e., exists)
as a subgraph of $G\in\cC_n(B)$: in terms of the
permutation assignment $\sigma$, this graph requires that
$\sigma(f_1)$ take on $5$ specified values:
$\sigma(f_1)i_j=i_{j+1}$ for $j=0,1,3$ and
$\sigma(\iota_B f_1)i_j=(\sigma(f_1))^{-1}i_j=i_{j+1}$ for $j=6,9$.
Similarly this graph requires that $\sigma(f_2)$ on $4$ specified values;
from this it is not hard to see that
the probability that $\sigma$ satisfies these constraints is
$$
\prod_{j=1}^{d/2} \frac{1}{n(n-1)\ldots(n-a(f_j)+1)}
$$
where $a(f_1)=5,a(f_2)=4$ count the number of times that $f_1,f_2$ occur
in $S$.
Hence for any given $S$, we can define $a(f_j)=a_S(f_j)$ to be 
the number of times that $f_j$ or $\iota_B f_j$ occurs in such a diagram,
and the probability that $S$ occurs is precisely
the product in \eqref{eq_exact_formula_order_zero}.
Similarly the number of ways of choosing $i_0,\ldots,i_{11}\in[n]$
subject to the indicated constraints is
$$
n(n-1)\ldots (n-b+1)
$$
where $b=b_S$ is the total number of vertices in $S$.  Hence the
expected number times $\sigma$ satisfies the above constraints
when $i_0,\ldots,i_{11}$ vary is 
\begin{equation}\label{eq_expected_S}
n(n-1)\ldots \bigl(n-b_S+1\bigr)
\prod_{j=1}^{d/2} \frac{1}{n(n-1)\ldots\bigl(n-a_S(f_j)+1\bigr)}
\end{equation} 
which is an analog of \eqref{eq_exact_formula_order_zero}
(with an additional factor of $n$ since $i_0$ is also varying).

\subsection{$B$-Graphs and Ordered Graphs}

We now make some observations regarding the formula \eqref{eq_expected_S} to
motivate the notions of a {\em $B$-graph} and an {\em ordered graph}:
\begin{enumerate}
\item this formula depends only on how $S$ maps to $B$ under
projection, which is depicted in the left diagram
of Figure~\ref{fi_B_ordered_graph_theta} (i.e., we omit
the $i_0,\ldots,i_{11}$);
\item this formula counts the $\cC_n(B)$ expected number of graphs, 
$S$, as above,
along with the order in which the vertices and edges are first encountered
along $w_G$,
and for each edge the orientation in which it is first encountered;
this order is depicted in the right diagram of
Figure~\ref{fi_B_ordered_graph_theta};
\item equivalently, this formula counts the $\cC_n(B)$
expected of walks that are {\em similar} to any fixed $w_G$ whose first 
$11$ steps are as indicated, where we call two walks similar if 
they differ by a renumbering of the $i_0,\ldots,i_{11}$ subject to
the imposed identities on these numbers (i.e., $i_7=i_2$, etc.).
\end{enumerate}

To express the above notions more precisely, we define:
\begin{enumerate}
\item a {\em $B$-graph} to be a graph, $S$, along with a map $S\to B$; and
\item an {\em ordered graph} to be a graph, $S$, along with (1) a total
ordering of its vertex set, (2) a total ordering of its edge set, and
(3) an orientation $\Eor_S\subset\Edir_S$; a walk, $w$ with
$S=\ViSu(w)$ induces its {\em first encountered} ordering on $S$ which
indicates the order in which the vertices, edges, and edge orientation
of each edge are first encountered along $S$.
(Some edges may be encountered in only one orientation.)
\end{enumerate}
We use the symbol $S_\Bg$ to denote a graph $S$ with a given 
$B$-graph structure,
i.e., a given morphism $S\to B$; we use $S^\og$ to denote a graph $S$
along with an ordering; and we use $S_\Bg^\og$ to denote a graph $S$
with both a $B$-graph structure and an ordering.
If $w_G$ is a walk in some $G\in\Coord_n(B)$, then
$w_G$ endows $S=\ViSu(w_G)$ with both a $B$-graph structure and an
ordering, and we sometimes write $\ViSu_\Bg^\og(w_G)$ to indicate
these structures.

We emphasize that if $S_\Bg^\og=\ViSu_\Bg^\og(w_G)$ for a walk 
in some $G\in\cC_n(B)$, then the ordering on $S$ is
important for a number of reasons:
\begin{enumerate}
\item
when we subdivide walks by homotopy type (below), the ordering is
useful to make sure that we are counting the number of walks correctly;
and
\item 
(more importantly) 
when $S_\Bg$ has non-trivial
automorphisms (as a $B$-graph),
then
\eqref{eq_expected_S} does not count the number of
subgraphs in $G\in\cC_n(B)$ that are isomorphic 
to $S_\Bg$ (as a $B$-graph).
\end{enumerate}
The second reason is subtle but extremely important:  if $S_\Bg^\og$
is an ordered $B$-graph, and $G_\Bg$ is a $B$-graph, we use the notation
\begin{enumerate}
\item
$[S_\Bg]$ for the class of $B$-graphs isomorphic to $S_\Bg$ (where
we forget its ordering), and $[S_\Bg]\cap G$ for the set of 
$B$-subgraphs $H_\Bg\subset G_\Bg$ with $H_\Bg\in[S_\Bg]$; and
\item
$[S_\Bg^\og]$ for the class of ordered $B$-graphs isomorphic to $S_\Bg^\og$
(as ordered $B$-graphs), and $[S_\Bg^\og]\cap G$ for the set of
$H_\Bg^\og\in [S_\Bg^\og]$ such that (when we forget the ordering on
$H_\Bg^\og$) $H_\Bg\subset G_\Bg$.
\end{enumerate}
Then for any $G_\Bg$ we easily see that
\begin{equation}\label{eq_ordered_vs_subgraphs}
\# [S_\Bg^\og]\cap G_\Bg = 
\bigl(\# [S_\Bg]\cap G_\Bg\bigr)\, \bigl(\# {\rm Aut}(S_\Bg) \bigr),
\end{equation} 
where ${\rm Aut}(S_\Bg)$ denotes the set of $B$-graph automorphisms of
$S_\Bg$; \eqref{eq_expected_S} equals
\begin{equation}\label{eq_expected_ordered}
\EE_{G\in\cC_n(B)}\bigl[ \# [S_\Bg^\og]\cap G_\Bg \bigr] 
=
\bigl(\# {\rm Aut}(S_\Bg) \bigr)
\EE_{G\in\cC_n(B)}\bigl[ \# [S_\Bg]\cap G_\Bg \bigr]  .
\end{equation} 
In Subsection~\ref{su_non_triv_aut} we give examples of $S_\Bg$ with
non-trivial automorphisms and comment more
on this formula.
Even for a fixed $S$, $\#{\rm Aut}(S_\Bg)$ generally depends on the
$B$-graph structure of $S$; 
for this reason, our trace methods---e.g., the definition of an
{\em algebraic} model---involve
$[S_\Bg^\og]\cap G$ rather than $[S_\Bg]\cap G$.

We remark that \eqref{eq_ordered_vs_subgraphs} implies that
$[S_\Bg^\og]\cap S_\Bg$ and ${\rm Aut}(S_\Bg)$ have the same size;
this formula, along with examples in Subsection~\ref{su_non_triv_aut},
may provide helpful intuition.


\subsection{Homotopy Type and Expected Walk Formulas in the Theta Graph
Example}

Let us return to the example depicted in
Figures~\ref{fi_sample_w_G_theta} and
\ref{fi_B_ordered_graph_theta}
of Subsection~\ref{su_theta_graph_example}.
We wish to define its {\em homotopy type} and a ``formula'' for the
number of SNBC walks of length $k$ whose ordered visited subgraph
is all of this graph, in the ordering depicted.
It is helpful to keep
in mind that eventually we will sum over all such $i_0,i_1,\ldots,i_{11}$
in $[n]$ (subject to $i_0=i_8$ and the other constraints), multiply
by the probability that such a graph is a subgraph of an
element of $\cC_n(B)$, and
ultimately we will use this---the
terminology and formulas---to estimate
\begin{equation}\label{eq_expected_theta_walks}
\EE_{G\in\cC_n(B)}[{\rm theta}(G,k)] ;
\end{equation} 
we will need to overcome several new difficulties that do not appear with
walks of order $0$, whose visited subgraphs are cycles.

So let $w_G$ be any walk in some $G\in\Coord_n(B)$ whose first 
$11$ steps are as in Figure~\ref{fi_sample_w_G_theta}, and such that
these $11$ steps already cover all of $S_\Bg^\og=\ViSu_\Bg^\og(w_G)$.
From the $B$-graph and ordered graph structure of $S$ depicted in
Figure~\ref{fi_B_ordered_graph_theta} we extract its 
{\em homotopy type} which is the ordered graph depicted on the
left side of Figure~\ref{fi_homotopy_walk}.

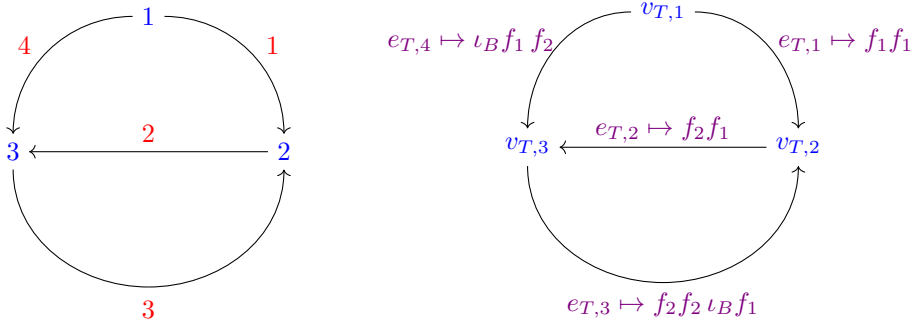
\begin{figure}
\begin{tikzpicture}[scale=0.6]
\node [blue] at (0,3)(0){$1$};
\node [blue] at (3,0)(2){$2$};
\node [blue] at (-3,0)(4){$3$};
\draw[->] (0) to [out=0,in=90] (2) ;
\draw[->] (2) -- (4) ;
\draw (4) to [out=-90,in=-180] (0,-3) ;
\draw[->] (0,-3) to [out=0,in=-90] (2) ;
\draw[->] (0) to [out=180,in=90] (4) ;
%
\node [red] at (2.75,2.35) {$1$} ;
\node [red] at (0,0.4) {$2$} ;
\node [red] at (0,-3.5) {$3$} ;
\node [red] at (-2.75,2.35) {$4$} ;
\end{tikzpicture}\quad\quad\quad%
\begin{tikzpicture}[scale=0.6]
\node [blue] at (0,3)(0){$v_{T,1}$};
\node [blue] at (3,0)(2){$v_{T,2}$};
\node [blue] at (-3,0)(4){$v_{T,3}$};
\draw[->] (0) to [out=0,in=90] (2) ;
\draw[->] (2) -- (4) ;
\draw (4) to [out=-90,in=-180] (0,-3) ;
\draw[->] (0,-3) to [out=0,in=-90] (2) ;
\draw[->] (0) to [out=180,in=90] (4) ;
\node [violet] at (4.05,2.35) {$e_{T,1}\mapsto f_1f_1$} ;
\node [violet] at (0,0.4) {$e_{T,2}\mapsto f_2f_1$} ;
\node [violet] at (0,-3.5) {$e_{T,3}\mapsto f_2 f_2 \,\iota_B f_1$} ;
\node [violet] at (-4.25,2.35) {$e_{T,4}\mapsto \iota_B f_1\,f_2$} ;
\end{tikzpicture}
\caption{A Homotopy Type, $T^\og$, and a Wording $\Edir_T\to (\Edir_B)^*$}
\label{fi_homotopy_walk}
\end{figure}

Here the {\em homotopy type} refers to the ordered graph $T^\og$
obtained by suppressing all the beads of
$S$ except its first vertex (if it is a bead); each directed edge of
$T$ correspond to a beaded path in $S$; the ordering on $S$
gives rise to a natural ordering on $T$.
The homotopy type in Figure~\ref{fi_homotopy_walk} has four vertices
and four edges that come oriented (eight directed edges) labeled there as
$$
V_T = \{ v_{T,1},v_{T,2},v_{T,3} \},  \quad
\Eor_T = \{ e_{T,1} , \ldots, e_{T,4} \} \subset 
\Edir_T = \{ e_{T,1},\iota_T e_{T,1},\ldots \iota_T e_{T,4} \} .
$$
To reconstruct $S_\Bg^\og$ from $T^\og$ we simply need to know
the walk in $B$ corresponding to directed edge in $T$;
since this walk is determined by its sequence of directed edges,
we define a {\em wording} on $T^\og$ to be this map
$$
W\from\Edir_T\to(\Edir_B)^* ,
$$
where $(\Edir_B)^*$ denotes the
set of finite sequences of elements of $\Edir_B$
(i.e., the set of words over the alphabet $\Edir_B$ is the sense
formal or regular language theory).
The right diagram of Figure~\ref{fi_homotopy_walk} depicts the wording
corresponding to the example of 
Figures~\ref{fi_sample_w_G_theta} and~\ref{fi_B_ordered_graph_theta};
in this diagram $W(\iota_T e_{T,j})$ is implicit, as it must be
the reverse walk of $W(e_{T,j})$.

Abstractly we say that a map
$W\from\Edir_T\to(\Edir_B)^*$ is a {\em wording} if
(1) for each $e\in\Edir_T$,
$W(e)$ is a non-backtracking walk and
$W(\iota_T e)$ is the reverse walk of $W(e)$, and (2)
for each $v\in V_T$, the first letter of each $W(e)$ with
$t_T e=v$ must have the same tail.
In this case, $W$ determines a $B$-graph (unique up to isomorphism)
that we denote by
$\VLG(T^\og,W)$ 
(in analogy with variable-length graphs,
see Article~III and Lemma~9.2 of \cite{friedman_alon}).

We emphasize that the above notion of homotopy type, $T^\og$, can
be defined for any SNBC walk, $w_G$, in any $G\in\Coord_n(B)$,
but this homotopy type depends only on $S_\Bg^\og=\ViSu_\Bg^\og(w_G)$.

The estimate of \eqref{eq_expected_theta_walks} in
\cite{friedman_random_graphs} is based on the following observations.
First if $S_\Bg^\og$ is any ordered $B$-graph whose homotopy
type is $T^\og$ above,
then the length of any walk $w_G$ whose
visited subgraph is $S_\Bg^\og$ equals
$$
k=k_1m_1+\cdots+k_4m_4 
$$
where $k_i$ is the length of the path in $S_\Bg^\og$ corresponding to
the $i$-th oriented edge of 
$T$, and $m_i$ is the number times that this edge is traversed in
$w_G$ (as an edge, i.e., in either orientation); 
we will use $\mec k$ as shorthand for the vectors whose components
are the $k_i$, which can also be regarded a vector
$E_T\to\naturals$ indexed on $E_T$, 
and similarly for $\mec m$, whereupon the length, $k$,
of the walk equals
$$
\mec k\cdot \mec m \eqdef k_1m_1+\cdots+k_4m_4 .
$$
Furthermore, for fixed $S_\Bg^\og$ the number of
walks with these values of $m_1,\ldots,m_4$ is a function
$g=g(m_1,\ldots,m_4)$ depending only on the $m_j$ and not on
the wording.  
We then (fairly easily) prove that the total number of SNBC walks of length
$k$ and homotopy
type $T^\og$ is
\begin{equation}\label{eq_dot_conv_first_form}
\sum_{W\from \Edir_T\to (\Edir_B)^*} 
\quad\sum_{\mec k_W \cdot \mec m = k}  
f(W,n)
g(\mec m)
\end{equation} 
where (1) $g$ is as above, (2) $f(W,n)$ denotes
\eqref{eq_expected_S} ($b_S$ and the $a_S(f_j)$ can be inferred from
$W$),
(3) $\mec k_W$ are the lengths of the paths
corresponding to the directed edges of $T$
(which depend only on $W$),
and (4) the sum is over all wordings $W$ that represent wordings
that are {\em realizable}, i.e., are wordings of the visited subgraph
of some SNBC walk in $G\in\cC_n(B)$; let us discuss the notion
of realizability a bit further.

\subsection{Realizable Wordings, Letterings, and Regular Languages}
\label{su_realizable_letterings_regular}

To analyze \eqref{eq_dot_conv_first_form} we need to study the set
of wordings $W\from\Edir_T\to(\Edir_B)^*$ that can arise from 
graphs $S_\Bg^\og=\ViSu_\Bg^\og(w_G)$.
Aside from the conditions on $W$ in the definition of wording above,
if $S_\Bg\subset G_\Bg$ with $G_\Bg\in\Coord_n(B)$ then
$S_\Bg$ must be \'etale; in the permutation model this turns out to be
sufficient for \eqref{eq_expected_S} to hold.
Furthermore,
the condition for $\VLG(T^\og,W)$ to be \'etale is precisely that
for each $v\in V_T$, the first letters of the $W(e)$ with $t_T e=v$
are distinct.
It turns out to be convenient \cite{friedman_random_graphs,friedman_alon}
to define the {\em lettering} of $W$ to be the information consisting
of the first letter of each $W(e)$ (which through $W(\iota_T e)$
also specifies its last letter).
These articles sum over \eqref{eq_dot_conv_first_form} by fixing
all letterings that give rise to \'etale graphs, and summing over 
such $W$.  The advantage is that the set of all such $W$
of a fixed lettering
is a direct product indexed over any orientation $\Eor_B$ of $B$
of regular languages
$$
{\rm NBWALK}(B,e,e') 
$$
with $e,e'\in\Edir_B$,
defined to be the set of non-backtracking walks in $B$ whose
first directed edge is $e$ and whose last is $e'$.
One can then prove asymptotic expansions for 
\eqref{eq_expected_theta_walks} by summing over all fixed letterings
of the wordings $W$ determined by
$$
\prod_{e\in\Eor_T} {\rm NBWALK}\bigl(B,{\rm first}(e),{\rm last}(e)\bigr)
$$
using the techniques of Chapter~2 of \cite{friedman_random_graphs},
where ${\rm first}(e),{\rm last}(e)$ are those letters, i.e.,
elements of $\Edir_B$, specified by the lettering.

It turns out that similar asymptotic expansion theorems holds whenever
one sums \eqref{eq_expected_S} over the wordings determined by any direct
product
\begin{equation}\label{eq_regular_product}
\prod_{e\in\Eor_T} \cR(e)
\end{equation} 
where for each $e\in\Eor_T$, $\cR(e)$ is a regular language.
This turns out to be useful in defining algebraic models:
the main property of algebraic models is that, roughly speaking,
for each ordered
$B$-graph, $S_\Bg^\og$, we require that 
$$
\EE_{G\in\cC_n(B)}\bigl[ \# [S_\Bg^\og]\cap G \bigr] 
$$
has an asymptotic expansion
$$
c_0(S_\Bg) + c_1(S_\Bg)/n + \cdots + c_{r-1}(S_\Bg)/n^{r-1} +
O(1) c_r(S_\Bg)/n^r,
$$
where the $c_i$ are polynomials in the variables $\mec a=\mec a(S_\Bg)$
which count the number of times each edge of $E_B$ appears in
$S_\Bg$;
by \eqref{eq_ordered_vs_subgraphs},
the $c_i$ don't depend on the ordered graph structure of
$S_\Bg^\og$, so we write $c_i(S_\Bg)$ instead of $c_i(S_\Bg^\og)$.
However, for the permutation model,
$c_0=1$ for $S_\Bg$ where $S$ is a cycle, but
$c_0=0$ and $c_1=1$ when $S$ is a theta graph; hence
we see that the polynomials
must depend---at the very least---on the order of $S$.
Although for the permutation model these polynomials depend only 
on the order of $S$,
for other of our basic models---including the cyclic model---the 
polynomials
for the $c_i=c_i(S_\Bg)$ not only depend on the order and homotopy type 
$S$, but also also depend on a finite number of other features.
However, as long as these features can be expressed by sets
of the form \eqref{eq_regular_product}, then all of our main
theorems hold.
In other words, our main results will still hold provided that
for each homotopy type, $T^\og$, we can partition all wordings
$W\from\Edir_T\to(\Edir_B)^*$ into a finite number of sets, each of which
is a product \eqref{eq_regular_product} for some $\cR$, such that
some polynomial of $\mec a=\mec a_{S_\Bg}$ gives each $c_i(S_\Bg)$
for all the $S_\Bg$ determined by such wordings.

\subsection{Examples of $B$-graphs with Non-trivial Automorphisms}
\label{su_non_triv_aut}

\begin{example}
Let $d=2$, so that $\sigma$ represents a single permutation
$\sigma(f_1)$.  For $i_0,k\in[n]$, 
we easily see that
the probability that 
$i_0$ lies on a cycle of length $k$ is exactly $1/n$\footnote{
Indeed, setting $i_j\ne\sigma(f_1)i_{j-1}$ for all $j\in [k]$, 
we have
$i_1\ne i_0$ with probability $(n-1)/n$, and given this
the probability that $i_2$ is distinct from $i_0,i_1$ is
$(n-2)/(n-1)$, etc.,
and given all these events, $i_k=i_0$ with probability $1/(n-k)$.
 }.
So the expected number of cycles of length $k$ is $1/k$, but the
expected number of walks $w_G$ of the form
$$
(u,i_0), (f_1,i_1), (u,i_1), \ldots, (f_1,i_{k-1}),(u,i_k)
$$
with $i_0,\ldots,i_{k-1}$ distinct and $i_k=i_0$ is $1$.
This difference from $1/k$ to $1$ occurs because the automorphism
group of the cycle with a given $B$-graph structure is of order $k$.
Moreover, if we forget the $B$-graph structure, the number of
automorphisms of a cycle of length $k$, as a graph, is $2k$.
\end{example}

\begin{example}
Consider a $B$-graph with two vertices $v_1,v_2$, joined by
two edges from $v_1$ to $v_2$, one labeled $f_1$, the other
$\iota_B f_1$, plus one whole-loop at $v_1$ and
one at $v_2$.  This graph has order $2$
and has one non-trivial automorphism iff the whole-loops have the
same label.
Similarly if the whole-loops are replaced by beaded paths in which
each directed edge is labeled.
\end{example}

\begin{example}
Say that we allow $B$ to have half-loops, then say $f_1$ is a half-loop
and $f_2$ a whole loop.  Consider the $B$-graph $S_\Bg$ where $S$
is a barbell graph, where the ``bar'' is a 
single edge over $f_1$, and the other edges are labeled with
$f_2,\iota_B f_2$.  Then $S_\Bg$ has a non-trivial automorphism iff
the two loops have the same length.
More generally, a $B$-graph $S_\Bg$ where $S$ is a barbell graph
can have non-trivial automorphisms depending on the lengths of its loops
and how we label its edges.
Keeping track of this automorphism group would make our methods 
significantly more tedious.
\end{example}

We remark that any automorphism
of a graph $S_\Bg$ must take vertices of degree three or more
to 
themselves, and if each vertex of $S$ has degree at least two, then
then number of such vertices is at least one and at most $2\ord(S)$
unless $S$ is a cycle.
It easily follows that the size of the automorphism group is bounded
as a function of $\ord(S)$ if $\ord(S)\ge 1$;
this contrasts the case where $S$ is a cycle, 
where the size of the automorphism
group can be any number dividing the length of the cycle, depending
on the structure map $S\to B$.

\section{$B$-Graphs, Orderings, and Strongly Algebraic Models}
\label{se_ordered_B_strong_alg}

In this section we define the notion of an {\em strongly algebraic model}
which is a special case of {\em algebraic models} that are easier
to describe.  The permutation models and the permutation-involution models
of even degree are examples of strongly algebraic models.

%
In order to define the term {\em strongly algebraic}
we introduce some terminology fundamental to our trace methods, such
as {\em $B$-graphs} and {\em ordered graphs}.
This terminology is illustrated in Section~\ref{se_bsf}, 
specifically the discussion regarding
Figure~\ref{fi_B_ordered_graph_theta}.

\subsection{$B$-Graphs}

\begin{definition}\label{de_B_graph}
Let $B$ be a graph.  By a {\em $B$-graph} we mean a graph, $G$, endowed
with a morphism
$\phi\from G\to B$; we typically write $G_\Bg$ for such a structure,
with $\phi$ understood.
We say that $\phi$ or $G_\Bg$ is an {\em \'etale} (respectively, {\em covering})
$B$-graph if $\phi$ is an \'etale (respectively, covering) morphism.
\end{definition}

A $B$-graph can therefore be viewed as a morphism $G\to B$, although the
concepts regarding $B$-graphs are usually understood as working with $G$
along with an underlying map $G\to B$.

\begin{definition}\label{de_B_graph_morphisms}
By a {\em morphism} of $B$-graphs,
from $\phi\from G\to B$ to $\phi'\from G'\to B$, or 
$G_\Bg\to G'_\Bg$, we mean a morphism
of the sources, i.e., $\nu\from G\to G'$, that respects the
$B$-structure in the evident sense, i.e., $\phi=\phi'\nu$.
\end{definition}
In the literature,
$B$-graphs are often called {\em graphs over $B$}, and this construction
is referred to as a
{\em slice category}.

We will use some common nomenclature regarding $B$-graphs.
If $G$ is a $B$-graph with $G\to B$ understood, we will speak of the
{\em $B$-graph structure on $G$} as the morphism $G\to B$.  Similarly,
if $G$ is a graph, to {\em endow $G$ with the
structure of a $B$-graph} means to specify a morphism $G\to B$.
If $\pi\from G\to B$ is a $B$-graph and $v\in V_B$, then the 
{\em vertex fibre of $v$
(in $G$)} refers to $\pi^{-1}(v)$ (more precisely $\pi_V^{-1}(v)$);
similarly for directed edge fibres and edge fibres.

\begin{example}
For any graph $B$ and $n\in\naturals$, any $G\in\Coord_n(B)$ comes
with its projection $\pi\from G\to B$,
which is its projection ``onto the first component,'' since $V_G=V_B\times[n]$
and $\Edir_G = \Edir_B\times [n]$
in view of \eqref{eq_coordinatized_digraph}.
We easily see that
$G$ is a covering $B$-graph; any subgraph of $G$ inherits a 
$B$-structure from $G$, and making it necessarily
an \'etale $B$-subgraph.
Each vertex fibre and directed edge fibre of $G$ is identified with $[n]$
(by its projection onto the second component, in 
\eqref{eq_coordinatized_digraph}).
\end{example}

\subsection{Ordered Graphs}

\begin{definition}\label{de_ordered_graph}
By an {\em ordered graph} we mean a $G^{\le}=(G,\Eor_G,\le_V,\le_E)$ where $G$
is a graph, $\Eor_G\subset\Edir_G$ is an orientation of $G$, $\le_V$ is a
total ordering of $V_G$, and $\le_E$ is a total ordering
of $\Eor_G$; we sometimes simply
write $G$ if the orientation and two orderings are understood.
By a {\em morphism} of ordered graphs $G_1^\og\to G_2^\og$
we mean a morphism of
graphs that preserves the orientations and two orderings in the evident sense.
\end{definition}
An ordered graph $G^\og$
has no non-trivial automorphisms, since
such a morphism would have to (1) be the identity on $V_G$ (since it
preserves the vertex ordering), (2) take
$\Eor_G$ to itself, (3) be the identity on $\Eor_G$
(by order preservation), and (4) therefore be the identity map on $\Edir_G$.
It follows that there is at most one isomorphism from one ordered graph to
another.

\subsection{First-Encountered Ordering}

\begin{definition}\label{de_first_encountered}
Let $w=(v_0,\ldots,e_k,v_k)$ be a walk in a graph, $G$,
and let $S=\ViSu_G(w)$.  By the
{\em $w$-first-encountered ordering of $S$} we mean the ordering on
the vertices of $S$ in the order in which they occur first in the
sequence $v_0,\ldots,v_k$, the orientation $\Eor_S\subset\Edir_S$ 
which gives the orientation in which an edge first occurs in the
sequence, $e_1,\ldots,e_k$, and the ordering on $E_S$ in the order
in which they first occur---in either orientation---in the sequence
$e_1,\ldots,e_k$.
We use $\ViSu^\og_G(w)$ to denote $S$ with this ordering.
\end{definition}

\subsection{The Fibre Counting Functions $\mec a,\mec b$}

\begin{definition}\label{de_fibre_counting}
If $\phi\from S\to B$ is a $B$-graph, the 
{\em vertex-fibre counting vector} of $\phi$ or $G_\Bg$ is the vector
$\mec b=\mec b_{S_\Bg}\from V_B\to\integers_{\ge 0}$ given by
$$
b(v) \eqdef \# \phi^{-1}_V(v);
$$
one similarly defines the {\em directed-edge-fibre counting vector}
$\mec a=\mec a_{S_\Bg}\from \Edir_B\to\integers_{\ge 0}$ defined by
$$
a(e) \eqdef \# \phi^{-1}_E(e) ;
$$
it follows that $a(\iota_B e)=a(e)$ for any $e\in\Edir_B$, and
when convenient we may view $\mec a$ as a function
$E_B\to\integers_{\ge 0}$ whose value on an $\{e,\iota_B e\}$
is $a(\iota_B e)=a(e)$.
If $w$ is a walk in a $B$-graph, then we similarly define
$\mec a=\mec a_w$ (and $\mec b=\mec b_w$)
as $\mec a_{S_\Bg}$ (and $\mec b_{S_\Bg}$) where
$S_\Bg=\ViSu_\Bg(w)$.
\end{definition}

The vectors $\mec a,\mec b$ are crucial to our notion of {\em algebraic
models}: see 
\eqref{eq_permutation_model_exact_form},
\eqref{eq_strongly_algebraic}, and \eqref{eq_algebraic}
below, and \eqref{eq_exact_formula_order_zero} 
and \eqref{eq_expected_S} in Section~\ref{se_bsf} (where $\#V_B=1$).

\subsection{Ordered $B$-Graphs and Strongly Algebraic Models}

In this section we describe most of the conditions for a model to
be {\em strongly algebraic}; this involves ordered $B$-graphs.

\begin{definition}\label{de_ordered_B_graph}
By an {\em ordered $B$-graph}, $G_\Bg^\og$, we mean a graph, $G$, which is
endowed with an ordering and the structure of a $B$-graph.
By a {\em morphism} $G_\Bg^\og\to H_\Bg^\og$ we mean a morphism of
graphs $G\to H$ which respects the $B$-structure and the ordering.
\end{definition}

\begin{example}
If $w$ is a walk in a $B$-graph $G_\Bg$, then its visited subgraph,
$\ViSu(w)$,
is naturally endowed with the structure of a $B$-graph and of an
ordered graph; we sometimes write $\ViSu_\Bg^\og(w)$ to emphasize
these structures.
\end{example}

\begin{definition}\label{de_subgraph_counting}
If $S_\Bg^\og$ is an ordered $B$-graph, we use 
$[S_\Bg^\og]$ to describe the class of all ordered $B$-graphs isomorphic
to $S_\Bg^\og$ (as ordered $B$-graphs).  If $G_\Bg$ is another $B$-graph,
we use $[S_\Bg^\og]\cap G_\Bg$ to denote the set of ordered $B$-graphs,
$U_\Bg^\og\in [S_\Bg^\og]$ such that $U_\Bg$ is a $B$-subgraph of $G_\Bg$.
\end{definition}

If $[S_\Bg]$ refers to the isomorphism class of $S$ as a $B$-graph, with
its ordering ignored, and $[S_\Bg]\cap G_\Bg$ is the set of subgraphs of
$G_\Bg$ isomorphic to $S_\Bg$ as $B$-graphs, then
it is easy to see that
\begin{equation}\label{eq_automorphisms}
\#[S_\Bg^\og]\cap G_\Bg = \bigl( \#{\rm Aut}(S_\Bg)\bigr)
\bigl( \#[S_\Bg]\cap G_\Bg \bigr)
\end{equation} 
where ${\rm Aut}(S_\Bg)$ is the number of automorphisms of $S_\Bg$
(as a $B$-graph).  
However, $\#[S_\Bg^\og]\cap G_\Bg$ are better adapted to developing
asymptotic expansions, as the following example illustrates:
if $B$ has no half-loops and $\cC_n(B)$ is the permutation model,
then for any \'etale $B$-graph, $S_\Bg$, and any ordering on $S_\Bg$ we have
\begin{equation}\label{eq_ordered_inclusion_expansion}
\EE_{G\in\cC_n(B)}[ \#[S_\Bg^\og]\cap G_\Bg]
=
n^{-\ord(S)}\bigl(1+c_1/n+\cdots+c_{r-1}/n^{r-1}+O(1/n^r)\bigr)
\end{equation} 
for large $n$,
where the $c_i$ are universal polynomials in $\mec a=\mec a_{S_\Bg}$
and $\mec b=\mec b_{S_\Bg}$ of degree $2i$, otherwise independent (!)
of $S_\Bg$;
indeed, similar reasoning as in \eqref{eq_expected_S} shows that 
the right-hand-side of 
\eqref{eq_ordered_inclusion_expansion} equals
\begin{equation}\label{eq_permutation_model_exact_form}
\prod_{v\in V_B} \Bigl( n(n-1)\ldots\bigl(n-b_{S_\Bg}(v)+1\bigr) \Bigr)
\prod_{e\in E_B} \frac{1}{n(n-1)\ldots\bigl(n-a_S(e)+1\bigr)} ;
\end{equation} 
now we expand this in an asymptotic series in $1/n$ (as in
the bottom of page~336 of \cite{friedman_random_graphs}, or extending
the derivation of 
\eqref{eq_cycle_k_refined} 
and \eqref{eq_next_coef_order_zero}.

\begin{definition}\label{de_pruned}
We say that a graph $S$ is {\em pruned} if each vertex of $S$ has
degree at least two.
\end{definition}
We easily see that if $w$ is an SNBC walk of positive length 
in some graph, $G$, then $\ViSu_G(w)$ is pruned: otherwise some $v\in V_G$
is isolated, or incident upon exactly one half-loop or exactly one
edge that isn't a self-loop, and we easily check
that no SNBC walk can pass through $v$
(note that if $e$ is a half-loop about $v$, then $v,e,\ldots,e,v$
is not SNBC, since $\iota_G e=e$).

\begin{definition}\label{de_strongly_algebraic}
Let $B$ be a graph and 
$\{\cC_n(B)\}_{n\in N}$ be a model over $B$.
We say that a $B$-graph $S_\Bg$ {\em occurs} in $\{\cC_n(B)\}_{n\in N}$
if for all sufficiently large $n\in N$
there is a $G\in\cC_n(B)$, such that $G_\Bg$ has a $B$-subgraph
isomorphic to $S_\Bg$.
We say that the family
of probability spaces $\{\cC_n(B)\}_{n\in N}$ is
{\em strongly algebraic} provided that
\begin{enumerate}
\item for each $r\in\naturals$
there is a function, $g=g(k)$, of growth $\mu_1(B)$
such that if $k\le n/4$ we have
\begin{equation}\label{eq_algebraic_order_bound}
\EE_{G\in\cC_n(B)}[ \snbc_{\ge r}(G,k)] \le  
g(k)/n^r
\end{equation} 
(recall Definition~\ref{de_order} for $\snbc_{\ge r}$); 
\item
for any $r$ there exists 
a function $g$ of growth $1$ and real $C>0$ such that the following
holds:
for any ordered $B$-graph, $S_\Bg^\og$, that is pruned and of
order less than $r$,
\begin{enumerate}
\item
if $S_\Bg$ occurs in $\cC_n(B)$, then for
$1\le \#\Edir_S\le n^{1/2}/C$,
\begin{equation}\label{eq_expansion_S}
\EE_{G\in\cC_n(B)}\Bigl[ \#\bigl([S_\Bg^\og]\cap G\bigr) \Bigr]
=
c_0 + \cdots + c_{r-1}/n^{r-1}
+ O(1) g(\# E_S) /n^r
\end{equation} 
where the $O(1)$ term is bounded in absolute value by $C$
(and therefore independent of $n$ and $S_\Bg$), and
where $c_i=c_i(S_\Bg)\in\reals$ such that
$c_i$ is $0$ if $i<\ord(S)$ and $c_i>0$ for $i=\ord(S)$;
and
\item
if $S_\Bg$ does not occur in $\cC_n(B)$, then for any
$n$ with $\#\Edir_S\le n^{1/2}/C$,
\begin{equation}\label{eq_zero_S_in_G}
\EE_{G\in\cC_n(B)}\Bigl[ \#\bigl([S_\Bg^\og]\cap G\bigr) \Bigr]
= 0 
\end{equation} 
(or, equivalently, no graph in $\cC_n(B)$ has a $B$-subgraph isomorphic to
$S_\Bg^\og$);
\end{enumerate}
\item
$c_0=c_0(S_\Bg)$ equals $1$ if $S$ is a cycle (i.e., $\ord(S)=0$ and
$S$ is connected) that occurs in $\cC_n(B)$;
\item
$S_\Bg$ occurs in $\cC_n(B)$ iff $S_\Bg$ is an \'etale $B$-graph
and $S$ has no half-loops; and
\item
there exist
polynomials $p_i=p_i(\mec a,\mec b)$ such that $p_0=1$
(i.e., identically 1), and for every
\'etale $B$-graph, $S_\Bg^\og$ we have that
\begin{equation}\label{eq_strongly_algebraic}
c_{\ord(S)+i}(S_\Bg) = p_i(\mec a_{S_\Bg},\mec b_{S_\Bg}) \ .
\end{equation}
\end{enumerate}
\end{definition}
We write $c_i(S_\Bg)$ rather than $c_i(S_\Bg^\og)$, and similarly
in \eqref{eq_strongly_algebraic}, since 
\eqref{eq_automorphisms} implies that
that the $c_i$ do not depend on the
ordering on $S_\Bg^\og$.

[Of course, if $S_\Bg$ does not occur in $\cC_n(B)$, then
\eqref{eq_zero_S_in_G} implies that
\eqref{eq_expansion_S} holds trivially, with all $c_i=0$;
however, it seems better pedagogically to separate the case of
$S_\Bg$ occurring and not occurring in the model.]

Note that if $B$ does not have half-loops, then if $S\to B$ is \'etale,
then $S$ has no half-loops; however, in $B$ does have half-loops,
then in our basic model of odd degree $n$,
elements of $G\in\cC_n(B)$ can have half-loops, and then
formulas for 
$$
\EE_{G\in\cC_n(B)}\Bigl[ \#\bigl([S_\Bg^\og]\cap G\bigr) \Bigr]
$$
in terms of $\mec a_{S_\Bg},\mec b_{S_\Bg}$
depend on the half-loops in $S$.
Hence the polynomials $p_i$ in
\eqref{eq_strongly_algebraic} must depend on the kind of
half-loops in $S_\Bg$;  such a model cannot be
strongly algebraic.

Notice that the condition on cycles, $S$, is implied by the condition
$p_0=1$; however, it is convenient to leave it there when we define
an {\em algebraic} model.
We also note that if $S=\ViSu(w)$ where $w$ is an SNBC walk, then
$S$ is necessarily connected (clearly) and if $\ord(S)=0$ then
$S$ must be a cycle (in view of \eqref{eq_order_local}, where
each half-loop is counted as contributing $2$ to the degree of $v$
because of the definition of order).

We remark that for our main theorems we only need
\eqref{eq_expansion_S} to hold in the range 
$1\le \#E_S\le h(n)$ for a function $h(n)$ asymptotically larger
than $\log n$; however, for all our basic models
\eqref{eq_expansion_S} holds in the larger range $1\le \#E_S\le n^{1/2}/C$.
Also, in all our basic models we can take $g(\# E_S)$ to be
merely $(\#E_S)^{2r+2}$ (making $g$ a polynomial, and hence of
growth $1$); this more precise bound is unimportant to us.
Finally we remark that in all our basic models,
\eqref{eq_expansion_S} holds for all $S$ without
isolated vertices, or all $S$ without conditions if we replace
$\#E_S$ by $\#E_S + \#V_S$; however, we only need
\eqref{eq_expansion_S} for graphs, $S$, each of whose vertex is of
degree at least two, and working with such $S$ simplifies some
later considerations.

In Article~V will prove that the permutation model and the 
permutation-involution model of even degree are algebraic, 
as an easy consequence of
of Lemmas~3.7---3.9 of \cite{friedman_random_graphs}.
However, in the permutation-involution model of odd degree, 
one has different polynomials $p_i$ in
\eqref{eq_strongly_algebraic} depending on how many half-loops
$S_\Bg^\og$ contains.
An {\em algebraic model} allows the polynomial to depend on some
finite amount of data that we call the {\em $B$-type} of $S_\Bg$;
we will define $B$-type in Section~\ref{se_new_algebraic}.

\section{The Homotopy Type of a Walk and VLG's (Variable-Length Graphs)}
\label{se_new_homot}

In this section we define the {\em homotopy type} of an SNBC walk
in a graph and {\em variable-length graphs (VLG's)},
similar to Section~3 of
\cite{friedman_alon}; see also Section~\ref{se_bsf} of this article.
The main theorems in Article~II,
used by Article~III, require these definitions.

In this section we make these notions precise, which are well known
but a bit tedious to spell out;
the essential ideas are illustrated by the 
examples in Section~\ref{se_bsf}.

\subsection{Bead Suppression}

\begin{definition}
\label{de_bead}
Let $S$ be a graph.
By a {\em bead of $S$} we mean a vertex of $S$
that is of degree two and not
incident upon a self-loop.  Let $V'\subset V_S$ be any subset consisting
entirely of beads;
by a {\em $V'$-beaded path} in $S$ we mean a non-backtracking walk
$(v_0,\ldots,e_k,v_k)$ in $S$ such that $v_0,v_k\notin V'$ but
$v_1,\ldots,v_k\in V'$.
\end{definition}
A beaded path is therefore a walk, involving directed edges; hence
its visited subgraph is a graph known as a ``path''
(i.e., a tree with two leaves and all other vertices of degree two).
However,
a beaded path---as opposed to a path---also specifies 
an orientation of the edges and a 
walk from one end to the other.

The idea of {\em homotopy type}---of a graph or an
ordered graph---is to classify them by
``suppressing'' as many beads as possible.  
(This is therefore a
refinement of the topological notion of homotopy type,
since any two connected graphs without half-loops and of the
same order are topologically homotopy equivalent).
However, this idea only works well on certain classes of graphs:
for example, if we work with connected graphs each of whose vertices
has degree at least two, then there are finitely many homotopy types
of graphs of a fixed order; see 
Subsection~\ref{su_homotopy_type_examples}.
However, if one allows graphs to have vertices of degree one,
then the number of homotopy types---defined by suppressing all the
beads of a graph---becomes infinite, 
even connected graphs of order $-1$, i.e., trees.
Similarly, the notion of the homotopy type of an ordered graph
works well on ordered graphs of the form $\ViSu^\og(w)$ when
$w$ is an SNBC walk in a graph (or a non-backtracking walk), but
not for general walks, $w$.
So some care must be taken when we ``suppress'' the beads of a graph
or of an ordered graph to define its homotopy type.
Here is an easy lemma in this direction.

\begin{lemma}\label{le_suppression_unique_diedge}
Let $S$ be a connected graph and $V'\subset V_S$ any set of beads
such that $V'$ is a proper subset of $V_S$ if $S$ is a cycle.
Then any $e\in \Edir_S$ lies on a unique $V'$-beaded path in $S$.
\end{lemma}
\begin{proof}
Set $e_1=e$, $v_1=h_S e$, $v_0=t_S e$.
If $v_1=h_S e$ lies in $V'$, then $v_1$ is of degree two and not
incident upon a self-loop; hence $e\ne \iota_S e$,
and $v_1$ is incident upon one edge
in $S$ other than $\{ e,\iota_S e \}$.  It follows that there
is a unique $e_2\in\Edir_S$ such that
$$
(v_0,e_1,v_1,e_2,h_S e_2)
$$
is non-backtracking.  Continuing in this manner, we construct
a unique non-backtracking walk
$$
(v_0,e_1,v_1,\ldots,e_m, v_m)
$$
for some $m\le \#V_S$
with $v_0,\ldots,v_{m-1}$ distinct elements of $V'$
and $v_m$ either equal to one of
$v_0,\ldots,v_{m-1}$ or $v_m\notin V'$.
We claim that (1) $v_m$ cannot equal any $v_i$ with $i<m$
(or otherwise $v_i$ is of degree greater than two),
(2) $v_m\ne v_0$ (or else $S$ is a cycle and $V'=V_S$).
Hence $v_m\notin V'$.  Similar we walk ``along $\iota_S e_1$'' to
construct a unique two-sided non-backtracking walk
$$
(v_{-m'},e_{-m'+1},v_{-m'+1},\ldots,v_0,e_1,\ldots,e_m,v_m)
$$
where $m'\ge 0$, $v_{-m'}\notin V'$, and all the $v_i$ with
$-m'\le i\le m-1$ distinct elements of $V'$.
This is clearly the unique $V'$-beaded path of $S$ containing $e$.
\end{proof}

This lemma motivates the following definition.

\begin{definition}\label{de_proper_bead_set}
Let $S$ be a graph.  We say that a subset $V'\subset V_S$ of beads
of $S$ is a {\em proper bead set of $S$}
if $V'$ does not contain all the vertices of any connected component
of $S$ that is a cycle.
\end{definition}

\begin{definition}\label{de_suppression}
Let $S$ be a graph, and $V'\subset V_S$ be a proper bead set of $S$.
We define the {\em suppression of $V'$ in $S$},
denoted $S/{V'}$, to be the graph, $T$, given as:
\begin{enumerate}
\item $V_T=V_S\setminus V'$ (i.e., the complement of $V'$ in $V_S$);
\item $\Edir_T$ is the set of $V'$-beaded paths in $S$;
\item for $e_T=(v_0,\ldots,e_k,v_k)\in \Edir_T$, we define its tail 
(i.e., $t_Te_T$) to be
$v_0$, its head (i.e., $h_Te_T$) to be $v_k$, and $\iota_Te_T$ to be
its {\em reverse walk}, i.e., $(v_k,\iota_S e_k,\ldots,\iota_S e_1,v_0)$.
\end{enumerate}
In addition, for $e_T=(v_0,\ldots,e_k,v_k)\in \Edir_T$, we define the
{\em length} of $e_T$ to be $k$;
since the lengths of $e_T$ and $\iota_T e_T$ are the same,
we define the {\em length} of an edge in $E_T$ to be the length of
an orientation of this edge.
\end{definition}

Notice that by definition each directed edge of $T=S/V'$ is
a walk in $S$; hence---for pedantic reasons---one can completely
reconstruct $S$ from $T=S/V'$: each directed edge
$e_T=(v_0,\ldots,e_k,v_k)$ comes paired with its inverse edge
$\iota_T e_T = (v_k,\iota_S e_k,\ldots, \iota_S e_1,v_0)$,
and so we can recover not only $V_S,\Edir_S,h_S,t_S$, but
the pairing 
allows us to determine how $\iota_S$ acts.

However, we define the {\em homotopy type} of a graph $S$
(or of an ordered graph) in
terms of the isomorphism class of $S/V'$ rather than $S/V'$ 
itself.  It becomes important to note that
we can reconstruct $S$ up to isomorphism (as a graph)
provided that we know a graph,
$T$, isomorphic to $S/V'$, and the function
$\mec k\from \Edir_T\to\naturals$ or $E_T\to\naturals$
that gives the length of directed edge or edge of $S/V'$ under
the isomorphism from $T$ to $S/V'$.

\begin{example}\label{ex_order_one}
The usual graph theoretic notion of the
{\em homotopy type} of a connected
graph $S$ is (any graph isomorphic to) the
suppression $S/V'$ of all beads of $V_S$,
except that
$V'$ omits one vertex if $S$ is a
cycle of length at least two.
For example, if $S$ is a connected graph of order $1$ without self-loops
and leaves (i.e., vertices of degree $1$), then
$S$ is of one of three homotopy types:
figure-eight, barbell, or theta (see, for example,
Figure~\ref{fi_m_equals_one} above or
Figure~6 in \cite{linial_puder}).
\end{example}

For an SNBC walk in a graph, we want to define a notion of
{\em homotopy type} that will ``remember'' its first encountered
ordering;
this will consist of the usual homotopy type of a graph, but
have some additional information that we now make precise.

\subsection{The Homotopy Type of a Non-Backtracking Walk and of its 
Ordered Visited Subgraph}

\begin{definition}\label{de_homotopy_walk}
Let $w$ be a non-backtracking walk
in some graph, $G$, and let $S^\og=\ViSu^\og(w)$.
By the {\em reduction} of $w$ is the ordered graph $R^\og$ where:
\begin{enumerate}
\item
$R=S/V'$, where $V'$ is the set of all beads of $w$ except the
first and last vertices of $w$ (if one or both of them are beads);
\item
the ordering $R^\og$ is given as follows:
\begin{enumerate}
\item the vertex ordering for $v_1,v_2\in V_R$ is $v_1<v_2$ iff
$v_1$ is encountered first before $v_2$ along $w$;
\item the orientation of $R$ are those $e\in \Edir_R$ whose 
corresponding beaded-path is encountered before the reverse beaded-path
along $w$;
\item the edge ordering is $e_1<e_2$ if the beaded-path
corresponding to the orientation of $e_1$ is encountered along $w$ before the
one corresponding to $e_2$.
\end{enumerate}
\end{enumerate}
We also write $S^\og/V'$ for $R^\og$ to emphasize the ordering.
By the {\em edge-lengths of $w$ on $R$} we mean the edge-lengths of $S/V'$.
We say that $w$ is {\em of homotopy type $T^\og$}
if $R^\og\isom T^\og$ (as ordered graphs); in this case the isomorphism
is unique, and the {\em edge-lengths of $w$ in $T^\og$} are the edge-lengths
$E_T\to\integers$ (and $\Edir_T\to\integers$) obtained from composing
this unique isomorphism with the edge-lengths on $R$.
\end{definition}

We easily see that in the above definition 
we can recover the ordering of $S^\og=\ViSu^\og_G(w)$ from the ordering
on the homotopy type of $w$: the point is that if a non-backtracking
walk encounters the first directed edge in a beaded path, then it must
immediately traverse the entire beaded path.
This is not true of general walks $w$, and the above definition
does not work well in this general case.

Notice that if in the above definition $w$ is an 
SNBC walk in some graph, then (1) the first and last vertices of $v$ 
are equal, and (2) each degree of a vertex in
$S=\ViSu(w)$ is at least two.
(If $w$ is merely closed and non-backtracking, then
property~(1) holds but not
generally property~(2).)
Since we are interested in SNBC walks, we will have properties (1)
and (2).
In Article~II it turns out to be convenient to know
that a graph, $G$, contains a 
$(\ge \nu,<r)$-tangle iff it contains such a tangle where each vertex
is of degree at least two (obtained by repeatedly
``pruning'' all leaves in the
tangle).

In the above definition we do not suppress the first and last 
vertices of $w$.
This implies that $V'$ is automatically a proper bead set
of $S$, which is convenient.  
However, the real reason we do not suppress the first and last 
vertices of $w$
(even if $S$ is not a cycle)
is that we need the first and last vertex to correctly
reconstruct the order $S^\og$
from the order on its homotopy type, $T^\og$
(which we cannot do if the walk does not begin and end on vertices 
in $T^\og$).

Clearly the reduction of $R^\og$ of a walk, $w$,
depends only on information that can be inferred from $\ViSu^\og(w)$;
this enables us to make the following definition.

\begin{definition}\label{de_homotopy_SNBC_visited}
If $S^\og$ is an ordered graph that is the visited subgraph of some 
non-backtracking
walk, $w$, on some graph, we define the {\em reduction}, {\em homotopy type},
and {\em edge-lengths} of $S^\og$ to be those of the walk $w$.
\end{definition}

Our trace methods will count the SNBC walks in a graph by dividing them
into
their homotopy types (as do 
\cite{broder,friedman_random_graphs,friedman_alon}).
Here are the particular counting functions.

\begin{definition}\label{de_snbc_walks_of_homotopy_type}
Let $G$ be a graph and $T^\og$ an ordered graph.  For $k\in\naturals$ 
and $\mec k$ functions $E_T\to\naturals$, we use
\begin{enumerate}
\item $\SNBC(T^\og; G,k)$ to denote the set of SNBC walks in $G$ of
length $k$ and homotopy type $T^\og$;
\item $\SNBC(T^\og,\mec k; G,k)$ to denote those elements of
$\SNBC(T^\og; G,k)$ whose edge-lengths in $T$ equal $\mec k$;
\item $\SNBC(T^\og,\ge\mec k; G,k)$ to denote
$$
\bigcup_{\mec k'\ge\mec k}\SNBC(T^\og,\mec k';G,k) 
$$
where $\mec k'\ge\mec k$ means $k'(e)\ge k(e)$ for all $e\in E_T$;
and
\item in the above, we replace $\SNBC$ with $\snbc$ to denote the
cardinality of such a set.
\end{enumerate}
\end{definition}
The sets $\SNBC(T^\og;G,k)$ and $\SNBC(T^\og,\mec k;G,k)$
are implicit in \cite{broder} and appear
explicitly in \cite{friedman_random_graphs};
however,
the sets $\SNBC(T^\og,\ge\mec k; G,k)$ are special to our
{\em certified traces}, a concept which is simpler than (and replaces) the
{\em selective traces} of \cite{friedman_alon}.

\subsection{Variable-Length Graphs (VLG's)}

Variable-length graphs (VLG's) is, in a sense, the opposite of
bead suppression.
They were introduced by Shannon \cite{shannon}
in the context of directed graphs
to model Morse Code; 
we will need VLGs 
(see Theorem~\ref{th_expansion_thm_for_types})
for the same reasons
they were needed in \cite{friedman_alon}, beginning in Section~3.4 there.

Informally, if $T$ is a graph and $\mec k\from E_T\to\naturals$ a 
function, then the variable-length graph $\VLG(T,\mec k)$ refers to
a graph (it is not unique)
obtained by replacing each edge $e\in E_T$ with a path of
length $k(e)$.
Here is a more precise definition.

\begin{definition}
\label{de_VLG}
Let $T$ be a graph and $\mec k\from E_T\to\naturals$ be a function
such that $k(e)=1$ whenever $e$ is a half-loop.
By a
{\em variable-length graph (VLG) on $T$ with edge-lengths $\mec k$}, denoted
$\VLG(T,\mec k)$, we mean any graph $S$ such that
for some subset of beads, $V'\subset V_S$ we have
(1) $T$ is isomorphic to $S/V'$, (2) the edge-lengths of $S/V'$
are $\mec k$ (under this isomorphism), and
(3) $V'$ omits at least vertex in
each connected component of $S$.
\end{definition}
It is immediate that $S=\VLG(T,\mec k)$ always exists: we form $S$ from
$T$, by taking each
$e\in E_T$
such that $k(e)>1$ and replacing $e$
with $k(e)$ edges that form 
an (undirected) path of length $k(e)$, which in the process
introduces $k(e)-1$ new vertices
(all of which become beads in $S$).
Sometimes VLG's are defined as pairs $T,\mec k$ as above, and
$\VLG(T,\mec k)$ is defined as the {\em realization} of $(T,\mec k)$.
Of course, $\VLG(T,\mec k)$ is only defined up to isomorphism.
We illustrate this construction in Figure~\ref{fi_vlg_example}.

\begin{figure}
\begin{tikzpicture}[scale = 0.6]
\node at (-3,0)(1){$v_1$};
\node at (3,0)(2){$v_2$};
\node at (0,4)(3){$v_3$};
\draw (1) to (3) ;
\node at (-1.9,2) {\red $4$} ;
\draw (2) to (3);
\node at (1.9,2) {\red $3$} ;
\draw (1) to [out=150,in=90] (-5,0) to [out=270, in=210] (1) ;
\node at (-5.3,0) {\red $2$};
\draw (1) to [out=30,in=150] (2) ;
\node at (0,1.4) {\red $1$} ;
\draw (1) to [out=-30,in=180] (0,-1) to [out=0,in=-150] (2) ;
\node at (0,-0.7) {\red $2$} ;
\end{tikzpicture}\quad\quad%
\begin{tikzpicture}[scale = 0.6]
\node at (-3,0)(1){$v_1$};
\node at (3,0)(2){$v_2$};
\node at (0,4)(3){$v_3$};
\draw[fill] (-0.75,3) circle [radius = 0.1] ;
\draw[fill] (-1.5,2) circle [radius = 0.1] ;
\draw[fill] (-2.25,1) circle [radius = 0.1] ;
\draw (1) to (3) ;
\draw[fill] (2,1.3333) circle [radius = 0.1] ;
\draw[fill] (1,2.6666) circle [radius = 0.1] ;
\draw (2) to (3);
\draw[fill] (-5,0) circle [radius = 0.1] ;
\draw (1) to [out=150,in=90] (-5,0) to [out=270, in=210] (1) ;
\draw (1) to [out=30,in=150] (2) ;
\draw[fill] (0,-1) circle [radius = 0.1];
\draw (1) to [out=-30,in=180] (0,-1) to [out=0,in=-150] (2) ;
\end{tikzpicture}
\caption{$T$ and $\mec k\from E_T\to\naturals$ (left)
and $\VLG(T,\mec k)$ (right)}
\label{fi_vlg_example}
\end{figure}

Definition~\ref{de_VLG} shows that forming VLG's is a sort of
``opposite'' of forming bead suppressions.

\section{Algebraic Models}
\label{se_new_algebraic}

In this section we define {\em algebraic models}, which
are models that are strongly algebraic except that the
polynomials $p_i$ in
\eqref{eq_strongly_algebraic} can depend on some information
including regarding $S_\Bg$ which includes the homotopy type of $S_\Bg$.
The precise information is called a {\em $B$-type}, which is based on
{\em $B$-wordings} which we define beforehand; we will also need some
background on {\em regular languages}, which we now review.

\subsection{Regular Languages and ${\rm NBWALKS}(B)$}

We will use some notions 
from the theory of regular languages (e.g., \cite{sipser},
Chapter~1): if $\cA$ is an alphabet (i.e., a finite
set) and $k\in\integers_{\ge 0}$, we use
$\cA^k$ to denote the set of {\em words} (or {\em strings}) of length $k$
over $\cA$ (i.e., finite sequences
of $k$ elements of $\cA$), and we use $\cA^*$ to denote the union of
$\cA^k$ over all $k\in\integers_{\ge 0}$; 
a {\em language} over $\cA$ (i.e., a subset of
$\cA^*$) is {\em regular} if
it is recognized by some (deterministic) finite automaton or, equivalently,
if it can be expressed as a regular expression.

\begin{definition}\label{de_NBWALKS}
If $B$ is a graph, recall that a walk $w=(v_0,\ldots,e_k,v_k)$ of positive
length (i.e., $k\ge 1$) is determined by its sequence 
$(e_1,\ldots,e_k)\in(\Edir_B)^*$ of directed edges.
If $B$ is a graph, we use ${\rm NBWALKS}(B)\subset (\Edir_B)^*$ to
denote those words $(e_1,\ldots,e_k)$ of positive length 
that are the directed edges of a non-backtracking walk in $B$, i.e., for
which $t_B e_i = h_B e_{i-1}$ and $\iota_B e_i\ne e_{i-1}$ for all 
$2\le i\le k$.
Similarly, if $e,e'\in\Edir_B$, we use
${\rm NBWALKS}(B,e,e')$ to denote the subset of $(e_1,\ldots,e_k)\in
{\rm NBWALKS}(B)$ for which $e_1=e$ and $e_k=e'$.
\end{definition}
For all graphs $B$ 
${\rm NBWALKS}(B)$ is a
regular language, since the possible values of $e_i$ are determined by
those of $e_{i-1}$ for a word $(e_1,\ldots,e_i)\in{\rm NBWALKS}(B)$;
similarly, for any $e,e'\in \Edir_B$, ${\rm NBWALKS}(B,e,e')$ is a
regular language.

At times we will identify an $(e_1,\ldots,e_k)\in{\rm NBWALK}$
with its associated non-backtracking walk $(v_0,e_1,\ldots,e_k,v_k)$
in $B$, if confusion is unlikely to occur.

\subsection{$B$-Wordings}

Consider any suppression, $T=S/V'$, of a graph, $S$, with edge
lengths $\mec k\from E_T\to\naturals$
(or $\Edir_T\to\naturals$); then $\VLG(T,\mec k)$ is isomorphic---as
a graph---to $S$.
In this subsection we consider the information needed to recover
$S$ when we endow it with the structure of a $B$-graph;
we will this information a {\em $B$-wording}.  Let us first abstractly
define this notion.

\begin{definition}\label{de_wording}
Let $B,T$ be a graphs.  By a {\em $B$-wording of $T$} we mean a function
$$
W \from \Edir_T \to {\rm NBWALK}(B)
$$
such that
\begin{enumerate}
\item for all $e\in\Edir_T$,
$W(\iota_T e)=W(e)^R$ is the reverse walk, i.e.,
if $W(e)=(e_1,\ldots,e_k)$, then
$W(e)^R\eqdef (\iota_B e_k,\ldots,\iota_B e_1)$;
\item if $e\in\Edir_T$ is a half-loop, then $W(e)$ is of length one
whose single letter is a half-loop of $B$;
\item
\label{it_vertex_consistency}
the first vertex in $W(e)$ (i.e., the tail of the first
directed edge) depends only on $t_T e$
(i.e., $W(e),W(e')$ have the same first vertex for any $e,e'\in\Edir_T$
for which $t_T e=t_T e'$).
\end{enumerate}
By the {\em edge lengths} of $W$ we mean the function
$\Edir_T\to\naturals$ mapping $e$ to the length of $W(e)$; since
the length of $W(e)$ equals that of $W(\iota_T e)$, we also
view the edge lengths as a function $E_T\to\naturals$ (whose value
on an edge is that of an orientation of $e$).
\end{definition}

Let us explain how wordings arise.

\begin{definition}\label{de_induced_wording}
If $S_\Bg$ is a $B$-graph, and $T=S/V'$ is a suppression of $S$,
the {\em wording on $T$ induced by $S_\Bg$} refers to the following
map $W\from \Edir_T\to {\rm NBWALK}(B)$:
let $\pi\from S\to B$ be the structure map of $S_\Bg$;
by definition (Definition~\ref{de_suppression}) each element
of $\Edir_T$ is non-backtracking walk
$$
e_T=(v_0,\ldots,e_k,v_k)
$$
in $S$; we set
$$
W(e_T) = \pi(e_1),\pi(e_2),\ldots,\pi(e_k).
$$
\end{definition}
We easily check in the above definition that an induced wording
is actually a wording.
We also check that for any $B$-wording, $W$, of a graph $T$, there is
a $B$-graph $S_\Bg$ with a suppression $S/V'$ such that
(1) there is an isomorphism $\mu\from T\to S/V'$, (2) 
the wording induced by $S_\Bg$ on $S/V'$, when pulled back via
$\mu$, is the wording $W$;
we easily see that any two such $S_\Bg$ are isomorphic as
$B$-graphs.

\begin{definition}\label{de_realization}
If $W$ is a $B$-wording of a graph, $T$, with edge lengths $\mec k$,
then the {\em realization of $W$}, denoted $\VLG_\Bg(T,W)$
or $\VLG(T,W)$,
refers to any $B$-graph $S_\Bg$ whose underlying graph is $S=\VLG(T,\mec k)$,
and whose $B$-graph structure is the one given in the last paragraph:
namely, if $e_T\in\Edir_T$ corresponds to the beaded path in $S$
given as $e_1,e_2,\ldots,e_k$,
then we take $e_i$ to the $i$-th letter in $W(e_T)$.
(The directed edge $e_i$ appears in exactly one beaded path
by Lemma~\ref{le_suppression_unique_diedge}.)
\end{definition}
Of course, since $\VLG(T,\mec k)$ does not refer to a unique graph,
the above definition gives rise to many possible $S_\Bg$.
But we easily see that all such $B$-graphs are isomorphic as $B$-graphs.

\begin{figure}
\begin{tikzpicture}[scale=0.7]
\tikzstyle{sommet}=[circle,outer sep = 2pt,scale=0.3,radius=5pt,draw,thick,fill=black]
\node[sommet] at (-10,0) (v1) {}; 
\node[below] at (v1.south) {\green $u_T$} ;
\node[sommet] at  (0,0) (v2) {};
\node[below] at (v2.south) {\green $v_T$} ;
\draw[->] (v1.north) to [out=30,in=150] 
  node[above]{$\scaleto{{\blue W({\green e_T})=e_Bf_B}\leftrightarrow {\red (u_B,e_B,w_B,f_B,v_B)}}{6pt}$} (v2.north) ;
\draw[->] (v2) to 
  node[below]{$\scaleto{{\blue W({\green \iota_T e_T})=(\iota_B f_B)(\iota_B e_B)}
  \leftrightarrow {\red (v_B,\iota_B f_B,w_B,\iota_B e_B,u_B)}}{6pt}$} (v1) ;
\node[sommet] at (2,0) (v1new) {}; 
\node[below] at (v1new.south) {$u_T$} ;
\node[above] at (v1new.north) {$\red\scaleto{u_B}{4pt}$} ;
\node[sommet] at (5,0) (vmid) {};
\node[below] at (vmid.south) {new} ;
\node[above] at (vmid.north) {$\red\scaleto{w_B}{4pt}$} ;
\node[sommet] at  (8,0) (v2new) {};
\node[below] at (v2new.south) {$v_T$} ;
\node[above] at (v2new.north) {$\red\scaleto{v_B}{4pt}$} ;
\draw[->] (v1new.north) to [out=30,in=150] 
  node[above]{$\red\scaleto{e_B}{4pt}$} (vmid.north) ;
\draw[->] (vmid) to 
  node[below]{$\red\scaleto{\iota_B e_B}{4pt}$} (v1new) ;
\draw[->] (vmid.north) to [out=30,in=150] 
  node[above]{$\red\scaleto{f_B}{6pt}$} (v2new.north) ;
\draw[->] (v2new) to 
  node[below]{$\red\scaleto{\iota_B f_B}{6pt}$} (vmid) ;
\end{tikzpicture}
\caption{
A $B$-wording, $W$, on a graph, $T$, of two vertices joined by one edge;
$V_T=\{u_T,v_T\}$
(labeled in {\green green}), $\Edir_T=\{e_T,\iota_T e_T\}$.
$W(e_T)$ is the word $e_B f_B$ over the alphabet $\Edir_B$, which is
identified with the NB walk $(u_B,e_B,w_B,f_B,v_B)$ in $B$, where 
$u_B=t_B e_B$, $w_B=h_B e_B=t_B f_B$, $v_B=h f_B$.
$\VLG_{/B}(T,W)$ is a path of length $2$, with the
indicated $B$-graph structure (in {\red red} on the right).
}
\label{fi_wording}
\end{figure}
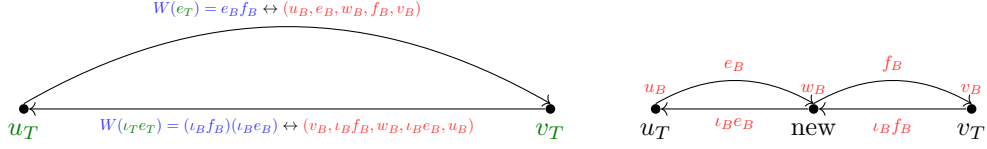

See Figure~\ref{fi_wording} for an example.

We remark that if $T$ is endowed with an ordering $T^\og$ arising from
an SNBC walk, 
then this ordering induces
one on $\VLG(T,\mec k)$ and therefore on the realization of $W$.
This ordering will be used in Article~II; however here we want to
define only what is meant by an {\em algebraic} model, and this notion
relies only on wordings of graphs rather than of ordered graphs.

Notice that if $S_\Bg^\og=\VLG(T^\og,W)$ for a $B$-wording, $W$,
then the invariants $\mec a_{S_\Bg}$ and $\mec b_{S_\Bg}$ can be
inferred from $W$, and we may therefore write $\mec a_W$ and
$\mec b_W$.

\subsection{$B$-Types}

\begin{definition}\label{de_B_type}
Let $B$ be a graph.
By a {\em $B$-type} we mean a pair
$T^{\rm type}=(T,\cR)$ where
$T$ is a graph, and
$\cR$ is function from $\Edir_T$ to the set of regular languages 
over the alphabet $\Edir_B$
that are subsets of
${\rm NBWALKS}(B)\subset(\Edir_B)^*$, such that
\begin{enumerate}
\item
for all $e\in\Edir_T$, $w\in \cR(e)$ iff $w^R\in\cR(\iota_T e)$;
\item
if $W$ is any function $\Edir_T\to {\rm NBWALKS}(B)$ such that
for all $e\in\Edir_T$,
$W(e)\in \cR(e)$ and
$W(\iota_T e) = W(e)^R$, then
$W$ is a $B$-wording.
\end{enumerate}
Furthermore, we say that any $B$-wording $W$ as in~(2) 
{\em is of type $\cR$} and
{\em belongs to $\cR$}, and the same with $T^{\rm type}$
replacing $\cR$.
\end{definition}
The novelty of this definition, which is crucial to Article~II
(and \cite{friedman_random_graphs}), is that
if $\Eor_T\subset\Edir_T$ is any orientation of $T$, then the set of 
$B$-wordings that belong to any $B$-type
$T^{\rm type}=(T,\cR)$ is in one-to-one correspondence with
$$
\prod_{e\in\Eor_T} \cR(e).
$$

\subsection{The Definition of Algebraic}

\begin{definition}\label{de_algebraic}
Let $B$ be a graph, and $N\subset\naturals$ an infinite set, and
for $n\in N$ let $\cC_n(B)$ be a probability space whose atoms are
elements of ${\rm Coord}_n(B)$.  
Recall that a $B$-graph $S_\Bg$ {\em occurs in $\cC_n(B)$}
if for all sufficiently large $n\in N$ there is a
$G\in\cC_n(B)$ such that
$[S_\Bg]\cap G$ is nonempty, i.e., $G$ contains a $B$-subgraph
isomorphic to $S_\Bg$.
Let $\cC_n(B)$ satisfy (1)--(3) of
Definition~\ref{de_strongly_algebraic}.
If $\cT$ is a subset of $B$-graphs,
we say that $\cC_n(B)$ is
{\em algebraic when restricted to $\cT$} if
either all $S_\Bg\in\cT$ occur in $\cC_n(B)$ or they all do not,
and (if so)
there are polynomials $p_{i}=p_{i}(\mec a_{S_\Bg})$ such that
for each $S_\Bg\in\cT$ and $i\in\naturals$ 
\begin{equation}\label{eq_algebraic}
c_i(S_\Bg) = p_{i}(\mec a_{S_\Bg})
\end{equation}
(where $c_i=c_i(S_\Bg)$ are as in Definition~\ref{de_strongly_algebraic}).
We say that the family
of probability spaces $\{\cC_n(B)\}_{n\in N}$ is 
{\em algebraic} provided that it satisfies 
conditions~(1)--(3) of 
Definition~\ref{de_strongly_algebraic}, and
\begin{enumerate}
\item
the number of $B$-graph isomorphism classes of \'etale $B$-graphs
$S_\Bg$ such that $S$ is a cycle of length $k$ and $S$ does
not occur in $\cC_n(B)$ is equals $h(k)$ where
$h$ is a function of growth $(d-1)^{1/2}$; and
\item
for any
pruned, ordered graph, $T^\og$, there is a finite number of
$B$-types, $T_j^{\rm type}=(T^\og,\cR_j)$, $j=1,\ldots,s$, 
such that (1) any $B$-wording, $W$, of $T$ belongs to exactly one
$\cR_j$, and
(2) $\cC_n(B)$ is algebraic when restricted to $T_j^{\rm type}$.
\end{enumerate}
\end{definition}
Let us make a few remarks on the above definition.

First,
in \eqref{eq_algebraic}, the $p_{i}$ are written in terms of
$\mec a$ alone since
(1) for fixed $B$-type $(T,\cR_j)$,
$\mec b$ turns out to be a fixed, linear function of $\mec a$, and
(2) this will be convenient to us in Article~II.

Second, in proving our main theorems in Article~II,
it is convenient to insist that
each wording belong to a {\em unique} $B$-type
$(T,\cR_j)$ rather than to {\em at least one}.
However, it is easy to prove that
if each $B$-wording belongs to {\em at least one} of $(T,\cR_j)$ in
the above definition,
then there is another set of $B$-types for which (1) and (2) of
Definition~\ref{de_algebraic} 
holds: indeed, let us give a proof, starting with the general fact that
if
$L_1,\ldots,L_u$ are any sets,
and for $A\subset [u]$ we set
$$
L^A\eqdef \bigcap_{a\in A} L_a \ \setminus\ \bigcup_{b\notin A} L_b,
$$
then each element of $L\eqdef L_1\cup\cdots\cup L_u$ lies in a unique $L^A$
(with $A\ne\emptyset$);
moreover, if each $L_i$ is a regular language over some
common alphabet, then so is each $L^A$ (by the closure properties
of regular languages); finally, it will be useful to note that
for any $a\in[u]$ and $A\subset[u]$, 
\begin{equation}\label{eq_meets_implies_containment} 
L^A \cap L_a\ne\emptyset \quad\implies\quad
a\in A \quad\implies\quad
L^A \subset L_a 
\end{equation} 
(both reverse implications hold whenever $L^A\ne \emptyset$).
Now take $L_1,\ldots,L_u$ to be all sets of the form $\cR_j(e)$
with $1\le j\le s$ and $e\in\Edir_T$ 
[one could use an orientation $\Eor_T$ instead of all of $\Eor_T$],
and consider all $B$-types of the form 
$(T,\cR')$ where for each $e\in\Edir_T$, $\cR'(e)$ is of the
form $L^{A(e)}$ for some $A(e)\subset[u]$, and where $(T,\cR')$ contains
at least one $B$-wording of $T$.
In this way,
then each $B$-wording of $T$, $W$, lies in a unique such $B$-type
$(T,\cR')$.
But now we claim that $\cC_n(B)$ restricted to
any such $(T,\cR')$ is algebraic: indeed, such a $(T,\cR')$
contains some wording, $W$,
and such a wording is of $\cR_j$ for at least one $j$;
fix such a $j$.
We have $W(e)\in\cR'(e)\cap\cR_j(e)$ for all $e\in\Edir_B$,
and therefore
\eqref{eq_meets_implies_containment} implies that
$\cR'(e)\subset\cR_j(e)$ for all $e\in\Edir_B$.
Therefore any wording of type $\cR'$ is also of type
$\cR_j$, and since $\cC_n(B)$ is algebraic
when restricted to $\cR_j$, it is also algebraic
when restricted to $\cR'$.

\subsection{The Eigenvalues of a Regular Language, of a $B$-type, and of 
an Algebraic Model}

In this subsection we define what we mean by the {\em eigenvalues}
of an algebraic model.
The {\em eigenvalues} of an algebraic model appear in the statements
of the main theorems of Articles~II
and III; they are also fundamental to the main theorem of Article~VI,
where we prove that $\tau_{\rm alg}=+\infty$ under certain conditions,
one of which is that all {\em larger eigenvalues} of the model are 
either $\pm(d-1)$.

First we point out an easy fact about regular languages.

\begin{proposition}\label{pr_eigens_regular_language}
Let $L$ be a regular language, and for $k\in\naturals$ let
$f(k)$ be the number of 
words in $L$ of length $k$.
Then $f(k)$ is a polyexponential function of $k$
(Definition~\ref{de_polyexponential_growth}), i.e.,
there are
unique distinct $\mu_1,\ldots,\mu_m\in\complex$ and unique non-zero
polynomials $p_1(k),\ldots,p_m(k)$ such that for all $k\in\naturals$
\begin{equation}\label{eq_reg_lang_words}
f(k) = \sum_{i=1}^m p_i(k) \mu_i^k ,
\end{equation} 
where we understand that
$p_i(k) \mu_i^k$ with $\mu_i=0$ refers to a non-zero function that is
zero for $k$ sufficiently large.
[And we understand the convention that if $f$ is identically zero, then
$m=0$ and $L=\emptyset$.]
\end{proposition}
The proof follows by considering a finite automaton, $M$, recognizing $L$, and
letting $A_M$ be the square matrix indexed
on the states, $Q$, of $M$ with $(A_M)_{q_1,q_2}$ equal to the number of 
letters in the alphabet taking state $q_1$ to state $q_2$.  Then 
$$
f(k) = \sum_{q\in F} (A_M^k)_{q_0,q}
$$
where $q_0$ is the initial state of $M$ and $F$ is the set of accepting states
of $M$; the proposition follows from the Jordan canonical form of $A_M$
(an eigenvalue $0$ of $A_M$ yields a nilpotent Jordan block, which explains
our convention for $\mu_i=0$).

\begin{example}
If $L={\rm NBWALKS}(B,e,e')$ for a graph $B$ and $e,e'\in\Edir_B$, 
then the number of strings of length $k$ in $L$ is just the $(e,e')$
entry of $H_B^k$; in this case
the $\mu_i$ above are
always a subset of the $\mu_i(B)$, the eigenvalues of the
Hashimoto matrix $H_B$.
Also the oriented line graph of $B$ easily yields a finite automaton
recognizing $L$, which shows that $L$ is a regular language.
\end{example}

Note that it is crucial that we view ${\rm NBWALKS}(B)$ as strings
in $\Edir_B$ so that the eigenvalues are what we want:
one could view any $L\subset {\rm NBWALKS}(B)$ as a set of
alternating strings
of $V_B$ and $\Edir_B$, then this would roughly double the length of
each word in $L$ and therefore change the eigenvalues of $L$.
Hence it is crucial that we omit the vertices when describing
NB walks in order to get the correct values
of {\em eigenvalues of $L$} that we need for our trace methods,
even though the notion of {\em regular language} is 
(easily checked to be) the same whether or not
we omit the vertices.

\begin{definition}\label{de_eigens_regular_language}
If $L$ is a regular language, then the {\em eigenvalues} of $L$ are
the unique $\mu_1,\ldots,\mu_m\in\complex$ in
Proposition~\ref{pr_eigens_regular_language}.
If $B$ is a graph, the {\em eigenvalues} of a $B$-type $(T^\og,\cR)$ is
the union of the eigenvalues of all the regular languages $\cR(e)$ with
$e\in\Edir_T$.
If $\{\cC_n(B)\}_n$ is an algebraic model, then {\em a set of eigenvalues}
of the model is any subset of $\complex$ that for any $T^\og$
contains all the eigenvalues of some set of $B$-types $(T^\og,\cR_j)$
satisfying the conditions of Definition~\ref{de_algebraic}.
\end{definition}
In the above definition, a {\em set of eigenvalues} is not unique; we
can always add some extraneous eigenvalues by
subdividing a $B$-type into a number of smaller $B$-types whose
underlying regular languages have additional eigenvalues; it is not clear 
to us (at least from
the definition) that there is a unique minimal set of eigenvalues
of a model.

\subsection{Our Basic Models are Algebraic}
\label{se_our_basic_models_are_algebraic}

In Article~V we will prove that all of
our basic models (Definition~\ref{de_models})
over a graph, $B$, are algebraic, and a set of eigenvalues
for these models
is possibly $1$ (for models involving full-cycles)
and some subset of the eigenvalues
of $H_B$.  Let us make some remarks regarding this proof; these 
remarks help to motivate our definition of $B$-type.

First, it is not hard to prove that
any model that is strongly algebraic is also algebraic.
The main point (see Article~V for details)
is that for a strongly algebraic model,
for any $S_\Bg$ \eqref{eq_strongly_algebraic} implies that
$$
c_{\ord(S)+i}(S_\Bg) = 
\left\{ \begin{array}{ll} p_i(\mec a_{S_\Bg},\mec b_{S_\Bg}) 
& \mbox{if $S_\Bg$ is an \'etale $B$-graph, and} \\
0 & \mbox{otherwise.} \end{array}\right.
$$
However, if $S_\Bg=\VLG_\Bg(T,W)$, then whether or not $S_\Bg$ is \'etale
depends only on the first and last letters of $W(e)$ for all $e\in\Edir_T$.
So if $(T,\cR)$ varies over all $B$-types where all $\cR(e)$ are
sets of the form ${\rm NBWALK}(B,e_1,e_2)$, then each $B$-wording belongs
to a unique $(T,\cR)$; furthermore, the polynomials
expressing $c_i(S_\Bg)$ in terms of 
$\mec a_{S_\Bg},\mec b_{S_\Bg}$ depend only on which $\cR$ gives
rise to
$S_\Bg$ (i.e., as $\VLG_\Bg(T,W)$ with $W\in\cR$).
Furthermore, it is easy to see that for a fixed $B$-type $(T,\cR)$,
the variables $\mec b_{S_\Bg}$ are fixed, linear functions of the
$\mec a_{S_\Bg}$.
For this reason all strongly algebraic models are algebraic.

We remark that $B$-types $(T,\cR)$,
as above, based on the first and last letters of
each $W(e)$,
corresponds to the notion of a {\em lettering} in
in \cite{friedman_random_graphs}, page~338,
and \cite{friedman_alon} (Definition~5.8).
[In both these articles, ${\rm Irr}_{k,e,e'}$ denotes the elements of
${\rm NBWALK}(B,e,e')$ of length $k$.]
See also 
Subsection~\ref{su_realizable_letterings_regular} of this article.

The main reason why we want to allow $\cR(e)$ to be a more general 
regular language (more general than ${\rm NBWALK}(B,e_1,e_2)$ is because
of the cyclic model:
a random full cycle, $\pi$, in $\cS_n$ with
$a$ of its values fixed occurs with probability
$$
\frac{1}{(n-1)\ldots (n-a)}
$$
provided that the values fixed do not force $\pi$ to have a cycle of length
$\le n-1$.
It follows that when we identify wordings of type $(T,\cR)$ with
$$
\prod_{\Eor_T} \cR(e)
$$
for an orientation $\Eor_T$, we must be careful to avoid wordings that
force the associated $\sigma\from\Edir_B\to\cS_n$ to have a cycle of
length $\le n-1$ at any whole-loop $e\in \Edir_B$
(i.e., that force $\sigma(e)$---which we insist is a full-cycle in this
model---to have a cycle of length less than $n$).
For example, if $e_T\in\Edir_T$ is a whole-loop, and $e_B\in\Edir_B$ is
a whole-loop, then we must forbid
$W(e_T)$ to be a word in $(e_B)^*$ (i.e., of the form 
$(e_1)^\ell=(e_1,\ldots,e_1)$).
Hence for the cyclic model
our $B$-types include $(T,\cR)$ where $\cR(e)$ can take on 
values such as
$$
{\rm NBWALK}(B,e_B,e_B)\setminus e_B^*, \quad
e_B^*,
$$
(whenever $e_B$ is a whole-loop),
as well as the sets ${\rm NBWALK}(e_1,e_2)$.
Since $e_B^*$ is a language whose eigenvalues are $1$
(there is exactly one word of any length),
the eigenvalues of the cyclic model must include $1$
in addition to the $\mu_i(B)$.
[For similar reasons, the language
${\rm NBWALK}(B,e_B,e_B)\setminus e_B^*$ has $1$ as an eigenvalue
in addition to some of the $\mu_i(B)$.]

The above remarks about the cyclic model were overlooked in
\cite{friedman_alon}; so working with regular languages and $B$-types
as defined here is one way to fix
this error.
This correction doesn't change any of the computations, since
these computations are done modulo functions of growth
$\nu>(d-1)^{1/2}$ in the NB walk statistics.

\section{Article II: Expansion Theorems}
\label{se_art_expansion}

Our asymptotic expansions theorems are akin to those in
\cite{friedman_random_graphs} and proven by the same methods.
Let us briefly describe a special case of the main result of Article~II
that indicates the general type of result.

\subsection{Asymptotic Expansions for Walks of a Given $B$-Type}

The main expansion theorems in Article~II can be understood for the
special case of the expected number of walks of a given type.

\begin{theorem}\label{th_expansion_thm_for_types}
Let $B$ be a graph, and $\{\cC_n(B)\}_{n\in N}$ an algebraic model over
$B$.  Let 
$T^\og$ be an ordered graph, let
$\bec\cert\from E_T\to\naturals$ be a function, 
and let
$$
\nu = \max\Bigl( \mu_1^{1/2}(B), \mu_1\bigl(\VLG(T,\bec\cert)\bigr) \Bigr).
$$
Then for any $r\ge 1$ we have
$$
f(k,n)\eqdef \EE_{G\in\cC_n(B)}[\snbc(T^\og,\ge\bec\cert;G;k)]
$$
has a $(B,\nu)$-bounded expansion 
$$
c_0(k)+\cdots+c_{r-1}(k)/n^{r-1}+ O(1) c_r(k)/n^r
$$
to order $r$, where the
bases of the coefficients $c_i=c_i(k)$ is a subset of a set
of eigenvalues of the model; furthermore $c_i(k)=0$ for
$i<\ord(T)$.
\end{theorem}

\subsection{Asymptotic Expansions for Walks Times Inclusions}

If $S_\Bg,G_\Bg$ are fixed $B$-graphs, then for any ordering on $S_\Bg$,
\begin{equation}\label{eq_define_N}
N(S_\Bg,G_\Bg) \eqdef \#[S_\Bg^\og]\cap G_B
\end{equation}
is independent of the ordering; the
quantity $N(S_\Bg,G_\Bg)$ features prominently in our trace methods
because 
(1) in Article~III we will use linear combinations of
$N(S_\Bg,G_\Bg)$---where $S_\Bg$ varies over a finite number
of $(\ge\nu,<r)$-tangles---to approximate the
indicator function
$$
\II_{{\rm HasTangles}(\ge\nu,<r)}(G)
$$
which we use to prove the main theorems there
(e.g., Theorem~\ref{th_main_tech_result}; see also
Appendix~\ref{se_append_cert_ind}),
and 
(2)
we can prove the following variant of 
Theorem~\ref{th_expansion_thm_for_types} with only minor additional
difficulties.


\begin{theorem}\label{th_expansion_thm_for_types_times_inclusions}
Let $\cC_n(B)$ be an algebraic model over a graph $B$.
Let $T^\og$ be an ordered graph, let $\bec\xi\from E_T\to\naturals$ be
a function, and let
$$
\nu = \max\Bigl( \mu_1^{1/2}(B), \mu_1\bigl(\VLG(T,\bec\cert)\bigr) \Bigr).
$$
Let $\psi_\Bg^\og$ be any ordered $B$-graph.
Then for any
$r\ge 1$ we have
\begin{equation}\label{eq_subgraphs_times_walks}
\EE_{G\in\cC_n(B)}[ 
(\#[\psi_\Bg^\og]\cap G)
\snbc(T^\og;\ge\bec\xi,G,k) ]
\end{equation} 
has a $(B,\nu)$-bounded expansion 
$$
c_0(k)+\cdots+c_{r-1}(k)/n^{r-1}+ O(1) c_r(k)/n^r,
$$
to order $r$; the bases of
the coefficients $c_i=c_i(k)$ are some subset of the
eigenvalues of the model, and $c_i(k)=0$ for $i$ less than the
order of any $B$-graph that contains both a walk of 
homotopy type $T^\og$
and a subgraph isomorphic to $\psi_\Bg$.
\end{theorem}

Technically
Theorem~\ref{th_expansion_thm_for_types_times_inclusions}
is a generalization of
Theorem~\ref{th_expansion_thm_for_types}, by taking $\psi_\Bg^\og$ to
be the empty graph.
However, we first prove
Theorem~\ref{th_expansion_thm_for_types},
since its proof is less notationally
cumbersome and illustrates the key ideas in the proof of
Theorem~\ref{th_expansion_thm_for_types_times_inclusions}.

We wish to make one technical remark regarding Article~II:
one could scale the entries of the matrix $H_B$ to prove a more
general expansion theorem; one limit of this scaling occurs
in what is called the Dot Convolution Theorem in this article;
this requires the Perron-Frobenius eigenvalue of any scaled
form of $H_B$ to be at least one (or else one in defining
$(B,\nu)$-functions one must require $\nu\ge 1$).

\section{Article III: The Certified Trace Expansion Theorems}
\label{se_art_certified}

Most of Article~III is devoted to proving the following result.

\begin{theorem}\label{th_main_tech_result}
Let $B$ be a connected graph with 
$\mu_1(B)>1$, and let 
$\{\cC_n(B)\}_{n\in N}$ be
an algebraic model over $B$.
Let $r>0$ be an integer and $\nu\ge\mu_1^{1/2}(B)$ be a real number.
Then 
\begin{equation}\label{eq_main_tech_result1}
f(k,n)\eqdef
\EE_{G\in\cC_n(B)}[ \II_{{\rm TangleFree}(\ge\nu,<r)}(G) \Trace(H^k_G) ]
\end{equation}
has a $(B,\nu)$-bounded expansion to order $r$,
$$
f(k,n)=c_0(k)+\cdots+c_{r-1}(k)/n^{r-1}+O(1)c_r(k)/n^r,
$$
where
$$
c_0(k)=\sum_{k'|k} \Trace(H_B^{k'}) 
$$
where the sum is over all positive integers, $k'$, dividing $k$; hence
$$
c_0(k) = \Trace(H_B^k) + O(k) \mu_1^{k/2}(B);
$$
furthermore, the larger
bases of each $c_i(k)$ (with respect to $\mu_1^{1/2}(B)$)
is some subset of the eigenvalues
of the model.
Finally, for any $r'\in\naturals$ the function
\begin{equation}\label{eq_main_tech_result2}
\widetilde f(n)  \eqdef
\EE_{G\in\cC_n(B)}[ \II_{{\rm TangleFree}(\ge\nu,<r')}(G)]
=
\Prob_{G\in\cC_n(B)}[ G\in {\rm TangleFree}(\ge\nu,<r') ]
\end{equation}
has an asymptotic expansion in $1/n$ to any order $r$,
$$
\widetilde c_0+\cdots+\widetilde c_{r-1}/n^{r-1}+O(1)/n^r ;
$$
where $\widetilde c_0=1$; furthermore, if $j_0$ is the 
smallest order of a $(\ge\nu)$-tangle occurring in $\cC_n(B)$,
then 
$\widetilde c_j=0$ for $1\le j<j_0$ and
$\widetilde c_j>0$ for $j=j_0$ (provided that $r\ge j_0+1$ so that
$\widetilde c_{j_0}$ is defined).
\end{theorem}


Notice that a model may have---at least in principle---an infinite number 
of eigenvalues, which
means that for each $r,\nu$, the number of bases of the
$c_i(k)$ may be unbounded as $i\to\infty$; 
however there are a few remarks to consider:
\begin{enumerate}
\item 
Taking $\nu=\mu_1^{1/2}(B)$, for each $r$, the $c_i(k)$ with $i<r$
have a finite number of exponent bases;
\item since for any fixed $k$ we have
$$
c_i(k)  = \lim_{n\to\infty} n^i \bigl( f(k,n) - c_0(k) - \cdots - 
n^{1-i}c_{i-1}(k) \bigr) ,
$$
the function 
$c_i(k)$ is uniquely defined and independent of $r$ over all $r>i$;
hence a fixed $c_i(k)$ has a finite
number of larger 
(than $\nu$) bases.
\item
In all our basic models, the eigenvalues consist only of the 
$\mu_j(B)$ and possibly the eigenvalue $1$; hence all larger bases
of the $c_i(k)$ lie in this finite set.
\end{enumerate}


We remark
(see \cite{friedman_kohler}) that one could replace the condition
$\nu\ge \mu_1^{1/2}$ by $\nu\ge \mu_1^{1/N}$ for any positive integer $N$, 
but
then the coefficients in the asymptotic expansions may depend on
the class of $n$ modulo the least common multiple of the numbers 
from $1$ to $N$.


A description of the methods used to prove 
Theorem~\ref{th_main_tech_result} is given in 
Appendix~\ref{se_append_cert_ind}.
These methods include ideas regarding
tangles, minimal tangles, and indicator function approximations
\cite{friedman_alon}, Section~9 (see Article~III for more detailed
references) with the expansion theorems proven in Article~II.
The main difference between Article~III and \cite{friedman_alon}
is that Article~III involves {\em certified traces},
which are simpler to define
and much easier to work with than the {\em selective traces}
of \cite{friedman_alon}.

In Article~V we will also need the following result, whose proof
is similar to the proof of the expansion for \eqref{eq_main_tech_result2}
above.

\begin{theorem}\label{th_extra_needed}
Let $\cC_n(B)$ be an algebraic model over a graph, $B$, and
let $S_\Bg$ be a connected, pruned graph of positive order
that occurs in this model.
Then for some constant, $C'$, and
$n$ sufficiently large,
$$
\Prob_{G\in\cC_n(B)}\Bigl[ [S_\Bg]\cap G\ne\emptyset
\Bigr] \ge
C' n^{-\ord(S_\Bg)}.
$$
\end{theorem}

\section{Article IV: The Sidestepping Theorem}
\label{se_art_sidestep}

In Article~IV we prove an improved form of 
the {\em Sidestepping Lemma} of \cite{friedman_alon}, suited to our
more general context.

\subsection{Intuition Behind the Sidestepping Theorem}

Let us give the intuition behind our Sidestepping Theorem.
To fix ideas, say that for each $n\in\naturals$, $\xi_n$ is a random variable
on a probability space $\cP_n$ with $\xi_n=2$ on an event $E_n$ 
of probability $1-(1/n^3)$, and $\xi_n$ is uniformly distributed on
$[4,5]$
on the complement of $E_n$.  In this case we have
$$
f(k,n)\eqdef \EE[\xi_n^k] = 2^k + (1/n^3)(c_k-2^k),
$$
where $c_k$ is the $k$-th moment of the uniform distribution on $[4,5]$.
Similarly, if $M_n$ is a random $n\times n$ matrix whose eigenvalues are
all distributed as $\xi_n$, then one has
\begin{equation}\label{eq_typical_side}
f(k,n)\eqdef \EE[\Trace(M_n^k)] = 2^k n + (1/n^2)(c_k-2^k).
\end{equation}
The Sidestepping Theorem attempts to invert this process, namely given
a random $n\times n$ matrix $M_n$ for which
$$
\EE[\Trace(M_n^k)] = c_{-1}(k) n + c_0(k) + \cdots + c_r(k)/n^{r-1} 
+ O\bigl(1/n^r\bigr),
$$
with some assumptions on the eigenvalues of $M_n$ and some assumptions
on the $c_i(k)$, this theorem guarantees (very) high probability bounds
on the locations of $M_n$'s eigenvalues.
Our Sidestepping Theorem assumes that $c_{-1}(k)=0$, since this is the
case for expected traces of powers of the Hashimoto matrices in
our models.

It turns out that all the techniques and results in Article~IV can be stated
as theorems regarding the random 
set of $n$ eigenvalues of $M_n$; for example, $\Trace(M_n^k)$ is
just the sum of their $k$-th powers.
However, we state our results in Article~IV in terms of 
$\Trace(M_n^k)$ because this is how we will apply these results,
namely 
to draw conclusions about these eigenvalues based on
facts regarding the
$G\in\cC_n(B)$ expected values of
$\Trace(H_G^k)$ times the indicator function of $G$ that indicates
that $G$ is free of certain tangles.
Let us state the precise theorem.

\subsection{Precise Statement of the Sidestepping Theorem}

\begin{definition}\label{de_matrix_model}
Let $\Lambda_0<\Lambda_1$ be positive real numbers.  
By a {\em $(\Lambda_0,\Lambda_1)$
matrix model} we mean a collection of finite probability spaces
$\{\cM_n\}_{n\in N}$ where $N\subset\naturals$ is an
infinite subset, and where the atoms of $\cM_n$ are
$n\times n$ real-valued matrices whose eigenvalues lie in the set
$$
B_{\Lambda_0}(0) \cup [-\Lambda_1,\Lambda_1]
$$
in $\complex$.
Let $r\ge 0$ be an integer and $K\from\naturals\to\naturals$
be a function such that $K(n)/\log n\to \infty$ as $n\to\infty$.
We say that this model has an
{\em order $r$ expansion} with {\em range $K(n)$} 
(with $\Lambda_0,\Lambda_1$ understood) if
as $n\to\infty$ we have that
\begin{equation}\label{eq_matrix_model_exp}
\EE_{M\in\cM_n}[\Trace(M^k)] =
c_0(k) + c_1(k)/n +\cdots+c_{r-1}(k)/n^{r-1}+ O(c_r(k))/n^r
\end{equation}
for all $k\in\naturals$ with $k\le K(n)$,
where (1) $c_r=c_r(k)$ is of growth $\Lambda_1$,
(2) the constant in the $O(c_r(k))$ is
independent of $k$ and $n$, and
(3) for $0\le i<r$, $c_i=c_i(k)$ is an approximate polyexponential with
$\Lambda_0$ error term and whose larger bases 
(i.e., larger than $\Lambda_0$ in absolute value)
lie in $[-\Lambda_1,\Lambda_1]$;
at times we speak of an {\em order $r$ expansion} without 
explicitly specifying $K$.
When the model has such an expansion,
then we use the notation $L_r$ to refer to the union of all larger
bases of $c_i(k)$ (with respect to $\Lambda_0$) over all $i$ between
$0$ and $r-1$, and call $L_r$ the {\em larger bases (of the order $r$
expansion)}.
\end{definition}
Note that in Definition~\ref{de_matrix_model},
the larger bases of the $c_i$ are an arbitrary finite subset
of
$[-\Lambda_1,\Lambda_1]\setminus[-\Lambda_0,\Lambda_0]$ 
(e.g., there is no bound on the
number of bases).
In our applications we will take $\Lambda_1=d-1$ and
$\Lambda_0$
slightly larger than $(d-1)^{1/2}$.

%
We also note that \eqref{eq_matrix_model_exp} implies that
for fixed $k\in\naturals$,
\begin{equation}\label{eq_c_i_limit_formula}
c_i(k) = \lim_{n\in N,\ n\to\infty}
\Bigl( \EE_{M\in\cM_n}[\Trace(M^k)] -
\bigl( c_0(k) + \cdots+c_{i-1}(k)/n^{i-1} \bigr) \Bigr) n^i
\end{equation}
for all $i\le r-1$; we conclude that the $c_i(k)$ are uniquely determined,
and that $c_i(k)$ is independent of $r$ for any $r>i$ for which
\eqref{eq_matrix_model_exp} holds.
We also see that if \eqref{eq_matrix_model_exp} holds for some value of
$r$, then it also holds for smaller values of $r$.
It follows that if \eqref{eq_matrix_model_exp} holds for some $r$, then
$L_i$ is defined for each $i<r$ (as the set of larger bases of the
functions $c_0(k),\ldots,c_{i-1}(k)$).
Furthermore $L_i$ is empty iff
$c_0(k),\ldots,c_{i-1}(k)$ are all functions of growth $\Lambda_0$.

Let us remark on what $N$ in the above definition typically looks like.
In our applications the matrices will be that of the new functions
on $H_G$ (times an indicator function),
over $G\in\cC_{n'}(B)$ where $n'\in N'$, with $N'$ some infinite set.
Hence the 
spaces $\cM_{n}$ would vary over
dimension $n=(n'-1)(\#\Edir_B)$ with $n'\in N'$.

If $\cM$ is a probability space of $n\times n$ matrices, and
$R\subset \complex$, it will be very useful to use the
shorthand
\begin{equation}\label{eq_Ein_Eout}
\EE{\rm in}_\cM[R] \quad\mbox{and}\quad
\EE{\rm out}_\cM[R] \ ,
\end{equation}
respectively,
as the expected number of eigenvalues of $M\in\cM$ (counted with multiplicity)
that lie, respectively, in $R$ and not in $R$.  Hence 
the sum of these two expected values is $n$.

\ignore{
The following is an intermediate step to proving the ``Sidestepping Lemma,''
and gives good intuition as to where the eigenvalues of matrices in a
$(\Lambda_0,\Lambda_1)$-bounded matrix model are located.

\begin{lemma}\label{le_sidestep_one}
Let $\Lambda_0<\Lambda_1$ and $\epsilon,\alpha$ all be
positive real numbers.
Then there is
a positive integer $r_0=r_0(\Lambda_0,\Lambda_1,\epsilon,\alpha)$ 
and a real
$\theta_0=\theta_0(\Lambda_0,\Lambda_1,\epsilon,\alpha)>0$
with the following property:
let $\{\cM_n\}_{n\in N}$ be a $(\Lambda_0,\Lambda_1)$-bounded matrix model
that has an order $r$ expansion \eqref{eq_matrix_model_exp} 
for some $r\ge r_0$, and let $L_r$ be the set of larger bases of 
$c_0(k),\ldots,c_{r-1}(k)$.
Then for any sufficiently large $n$ and any $\theta\le\theta_0$ we have
\begin{equation}\label{eq_sidestep_r_one}
\EE{\rm out}
\bigl[ B_{\Lambda_0+\epsilon}(0) \cup B_{n^{-\theta}}(L_r) \bigr] 
\le n^{-\alpha} \ .
\end{equation} 
\end{lemma}
In other words, an eigenvalue of $M\in\cM_n$ that lies outside of
$B_{\Lambda_0+\epsilon}(0) \cup B_{n^{-\theta}}(L_r)$ is 
``exceptional'' in the sense that the expected number
of such eigenvalues decays larger than any fixed power of $n$,
for models with expansions of order $r$ with $r$ sufficiently large,
namely $r\ge r_0(\Lambda_0,\Lambda_1,\epsilon,\alpha)$.
We note that it will be useful to know that $r_0$ does
depend on the coefficients $c_i=c_i(k)$ of
\eqref{eq_matrix_model_exp} or their bases
(provided the larger bases lie in $[-\Lambda_1,\Lambda_1]$).
}

Here is the main theorem of Article~IV, a strengthening
and easier to apply version of the Sidestepping Lemma
of \cite{friedman_alon}.

%

\begin{theorem}\label{th_sidestep}
Let $\{\cM_n\}_{n\in N}$ be a $(\Lambda_0,\Lambda_1)$-bounded matrix model,
for some real $\Lambda_0<\Lambda_1$, that
for all $r\in\naturals$ has an order $r$ expansion;
let $p_i(k)$ denote the polyexponential part of $c_i(k)$
(with respect to $\Lambda_0$) in \eqref{eq_matrix_model_exp}
(which is independent of $r\ge i+1$ by
\eqref{eq_c_i_limit_formula}).
If $p_i(k)=0$ for all $i\in\integers_{\ge 0}$, then
for all $\epsilon>0$ and $j\in\integers_{\ge 0}$
\begin{equation}\label{eq_largest_j}
\EE{\rm out}_{\cM_n}
\bigl[ B_{\Lambda_0+\epsilon}(0) \bigr]
= O(n^{-j}) .
\end{equation}
Otherwise let $j$ be the smallest integer for which
$p_j(k)\ne 0$.
Then for all $\epsilon>0$, and for all $\theta>0$ sufficiently small we have
\begin{equation}\label{eq_mostly_near_Lambda_0_or_Ls}
\EE{\rm out}_{\cM_n}
\bigl[ B_{\Lambda_0+\epsilon}(0)\cup B_{n^{-\theta}}(L_{j+1}) \bigr]
= o(n^{-j});
\end{equation}
moreover, if $L=L_{j+1}$ is the (necessarily nonempty) set of bases of $p_j$,
then for each $\ell\in L$ there is a real $C_\ell>0$ such that
\begin{equation}\label{eq_thm_p_j_pure_exp}
p_j(k)=\sum_{\ell\in L} \ell^k C_\ell ,
\end{equation}
and
for all $\ell\in L$
for sufficiently small $\theta>0$,
\begin{equation}\label{eq_thm_C_ell_as_limit}
\EE{\rm in}_{\cM_n}\bigl[ B_{n^{-\theta}}(\ell) \bigr]
= n^{-j} C_\ell + o(n^{-j}) .
\end{equation}
\end{theorem}

\section{Article V: Proofs of the First Main Theorem}
\label{se_art_relativized}

Article~V begins by proving some easy results that we stated
in this article, namely
(1) the Ihara Determinantal formula in the context of graphs that can
have half-loops, and
(2) that all our standard models (Definition~\ref{de_models}) are
algebraic.
Our proof of the Ihara formula is a simple adaptation of
Bass' elegant proof \cite{bass_elegant} of the usual Ihara formula.
Our proof that our basic models are algebraic is based in
\cite{broder,friedman_random_graphs};
we also correct an error in \cite{friedman_alon} regarding the
cyclic (and cyclic-involution) models 
that we described in Subsections~\ref{su_realizable_letterings_regular}
and~\ref{se_our_basic_models_are_algebraic}.

After these easy results we gather the results of Articles~III
and~IV to easily prove the first main theorem,
Theorem~\ref{th_rel_Alon_regular_simple}, the relativized Alon
conjecture for any algebraic model over a 
$d$-regular graph, $B$.

Then we prove Theorem~\ref{th_rel_Alon_regular}.
Most of the work is to show the following result:
let $B$ be a $d$-regular graph for some $d\ge 3$, and let
$S$ be a graph with $\mu_1(S)>(d-1)^{1/2}$.
Let $\epsilon>0$ be real.
Then for sufficiently large $n$,
any covering map $G\to B$ of degree $n$
has a new adjacency eigenvalue that is larger than
$$
\mu_1(S) + \frac{d-1}{\mu_1(S)} - \epsilon
$$
provided that $G$ has a subgraph isomorphic to $S$.
Our proof uses the ``Curious Theorem'' (in Section~3.8) of 
\cite{friedman_alon}, as well as the methods of
Friedman-Tillich \cite{friedman_tillich_generalized}.

We finish Article~V with a section showing 
a number of examples regarding trace methods for new eigenvalues
of $G\in\cC_n(B)$ where $B$ is $d$-regular.
These examples concern trace methods proving high probability
bounds for the new spectral radius of $A_G$;
we show in these examples---where the high probability bound
is larger than the Alon bound, $2(d-1)^{1/2}$---one gets
improved bounds by applying
the trace method to $H_G$, and then converting the high probability bounds
on the new spectral radius of $H_G$ back to those on $A_G$
(using the Ihara determinantal formula).
We have no rigorous proof that applying
a trace method to $H_G$ gives better bounds than applying it to $A_G$
for random graphs $G$; this just happens to be the case in
all examples that we know.

\section{Article~VI: Sharp Exponents for Regular Ramanujan Base Graphs}
\label{se_art_sharp}

In this article we prove a number of results needed to prove
our second main theorem, Theorem~\ref{th_second_main_theorem},
for our basic models $\cC_n(B)$ when $B$ is Ramanujan.
This requires a number of independent results.

%

The main computation in this article regards the {\em magnification}
\cite{alon_eigenvalues}
of a graph, $G$, which refers to the smallest $\gamma>0$ such that
for all $U\subset V_G$ with $\#U\le (1/2)(\#V_G)$ we have
$$
\# \bigl( \Gamma(U) \setminus U \bigr) \ge \gamma (\# U)
$$
where $\Gamma(U)$ is the set of neighbouring vertices of $U$
(i.e., of distance one, i.e., joined to $U$ by some edge in $G$).

First, we show that all our basic models $\cC_n(B)$ have a sort of
{\em magnification} property that, roughly speaking, says
that for $G\in\cC_n(B)$ with very high probability---larger than
$1-O(n^{-s})$ for any fixed $s\in\naturals$---$G$
will only be a poor magnifier if it
has a small sized, disconnected component.
This is a fairly involved computation that closely resembles
those in \cite{friedman_alon}, Chapter~12.

Second, it is then an easy consequence that if $B$ is Ramanujan,
then $\tau_{\rm alg}=+\infty$ for our basic models and, more generally,
any algebraic model over such $B$
that satisfies a similar magnification property.
[Beyond this,
we conjecture that $\tau_{\rm alg}=+\infty$ for our basic models
regardless of $B$.]

These two results imply that the probability of having a non-Alon
eigenvalue when $B$ is $d$-regular and Ramanujan is bounded from 
below and above
as of order $n^{-\tau_{\rm tang}}$.
More precisely, having a non-Alon eigenvalue $2(d-1)^{1/2}+\epsilon$
or greater, for fixed $\epsilon>0$ is bounded above by
$n^{-\tau_{\rm tang}}$ times a constant depending on $\epsilon$,
and bounded below---for $G\in\cC_n(B)$ with $n$ sufficiently large---by
an absolute constant times $n^{-\tau_{\rm tang}}$, for $\epsilon$
sufficiently small.
This is always the case when $\tau_{\rm alg}=+\infty$, or, more
weakly, when $\tau_{\rm alg}\ge \tau_{\rm tang}+1$.

We finish Article~VI with
some remarks on the value of $\tau_{\rm tang}$ for our basic
models.
In particular, we prove
Corollary~\ref{co_tau_tang_bounds}
using the results of
\cite{friedman_alon}, Sections~6.3 and~6.4.

\ignore{

If $\{\cC_n(B)\}_{n\in N}$ is an algebraic model over a $d$-regular base
graph, $B$, we can determine---at least in principle---compute an
integer $\tau_{\rm tang}$, and then check if 
$\tau_{\rm alg}\le \tau_{\rm tang}$ and, if so, determine
$\tau_{\rm alg}$; this allows to determine within a factor of $n$ the
probability
$$
\Prob_{G\in\cC_n(B)}[ {\rm NonAlon}_B(G;\epsilon)>0 ].
$$

Unfortunately, $\tau_{\rm alg}$ involves computing various types of 
expansions over various $B$-types, which makes it difficult to compute.
However, for all of our basic models we can prove that
$\tau_{\rm alg}=+\infty$ by magnification results whose ideas go back to
\cite{friedman_random_graphs,friedman_alon}.
The magnification estimate needed is quite easy for 
{\red $d\ge 4$ ???} sufficiently
large.  The magnification estimate---really valid for arbitrary base
graphs of negative order (including, of course, $d$-regular graphs for
$d\ge 3$)---are a bit more subtle.
This technical computation is given in Article~VI.

Whenever $\tau_{\rm alg}=+\infty$, 
$$
\Prob_{G\in\cC_n(B)}[ {\rm NonAlon}_B(G;\epsilon)>0 ]
$$
is proportional to $n^{-\tau_{\rm tang}}$ for $\epsilon>0$ sufficiently
small.
}

\appendix

\section{Certified Trace and Indicator Function Approximation}
\label{se_append_cert_ind}

In this appendix we give an overview of contents of Article~III,
which is the proof
of Theorem~\ref{th_main_tech_result}.
This appendix involves some remarks on {\em certified traces}, which 
are new to this series of articles (and \cite{friedman_kohler}), which
significantly simplify the {\em selective traces} of 
\cite{friedman_alon}.
In addition, we make some remarks of indicator function
approximation from \cite{friedman_alon} that are used essentially
verbatim in Article~III.

This appendix also illustrates how Articles~II and~III work
together.

Our strategy to obtain expansion theorems for 
\begin{equation}\label{eq_art_III_goal}
f(k,n)\eqdef
\EE_{G\in\cC_n(B)}[ \II_{{\rm TangleFree}(\ge\nu,<r)}(G) \Trace(H^k_G) ],
\end{equation}
as well that in\cite{friedman_alon}, can be explained as follows.
Given a model, $\cC_n(B)$, 
and an $r\in\naturals$ and 
$\nu\in\reals$, say that we find a ``modified trace'' function,
${\rm ModifiedTrace}(G,k,\nu,r)$, with the following
properties:
\begin{enumerate}
\item
for any graph, $G$, and any $k,\nu,r$,
\begin{equation}\label{eq_modified_positivity}
0 \le {\rm ModifiedTrace}(G,k,\nu,r) \le \Trace(H^k_G) ,
\end{equation} 
\item
we have
\begin{equation}\label{eq_modified_trace_agrees}
\snbc_{<r}(G,k) = {\rm ModifiedTrace}(G,k,\nu,r) \quad
\mbox{if $G\in {\rm TangleFree}(\ge\nu,<r)$},
\end{equation} 
and
\item 
both the functions
\begin{equation}\label{eq_modified_trace}
\EE_{G\in\cC_n(B)}[ {\rm ModifiedTrace}(G,k,\nu,r) ]
\end{equation} 
and
\begin{equation}\label{eq_modified_trace_has_tangles}
\EE_{G\in\cC_n(B)}[ \II_{{\rm HasTangles}(\ge\nu,<r)}(G)
{\rm ModifiedTrace}(G,k,\nu,r) ]
\end{equation} 
have $(B,\nu)$-bounded asymptotic expansions to order $r$.
\end{enumerate}
Then, of course, one also has such an expansion for the difference of
\eqref{eq_modified_trace} and 
\eqref{eq_modified_trace_has_tangles},
which equals \eqref{eq_art_III_goal}.

In Theorem~2.2 of \cite{friedman_alon} one sees that, intuitively,
$\Trace(H^k_G)$ is so large on graphs with tangles, that
the function
$$
f(k,n)\eqdef
\EE_{G\in\cC_n(B)}[ 
\Trace(H^k_G) ],
$$
cannot have $(B,\nu)$-bounded coefficients, at least
for $r$ larger than roughly $d^{1/2}$.
In this series of articles, and in \cite{friedman_alon},
we define
${\rm ModifiedTrace}(G,k,\nu,r)$ as a function that counts some elements of
$\SNBC(G,k)$ and ignores others, with the property
that we only ignore
elements of $\SNBC(G,k)$ when
$G$ has a $(\ge\nu,<r)$-tangle;
such a count automatically satisfies
\eqref{eq_modified_positivity} and
\eqref{eq_modified_trace_agrees}.
The difficulty is to define ${\rm ModifiedTrace}(G,k,\nu,r)$ in a
way that guarantees that
\eqref{eq_modified_trace} and
\eqref{eq_modified_trace_has_tangles} have the desired expansions.

The modified trace functions in \cite{friedman_alon} were
called {\em selective traces}, and are a bit involved to describe
precisely.
Our approach in this series of articles is much simpler: namely
we take 
we discard {\em all} walks whose 
visited subgraph, $S$, satisfies $\mu_1(S)\ge\nu$ or $\ord(S)\ge r$.
In other words, we take
${\rm ModifiedTrace}(G,k,\nu,r)$ to be
$$
{\rm cert}_{<\nu,<r}(G,k) = \# {\rm CERT}_{<\nu,<r}(G,k),
$$
where ${\rm CERT}_{<\nu,<r}(G,k)$ is the subset of $\SNBC(G,k)$
of elements, $w$, whose
visited subgraph, $S=\ViSu(w)$ has
$\mu_1(S)<\nu$ and $\ord(S)<r$.
It is immediate that ${\rm cert}_{<\nu,<r}(G,k)$
satisfies
\eqref{eq_modified_positivity} and
\eqref{eq_modified_trace_agrees}.
Our strategy for proving that
\eqref{eq_modified_trace} and
\eqref{eq_modified_trace_has_tangles} have 
$(B,\nu)$-expansions to order $r$ begins with the following
preliminary observations.
\begin{enumerate}
\item
We are free to discard all
walks of order $r$ or greater when computing
expansions to order $r$
(see \eqref{eq_walks_large_order} or \eqref{eq_algebraic_order_bound}).
\item
For fixed $r$,
there are finitely many homotopy types of SNBC walks of order less than
$r$.
\item
Since each SNBC walk is of a unique homotopy type, $T^\og$,
the elements of ${\rm CERT}_{<\nu,<r}(G,k)$ is a sum over
all homotopy types of order less than $r$, $T^\og$, of those
walks, $w$, whose ordered visited subgraph, $S=\ViSu^\og(w)$
satisfies $\mu_1(S)<\nu$.
\item
It therefore suffices to fix a homotopy type, $T^\og$,
and prove a $(B,\nu)$-asymptotic expansion exists
for the $\cC_n(B)$-expected number of 
walks $w\in\SNBC(G,k)$, whose visited
subgraph is isomorphic to $\VLG(T^\og,\mec k)$ where $\mec k$ lies in
$$
{\rm Certified}(T^\og,<\nu)
\eqdef \{ \mec k\from E_T\to\naturals \ | \ \mu_1(\VLG(T,\mec k)) < \nu \} .
$$
\end{enumerate}

So we now focus on the above sets ${\rm Certified}(T^\og,<\nu)$.
We may view
the functions $E_T\to\naturals$ as a partially ordered set,
by the usual partial order on $\naturals^{E_T}$, i.e.,
$\mec k\le \mec k'$ if for all $e\in E_T$, $k(e)\le k'(e)$.
Then ${\rm Certified}(T^\og,<\nu)$ is an {\em upper set}, i.e., if
it contains $\mec k$ then it contains all $\mec k'$ with
$\mec k\le \mec k'$ (this follows from the theory of VLG's).
It then follows (by general remarks on upper sets in
$\naturals^{E_T}$)
that for fixed $T^\og$ and $\nu$ there are a 
finite number of 
elements $\bec\xi_1,\ldots,\bec\xi_m\in\naturals^{E_T}$ such that
${\rm Certified}(T^\og,<\nu)$ is the same as
$$
\{ \mec k\from E_T\to\naturals \ | \  
\mbox{$\mec k\ge\bec\xi_i$ for some $i$} \} .
$$
We call $\bec\xi_1,\ldots,\bec\xi_m$ a set of
{\em certificates for $T^\og,\nu$}.

Inclusion-exclusion then implies that it is enough to individually
prove asymptotic expansion theorems for expected counts of walks
$$
\{ w\in\SNBC(G,k)  \ | \ \ViSu^\og(w)\isom \VLG(T^\og,\mec k), \ 
\mec k\ge\bec\xi \} 
$$
for a fixed $\xi$ with $\VLG(T^\og,\bec\xi)$ having $\mu_1<\nu$.
This is precisely
\begin{equation}\label{eq_expected_snbc_cert}
f(k,n)=\EE_{G\in\cC_n(B)}[\snbc(T^\og,\ge\bec\xi;G,k)] \ .
\end{equation}
The proof that such a function has a $(B,\nu)$-asymptotic expansion
is a straightforward adaptation of the methods of 
\cite{friedman_random_graphs}; we do this in Article~II, where
we factor this proof into a number of general results and
simplify part of \cite{friedman_random_graphs}.
Such results then imply a $(B,\nu)$-asymptotic expansion to order
$r$ for 
\eqref{eq_modified_trace}, where
${\rm cert}_{<\nu,<r}(G,k)$ is used for
${\rm ModifiedTrace}(G,k,\nu,r)$.

To prove
\eqref{eq_modified_trace_has_tangles} we adapt the methods
of Chapter~9 of \cite{friedman_alon} to our situation.
The idea, roughly speaking, is to count pairs $(w,{\tilde S}_\Bg^\og)$
where $w\in \SNBC(T^\og,\ge\xi;G,k)$ and ${\tilde S}_\Bg^\og$
is a fixed $(\ge\nu,<r)$-tangle.
In Article~II we proves the existence of asymptotic expansions for
$$
f_{T^\og,\xi,{\tilde S}_\Bg^\og}(k,n)=\EE_{G\in\cC_n(B)}[
(\#[{\tilde S}_\Bg^\og]\cap G_\Bg) \,
\snbc(T^\og,\ge\bec\xi;G,k)] ;
$$
the proof uses the same methods in \cite{friedman_random_graphs}
as for \eqref{eq_expected_snbc_cert},
except that we need the notion of pairs and their homotopy
types, along the lines of Chapter~9 of \cite{friedman_alon}.

The techniques in Chapter~9 of \cite{friedman_alon} are used to
show that
$$
\II_{{\rm HasTangles}(\ge\nu,<r)}(G)
$$
can be approximated by linear combinations of the functions
$\#[{\tilde S}_\Bg^\og]\cap G_\Bg$ so well---in various expressions
involving $G\in\cC_n(B)$ expected values---that we can use
functions $f_{T^\og,\bec\xi,{\tilde S}_\Bg^\og}(k,n)$ to approximate
$$
f(k,n)=\EE_{G\in\cC_n(B)}[
\II_{{\rm HasTangles}(\ge\nu,<r)}(G) \,
\snbc(T^\og,\ge\bec\xi;G,k)] .
$$

\section{Definition Summary and Notes}
\label{se_def_summary}

In this section we list terminology appearing in the definitions in this paper.
We also make some notes, mostly to explain the definitions that are
not standard or to indicate some finer points regarding these definitions.

\medskip

In Sections~\ref{se_term1_basic} and~\ref{se_some_results} we
introduce terminology to state our first main result,
namely a proof of the
relativized Alon conjecture for some basic models of a random
covering map to a regular base graph.

We define a {\em graph} (i.e., {\em undirected graph}) to be a directed
graph with some additional structure, as in
\cite{friedman_geometric_aspects}.
This allows our graphs to have for multiple edges
and two types of self-loops: whole-loops (which contribute $2$ to the
degree of a vertex), and half-loops (which contribute $1$).
Half-loops are not entirely standard, but
are useful to describe random $d$-regular graphs on
$n$ vertices for a fixed odd integer $d$.
Half-loops are also needed to make the models in
\cite{friedman_alon} special cases of our second main result
when $B$ is a Ramanujan regular graph
(this result closes a gap in upper and lower bounds
in some cases in \cite{friedman_alon} for the models there).

Another advantage of viewing a graph as a directed graph with some
additional structure is that there are a number of concepts that
are much easier to define:
for example, the correct definition of a 
{\em covering map} of graphs---in the presence of self-loops and multiple
edges---is a bit tricky to define correctly.  However, the correct
notion of a covering map of directed graphs is easy to define,
and the correct notion on graphs is merely that the morphism
be a covering map of the
underlying directed graphs.
In fact, we define most of our graph theoretic concepts first on
directed graphs, and then easily ``extend'' these concepts to graphs.

We also remark that our trace methods involve the {\em oriented line
graph} of a graph, which itself is a directed graph, not a graph;
for this and related reasons, we need to keep directed graphs in mind
for much of this series of articles.

\medskip

\noindent
Definition~\ref{de_digraph}: {\em directed graph},
$B=(V_B,\Edir_B,t_B,h_B)$, {\em vertex set} $V_B$, 
{\em directed edge set}
$\Edir_B$, and {\em heads} and {\em tails maps}, $t_B,h_B$, from
$\Edir_B\to V_B$, {\em self-loops} ($e\in\Edir_B$ with $t_Be=h_Be$).

\noindent
Definition~\ref{de_graph}: {\em (undirected) graph},
$G=(V_G,\Edir_G,t_G,h_G,\iota_G)$,
its {\em underlying directed graph}, 
$(V_G,\Edir_G,t_G,h_G)$, the {\em edge involution}
$\iota_G$; 
{\em self-loops}, 
{\em whole-loops}, 
{\em half-loops} ($e\in\Edir_G$ with $\iota_G e=e$, i.e.,
directed edges paired with themselves),
{\em edge set} $E_G$ (the $\iota_G$ orbits in $\Edir_G$),
{\em orientation} of an edge or of the graph
(a choice of orbit representative(s)).

\noindent
Definition~\ref{de_coordinatized_digraph}:
{\em coordinatized cover} for digraphs, 
$V_G = V_B \times [n]$,
$\Edir_G = \Edir_B \times [n]$,
$t_G(e,i)=(t_B e,i)$,
$h_G(e,i) = \bigl(h_B e, \sigma(e) i\bigr)$,
where $\sigma$ is the associated permutation map
$\sigma\from\Edir_B\to\cS_n$.

\noindent
Definition~\ref{de_coordinatized_graph}: 
{\em coordinatized covers} for graphs,
i.e., $\iota_G$ must satisfy $\iota_G(e,i)=(\iota_B e, \sigma(e)i)$
(this implies $\sigma$ must satisfy $\sigma(\iota_B e)=\sigma(e)^{-1}$).

\noindent
Definition~\ref{de_model_conventions}: 
{\em models} (of covering maps to a fixed based graph), {\em edge-independence}
for a model, meaning that the $\sigma(e)\in\cS_n$ are independent
in $e$ ranging over any orientation of $B$.

\noindent
Definition~\ref{de_models}: our {\em basic models}: the
{\em permutation model},
the {\em permutation-involution model} (of even or odd degree), 
the {\em full cycle model} (or {\em cyclic model}), and the
{\em full cycle-involution model} 
(or {\em cyclic-involution}) (even or odd degree).

\noindent
Definition~\ref{de_in_out_degree}: 
{\em indegree and outdegree} (of a vertex)
and {\em adjacency 
matrix} for a digraph;
{\em degree} and {\em adjacency matrix}
for graphs;
{\em strongly regular} digraphs and {\em regular} graphs.

\noindent
Definition~\ref{de_bouquet}: {\em bouquets} (graphs on one vertex).

\noindent
Definition~\ref{de_digraph_morphisms}: 
{\em morphisms} of digraphs;
{\em vertex/edge fibres}; 
{\em covering} and {\em \'etale} morphisms; {\em degree} of a covering 
morphism.

\noindent
Definition~\ref{de_graph_morphisms}: 
{\em morphisms} of graphs;
{\em covering} and {\em \'etale} morphisms.

\noindent
Definition~\ref{de_graph_adjacency}:
notation: $\lambda_i(B)$ of a graph, $B$ (the eigenvalues
of $A_B$ in decreasing order).

\noindent
Definition~\ref{de_new_spec},\ \ref{de_new_spec_graphs}:
{\em mew/old functions}, {\em new/old spectrum}, 
{\em new spectral radius},
for digraphs and graphs.

\medskip

\noindent
Definition~\ref{de_nonAlon}: 
notation: ${\rm NonAlon}_B(G;\epsilon)$, the number of
{\em $\epsilon$-non-Alon eigenvalues}.

\noindent
Definition~\ref{de_Ramanujan}:
{\em Ramanujan} regular graphs.

\medskip

The definitions in Section~\ref{se_term2_walks_traces} allow us to
state the two main theorems in this series of articles.
These theorems involve two invariants,
$\tau_{\rm tang}$ and $\tau_{\rm alg}$, of a model of random
covering of a fixed base graph.

\medskip

\noindent
Definition~\ref{de_walks}:
{\em walk} in a digraph or graph, {\em closed} walk, 
{\em length} of a walk.

\noindent
Definition~\ref{de_nonback_oriented}:
{\em non-backtracking} walk, {\em strictly non-backtracking closed (SNBC)}
walk in a graph;
{\em oriented line graph}, {\em Hashimoto}
(or {\em non-backtracking}) matrix of a graph.

\noindent
Definition~\ref{de_SNBC}:
notation $\SNBC(G,k),\snbc(G,k)$ (the set and number of
strictly non-backtracking walks of length $k$ in $G$).

\noindent
Definition~\ref{de_visited_subgraph}:
the 
{\em visited subgraph}, $\ViSu_G(w)$, of a walk, $w$, in a 
digraph or graph, $G$.
Note that to determine $\ViSu_G(w)$, it is not generally enough
to know the sequence $w$, except under some assumption(s)
on $G$;
it is enough if we restrict ourselves to $G$ that 
are coordinatized covers.

\noindent
Definition~\ref{de_order}:
the {\em order} $\ord(G)=\#E_G - \#V_G$ of a graph, $G$; 
$\SNBC_r(G,k)$,
$\SNBC_{<r}(G,k)$,
$\SNBC_{\ge r}(G,k)$, and their cardinalities,
$\snbc_r(G,k)$,
$\snbc_{<r}(G,k)$,
$\snbc_{\ge r}(G,k)$.

\noindent
Definition~\ref{de_polyexponential_growth}:
(univariate) {\em polyexponential} functions, their {\em bases},
functions {\em of growth} $\rho$
for a $\rho>0$.

\noindent
Definition~\ref{de_Ramanujan_function},\ \ref{de_Ramanujan_expansion}:
{\em $(B,\nu)$-bounded} function, 
{\em $(B,\nu)$-Ramanujan} function, 
{\em $(B,\nu)$-bounded asymptotic expansions}, 
{\em $(B,\nu)$-Ramanujan asymptotic expansion}.

\noindent
Definition~\ref{de_tangle}:
{\em $\ge\nu$-tangle}, {\em $(\ge\nu,<r)$-tangle},
${\rm TangleFree}(\ge\nu,<r)$,
${\rm HasTangles}(\ge\nu,<r)$
(here the weak inequality $\ge\nu$ is crucial).



\medskip

\noindent
Definition~\ref{de_occurs}:
$S$ {\em occurs} in a model $\{\cC_n(B)\}_{n\in N}$.

\noindent
Definition~\ref{de_tau_tang}:
The {\em tangle power}, $\tau_{\rm tang}$, of a model.

\noindent
Definition~\ref{de_algebraic_power}:
The {\em algebraic power}, $\tau_{\rm alg}$, of a model.

\medskip

The definitions in Section~\ref{se_ordered_B_strong_alg} culminates
in our definition of a {\em strongly algebraic model} (of a family of 
random, degree $n$ covering maps, $\cC_n(B)$, to a fixed graph $B$);
if $B$ has no half-loops, then the {\em permutation model} is an example
of a strongly algebraic model.

\medskip

\noindent
Definition~\ref{de_B_graph},\ \ref{de_B_graph_morphisms}:
{\em $B$-graphs}, denoted $G_\Bg$, and {\em morphisms} of $B$-graphs.

\noindent
Definition~\ref{de_ordered_graph},\ \ref{de_first_encountered}:
{\em ordered} graphs, denoted $G^\og$, the {\em first-encountered ordering},
$\ViSu^\og(w)$ that $w$ endows upon its visited subgraph $\ViSu(w)$.

\noindent
Definition~\ref{de_fibre_counting}:
the {\em fibre counting} functions $\mec a\from E_B\to\integers_{\ge 0}$
(at times $\Edir_B\to\integers_{\ge 0}$),
and $\mec b\from V_B\to\integers_{\ge 0}$,
defined for a walk in, or a subgraph of, a coordinatized graph over $B$.

\noindent
Definition~\ref{de_ordered_B_graph}:
{\em ordered $B$-graph}, denoted $G_\Bg^\og$, and their {\em morphisms}.

\noindent
Definition~\ref{de_subgraph_counting}: the class
$[S_\Bg^\og]$, the set $[S_\Bg^\og]\cap G_\Bg$ for a $B$-graph, $G_\Bg$.
We similarly define
$[S_\Bg]$ and $[S_\Bg]\cap G_\Bg$ after this definition.

\noindent
Definition~\ref{de_pruned}:
a graph that is {\em pruned} (all vertices have degree at least two).

\noindent
Definition~\ref{de_strongly_algebraic}: 
a $B$-graph {\em occurs} in a model,
{\em strongly algebraic} models.



\medskip


The definitions in Section~\ref{se_new_homot} are needed to state
the results of Article~II (this terminology is not needed after
Article~III).

\medskip

\noindent
Definition~\ref{de_bead}:
{\em beads} in a graph,
(vertices of degree two not incident upon any half-loops),
{\em beaded paths} in a graph.

\noindent
Definition~\ref{de_proper_bead_set}:
a {\em proper bead set} (a set of beads that does not contain
all the vertices of any connected component of a graph, which
would have to be a cycle).

\noindent
Definition~\ref{de_suppression}:
the {\em reverse} walk in a graph, 
{\em bead suppression} and the resulting
{\em lengths} of its edges.

\noindent
Definition~\ref{de_homotopy_walk}:
{\em reduction}, {\em edge-lengths}, and
{\em homotopy type} of an non-backtracking walk, $w$.
The {\em reduction} is the graph obtained from $\ViSu^\og(w)$
by suppressing its beads and its first and last vertices; {\em homotopy type}
is the isomorphism class (as ordered graphs) of the {\em reduction}.
It is crucial that the {\em reduction}
and {\em homotopy type} are ordered graphs.
We need to suppress the first and last vertices in order 
to reconstruct the first encountered ordering
that $w$ induces on $\ViSu^\og(w)$ from the homotopy type of $w$.

\noindent
Definition~\ref{de_homotopy_SNBC_visited}:
{\em reduction}, {\em edge-lengths}, and
{\em homotopy type} of the (ordered) visited subgraph, $S^\og$,
of an SNBC walk.

\noindent
Definition~\ref{de_snbc_walks_of_homotopy_type}:
notation:
$\SNBC(T^\og; G,k)$ for
the set of walks in a graph of a given edge-lengths, and homotopy type;
also with, in addition, specified edge lengths,
$\SNBC(T^\og,\mec k; G,k)$,
or edge length lower bounds,
$\SNBC(T^\og,\ge\mec k; G,k)$; their cardinalities
$\snbc(T^\og; G,k)$,
$\snbc(T^\og,\mec k; G,k)$,
$\snbc(T^\og,\ge\mec k; G,k)$.
The function $\snbc(T^\og,\ge\mec k; G,k)$ is fundamental to our
treatment of certified traces, in Article~III.

\noindent
Definition~\ref{de_VLG}:
{\em variable-length graph} or {\em VLG}.

\medskip

The definitions in Section~\ref{se_new_algebraic} are
needed to define the notion of an {\em algebraic} model.

\medskip

\noindent
Definition~\ref{de_NBWALKS}:
The notation
${\rm NBWALKS}(B)$ to denote the non-backtracking walks of positive
length in a graph, $B$,
viewed as words over the alphabet $\Edir_B$ (i.e., we omit the vertices,
which are redundant in a walk of positive length in a known graph $B$);
also ${\rm NBWALKS}(B,e,e')$.
It is essential that we omit the vertices in order to get the
correct {\em eigenvalues} (defined below)
of the relevant regular languages. 

\noindent
Definition~\ref{de_wording}:
a {\em $B$-wording} of a graph, $T$,
that describes the $B$-graph structure $S_\Bg$ of a graph, $S$
where $T$ is a suppression of $S$;
the {\em edge lengths} of a wording.

\noindent
Definition~\ref{de_induced_wording}:
the {\em $B$-wording induced by} a $B$-graph, $S_\Bg$, on
any suppression of $S$.

\noindent
Definition~\ref{de_realization}:
the {\em realization}, $\VLG_\Bg(T,W)$, of a $B$-wording, $W$,
on a graph, $T$.

\noindent
Definition~\ref{de_B_type}: $B$-type, 
$T^{\rm type}=(T,\cR)$ for a graph $T$ and map $\cR$ from
$\Edir_T$ to regular languages over the alphabet $\Edir_B$;
a wording that {\em belongs} to a $B$-type;

\noindent
Definition~\ref{de_algebraic}: 
{\em algebraic} model.

\noindent
Definition~\ref{de_eigens_regular_language}:
the {\em eigenvalues} of a regular language or a $B$-type;
a {\em set of eigenvalues} of a model.

\medskip

\noindent
Definition~\ref{de_matrix_model}:
a {\em $(\Lambda_0,\Lambda_1)$ matrix model}, needed
in the statement of the Sidestepping Theorem.

\noindent
Section~\ref{se_art_sidestep} has some additional notation, such as
$\EE{\rm in}_\cM[R]$ and 
$\EE{\rm out}_\cM[R]$ in \eqref{eq_Ein_Eout}.

\providecommand{\bysame}{\leavevmode\hbox to3em{\hrulefill}\thinspace}
\providecommand{\MR}{\relax\ifhmode\unskip\space\fi MR }
\providecommand{\MRhref}[2]{%
  \href{http://www.ams.org/mathscinet-getitem?mr=#1}{#2}
}
\providecommand{\href}[2]{#2}


\begin{thebibliography}{ABG10}

\bibitem[ABG10]{a-b}
Louigi Addario-Berry and Simon Griffiths, \emph{The spectrum of random lifts},
  December 2010, available as \url{http://arxiv.org/abs/1012.4097}, 35 pages.

\bibitem[Alo86]{alon_eigenvalues}
N.~Alon, \emph{Eigenvalues and expanders}, Combinatorica \textbf{6} (1986),
  no.~2, 83--96, Theory of computing (Singer Island, Fla., 1984).
  \MR{88e:05077}

\bibitem[Bas92]{bass_elegant}
Hyman Bass, \emph{The {I}hara-{S}elberg zeta function of a tree lattice},
  Internat. J. Math. \textbf{3} (1992), no.~6, 717--797. \MR{1194071
  (94a:11072)}

\bibitem[BC19]{bordenave_collins2019}
Charles Bordenave and Beno\^{\i}t Collins, \emph{Eigenvalues of random lifts
  and polynomials of random permutation matrices}, Ann. of Math. (2)
  \textbf{190} (2019), no.~3, 811--875. \MR{4024563}

\bibitem[BL06]{bilu}
Yonatan Bilu and Nathan Linial, \emph{Lifts, discrepancy and nearly optimal
  spectral gap}, Combinatorica \textbf{26} (2006), no.~5, 495--519. \MR{2279667
  (2008a:05160)}

\bibitem[Bor15]{bordenave}
Charles Bordenave, \emph{A new proof of friedman's second eigenvalue theorem
  and its extension to random lifts}, 2015.

\bibitem[BS87]{broder}
Andrei Broder and Eli Shamir, \emph{On the second eigenvalue of random regular
  graphs}, Proceedings 28th Annual Symposium on Foundations of Computer
  Science, 1987, pp.~286--294.

\bibitem[FK14]{friedman_kohler}
Joel Friedman and David{-}Emmanuel Kohler, \emph{The relativized second
  eigenvalue conjecture of alon}, Available at
  \url{{http://arxiv.org/abs/1403.3462}}.

\bibitem[FKS89]{friedman_kahn_szemeredi}
J.~Friedman, J.~Kahn, and E.~Szemer{\'e}di, \emph{On the second eigenvalue of
  random regular graphs}, 21st Annual ACM Symposium on Theory of Computing,
  1989, pp.~587--598.

\bibitem[FM99]{frieze_molloy}
Alan~M. Frieze and Michael Molloy, \emph{Splitting an expander graph}, J. of
  Algorithms \textbf{33} (1999), 166--172.

\bibitem[Fri91]{friedman_random_graphs}
Joel Friedman, \emph{On the second eigenvalue and random walks in random
  $d$-regular graphs}, Combinatorica \textbf{11} (1991), no.~4, 331--362.
  \MR{93i:05115}

\bibitem[Fri93a]{friedman_relative_boolean}
\bysame, \emph{Relative expansion and an extremal degree two cover of the
  boolean cube}, preprint, unpublished, available at
  \url{http://www.math.ubc.ca/~jf/pubs/web_stuff/cover.html}.

\bibitem[Fri93b]{friedman_geometric_aspects}
\bysame, \emph{Some geometric aspects of graphs and their eigenfunctions}, Duke
  Math. J. \textbf{69} (1993), no.~3, 487--525. \MR{94b:05134}

\bibitem[Fri03]{friedman_relative}
\bysame, \emph{Relative expanders or weakly relatively {R}amanujan graphs},
  Duke Math. J. \textbf{118} (2003), no.~1, 19--35. \MR{1978881}

\bibitem[Fri08]{friedman_alon}
\bysame, \emph{A proof of {A}lon's second eigenvalue conjecture and related
  problems}, Mem. Amer. Math. Soc. \textbf{195} (2008), no.~910, viii+100.
  \MR{2437174}

\bibitem[FT05]{friedman_tillich_generalized}
Joel Friedman and Jean-Pierre Tillich, \emph{Generalized {A}lon--{B}oppana
  theorems and error-correcting codes}, SIAM J. Discret. Math. \textbf{19}
  (2005), 700--718.

\bibitem[Has90]{hashimoto1}
Ki-ichiro Hashimoto, \emph{On zeta and {$L$}-functions of finite graphs},
  Internat. J. Math. \textbf{1} (1990), no.~4, 381--396. \MR{1080105
  (92e:11089)}

\bibitem[Iha66]{ihara}
Yasutaka Ihara, \emph{On discrete subgroups of the two by two projective linear
  group over {${\mathfrak p}$}-adic fields}, J. Math. Soc. Japan \textbf{18}
  (1966), 219--235. \MR{0223463 (36 \#6511)}

\bibitem[LN98]{nagnibeda}
Alexander Lubotzky and Tatiana Nagnibeda, \emph{Not every uniform tree covers
  {R}amanujan graphs}, J. Combin. Theory Ser. B \textbf{74} (1998), no.~2,
  202--212. \MR{1654133 (2000g:05052)}

\bibitem[LP10]{linial_puder}
Nati Linial and Doron Puder, \emph{Word maps and spectra of random graph
  lifts}, Random Structures Algorithms \textbf{37} (2010), no.~1, 100--135.
  \MR{2674623}

\bibitem[LSV11]{lubetzky}
Eyal Lubetzky, Benny Sudakov, and Van Vu, \emph{Spectra of lifted {R}amanujan
  graphs}, Adv. Math. \textbf{227} (2011), no.~4, 1612--1645. \MR{2799807
  (2012f:05181)}

\bibitem[MSS15]{mssI}
Adam~W. Marcus, Daniel~A. Spielman, and Nikhil Srivastava, \emph{Interlacing
  families {I}: {B}ipartite {R}amanujan graphs of all degrees}, Ann. of Math.
  (2) \textbf{182} (2015), no.~1, 307--325. \MR{3374962}

\bibitem[PP15]{puder_p}
Doron Puder and Ori Parzanchevski, \emph{Measure preserving words are
  primitive}, J. Amer. Math. Soc. \textbf{28} (2015), no.~1, 63--97.
  \MR{3264763}

\bibitem[Pud15]{puder}
Doron Puder, \emph{Expansion of random graphs: new proofs, new results},
  Invent. Math. \textbf{201} (2015), no.~3, 845--908. \MR{3385636}

\bibitem[Ser77]{serre_arbres}
Jean-Pierre Serre, \emph{Arbres, amalgames, {${\rm SL}_{2}$}}, Soci\'{e}t\'{e}
  Math\'{e}matique de France, Paris, 1977, Avec un sommaire anglais,
  R\'{e}dig\'{e} avec la collaboration de Hyman Bass, Ast\'{e}risque, No. 46;
  article available at \url{http://www.numdam.org/issue/AST_1983__46__1_0.pdf}.
  \MR{0476875}

\bibitem[Sip96]{sipser}
Michael Sipser, \emph{Introduction to the theory of computation}, 1st ed.,
  International Thomson Publishing, 1996.

\bibitem[ST96]{st1}
H.~M. Stark and A.~A. Terras, \emph{Zeta functions of finite graphs and
  coverings}, Adv. Math. \textbf{121} (1996), no.~1, 124--165. \MR{1399606
  (98b:11094)}

\bibitem[ST00]{st2}
\bysame, \emph{Zeta functions of finite graphs and coverings. {II}}, Adv. Math.
  \textbf{154} (2000), no.~1, 132--195. \MR{1780097 (2002f:11123)}

\bibitem[Sta83]{stallings83}
John~R. Stallings, \emph{Topology of finite graphs}, Invent. Math. \textbf{71}
  (1983), no.~3, 551--565. \MR{695906 (85m:05037a)}

\bibitem[Sun86]{sunada1}
Toshikazu Sunada, \emph{{\(L\)}-functions in geometry and some applications},
  Lecture Notes in Mathematics \textbf{1201} (1986), 266--284.

\bibitem[SW49]{shannon}
Claude~E. Shannon and Warren Weaver, \emph{The mathematical theory of
  communication}, University of Illinois Press, 1971/1949.

\bibitem[Ter11]{terras_zeta}
Audrey Terras, \emph{Zeta functions of graphs}, Cambridge Studies in Advanced
  Mathematics, vol. 128, Cambridge University Press, Cambridge, 2011, A stroll
  through the garden. \MR{2768284 (2012d:05016)}

\bibitem[TS07]{st3}
A.~A. Terras and H.~M. Stark, \emph{Zeta functions of finite graphs and
  coverings. {III}}, Adv. Math. \textbf{208} (2007), no.~1, 467--489.
  \MR{2304325 (2009c:05103)}

\end{thebibliography}
\end{document}